\documentclass{article}

\usepackage[english]{babel}
\usepackage{natbib}
\usepackage[letterpaper,top=2cm,bottom=2cm,left=3cm,right=3cm,marginparwidth=1.75cm]{geometry}
\usepackage{amsmath, amsthm, amssymb}
\usepackage{xcolor}
\usepackage{graphicx}
\newif\ifrevision
\revisionfalse
\ifrevision
  \usepackage[hidelinks]{hyperref}
\else
  \usepackage[colorlinks=true, allcolors=blue]{hyperref}
\fi
\usepackage{subcaption}
\usepackage{mathtools}
\usepackage{wrapfig}
\newcommand{\blind}{1}
\usepackage{bibunits}
\defaultbibliographystyle{abbrvnat}
\usepackage[toc,page]{appendix}
\newcommand{\jasa}{0}

\usepackage{algorithm}
\usepackage{algpseudocode}

\usepackage{multirow}

\newcommand{\vecmat}[1]{#1}%{\boldsymbol{#1}}

\DeclareMathOperator*{\argmax}{argmax}
\DeclareMathOperator*{\tr}{tr}
\newtheorem{theorem}{Theorem}[section]
\newtheorem{proposition}[theorem]{Proposition}
\newtheorem{lemma}{Lemma}[section]
\newtheorem{corollary}{Corollary}[section]
\newtheorem*{notation}{Notation}
\newcommand{\rev}[1]{\ifrevision\textcolor{red}{#1}\else#1\fi}

\newcommand{\DIFaddbegin}{\ifrevision\color{red}\fi}
\newcommand{\DIFaddend}{\ifrevision\color{black}\fi}

\title{Scalable Krylov Subspace Methods for Generalized Mixed-Effects Models with Crossed Random Effects} 
\author{
  Pascal Kündig\footnotemark[1]  \footnotemark[3]\\
  \and
  Fabio Sigrist\footnotemark[2] \footnotemark[1]  \footnotemark[4]
}

\date{}

\begin{document}
\maketitle

\footnotetext[1]{Lucerne University of Applied Sciences and Arts}
\footnotetext[2]{Seminar for Statistics, ETH Zurich}
\footnotetext[3]{University of Basel}
\footnotetext[4]{Corresponding author: fabio.sigrist@stat.math.ethz.ch}

\begin{abstract}
Mixed-effects models are widely used to model data with complex grouping structures and high-cardinality categorical predictor variables. However, for high-dimensional crossed random effects, current standard computations relying on Cholesky decompositions can become prohibitively slow. In this work, we present Krylov subspace-based methods that address existing computational bottlenecks, and we analyze them both theoretically and empirically. In particular, we derive new results on the convergence and accuracy of the preconditioned stochastic Lanczos quadrature and conjugate gradient methods for mixed-effects models, and we develop scalable methods for calculating predictive variances. In experiments with simulated and real-world data, the proposed methods yield \rev{speedups of several orders of magnitude and are more computationally robust} than Cholesky-based computations\rev{, while maintaining essentially the same accuracy.}% as implemented in state-of-the-art packages such as \texttt{lme4} and \texttt{glmmTMB}. %Our methodology is available in the open-source C++ software library \texttt{GPBoost}, with accompanying high-level Python and R packages.
\end{abstract}
% !TeX root = JCGS.tex

\section{Introduction}
Mixed-effects models are widely used in various scientific disciplines to analyze grouped data \citep{pinheiro2000linear, mcculloch2004generalized}. High-dimensional crossed random effects occur frequently in practice, for instance, when modeling ratings in recommender systems with categorical grouping variables corresponding to customers and products \citep{gao2017efficient, ghosh2022backfitting, simchoni2023integrating} or, in general, when there are high-cardinality categorical predictor variables \citep{sigrist2023comparison}. However, current standard sparse Cholesky decomposition-based computations can become prohibitively slow for high-dimensional crossed random effects. For instance, the cost of computing Cholesky factors has up to cubic complexity in the dimension of the random effects \citep{pandolfi2024conjugate}.
Various methods have been proposed to overcome this computational bottleneck, including method of moments estimation \citep{gao2017efficient, gao2020estimation}, collapsed Gibbs samplers \citep{papaspiliopoulos2020scalable, papaspiliopoulos2023scalable}, backfitting algorithms \citep{ghosh2022scalable, ghosh2022backfitting}, and composite likelihoods \citep{bellio2023consistent}. \rev{Krylov subspace methods trace back to the conjugate gradient (CG) and Lanczos algorithms \citep{hestenes1952methods,lanczos1950iteration}; a broader treatment is given by \citet{saad2003iterative}.} \rev{In mixed-effects modeling,} the preconditioned conjugate gradient method with relatively simple diagonal or block-diagonal preconditioners has been applied to solve high-dimensional mixed model equations \citep{stranden1999solving, vandenplas2018deflated}. Recently, \citet{pandolfi2024conjugate} have analyzed the CG method with diagonal preconditioning for Bayesian generalized linear mixed models (GLMMs). For Gaussian likelihoods and the special case where all random effects have the same variance, \citet{border2019stochastic} and \citet{cheng2023slemm} have used stochastic Lanczos quadrature (SLQ) \citep{ubaru2017fast} without preconditioning in combination with a shift-invariance property to calculate log-determinants in log-marginal likelihoods. However, without preconditioning, SLQ-approximated log-marginal likelihoods can have high variances and can thus be inaccurate; see our theoretical and empirical analyses below. To our knowledge, an accurate and scalable method for calculating log-determinants tailored to mixed-effects models, together with a theoretical understanding of its convergence properties, is thus currently missing. Furthermore, computing predictive variances at large numbers of prediction points remains a computational bottleneck for high-dimensional crossed random-effects models.%, stranden2001parallel, tsuruta2001use, taskinen2017single, garrick2019alternative, % xu2023gaussian

% , and variational inference \citep{goplerud2024partially}

% Overall, there are two major unsolved computational challenges: (i) efficiently computing log-determinants in log-likelihoods and their gradients, and (ii) solving high-dimensional sparse linear systems with many right-hand sides for prediction. 

%, when measuring public opinion and voter turnout with random effects modeling effects of a large number of geographic units and other categorical variables \citep{warshaw2012should, ghitza2013deep, broockman2018bias}

In this article, \rev{we consider generalized mixed-effects models with grouped random intercepts or coefficients and Gaussian and non-Gaussian likelihoods. We introduce Krylov subspace methods for computationally efficient parameter estimation and prediction in such mixed-effects models. Krylov subspace methods allow for fast computations as they rely on matrix-vector multiplications with sparse matrices, which can be trivially parallelized. Specifically, for a Gaussian likelihood using the Woodbury identity and matrix-determinant lemma and for non-Gaussian likelihoods with a Laplace approximation, computations are performed in the random-effect space with the sparse matrix \(\Sigma^{-1}+Z^\top WZ\) (see Section \ref{sec_model}).} Our main contributions are: (i) preconditioned SLQ and stochastic trace estimation (STE) methods for likelihood and gradient calculations; (ii) convergence theory explaining the roles of symmetric successive over-relaxation (SSOR) and diagonal preconditioning; (iii) scalable stochastic estimators for predictive variances; and (iv) an open-source implementation together with large-scale empirical validation. In experiments, we find that Krylov subspace-based methods \rev{achieve speedups by factors of up to about $30{,}000$ and are more computationally robust than Cholesky-based calculations, while yielding essentially identical parameter estimates and predictive accuracy. For instance, we are able to estimate a model with two crossed random effects on the large-scale MovieLens\_32m data set in less than four minutes, whereas the widely used \texttt{lme4} package requires more than 38 hours and Cholesky-based computations fail in other packages (\texttt{glmmTMB}, \texttt{MixedModels.jl}, \texttt{R-INLA}, \texttt{GPBoost}) due to memory limitations.}

\section{Preliminaries on Generalized Mixed-Effects Models}\label{sec_model}
\subsection{\rev{Notation and model setup}}
We assume that the response variable $y=(y_1, \dots, y_n)^T \in \mathbb{R}^n$ follows a distribution with a density $p(y|\mu,\xi)=\prod_{i=1}^n p(y_i|\mu_i,\xi)$, where $\mu\in\mathbb{R}^n$ are, potentially link function-transformed, parameters, and $\xi \in \Xi \subset \mathbb{R}^r$ are auxiliary parameters. For instance, $\mu_i$ and $\xi$ can be the mean and variance of a Gaussian distribution, the log-mean and the shape parameter of a gamma likelihood, or $\mu_i$ can be the logit-transformed success probability of a Bernoulli distribution for which there is no additional parameter $\xi$. The main parameters of interest $\mu$ are modeled as the sum of fixed \rev{effects,} $F(X)$, and random effects, $Zb$:
\begin{equation*}
    \mu = F(X) + Zb, ~~~ b \sim N(0,\Sigma),
\end{equation*}
where $X\in\mathbb{R}^{n\times p}$ contains predictor variables. In generalized linear mixed-effects models (GLMMs), the fixed effects function is linear, $F(X) = X\beta$, but $F(\cdot)$ can also be modeled using machine learning methods such as random forests \citep{hajjem2014mixed}, tree-boosting \citep{sigrist2022gaussian, sigrist2022latent}, or neural networks \citep{simchoni2023integrating}. For notational simplicity, we will denote the parameters of this function $F(\cdot)$ by $\beta$, regardless of whether $F(\cdot)$ is linear or not. The latent Gaussian variables $b = (b_1^T,b_2^T,\dots,b_K^T)^T\in\mathbb{R}^m$ consist of $K$ grouped random effect components $b_k\in\mathbb{R}^{m_k}$, $k \in \{1,\dots,K\}$, $m = \sum_{k=1}^{K} m_k$, and the covariance matrix $\Sigma\in\mathbb{R}^{m\times m}$ is diagonal and depends on a set of variance parameters $\theta = \{\sigma_1^2, \sigma_2^2, \dots, \sigma_K^2\}$: 
\begin{equation*}
\Sigma = \text{diag}(\underbrace{\sigma_1^2,\dots,\sigma_1^2}_{m_1-\text{times}},
                     \underbrace{\sigma_2^2,\dots,\sigma_2^2}_{m_2-\text{times}},
                     \dots,
                     \underbrace{\sigma_K^2,\dots,\sigma_K^2}_{m_K-\text{times}}) \in \mathbb{R}^{m\times m}.
\end{equation*}
In addition, $Z = (Z_1, Z_2, \dots, Z_K) \in \mathbb{R}^{n \times m}$, $Z_k \in \mathbb{R}^{n \times m_k}$, is usually a binary incidence matrix that maps the random effects $b$ to the corresponding observations, but it can also contain predictor variables when modeling random coefficients. We denote such a model as a generalized mixed-effects model. Note that the diagonal assumption for $\Sigma$ is not crucial and could be relaxed, e.g., to account for prior correlation among random slopes and intercepts.

\subsection{Parameter estimation}

% For non-Gaussian likelihoods, marginal likelihoods are approximated using the Laplace approximation since it is computationally efficient and converges asymptotically to the correct quantity. Supporting this argument, we show in Figure \ref{fig:LaplaceBias} the estimated variance parameter of one random effect in a model with two crossed random effects obtained with a Laplace approximation for varying numbers of repeated observations $d$ per random effect on simulated data with a Bernoulli likelihood and the setting described in Section \ref{exp_setting} with $m=4'000$. Estimation is done using the Krylov subspace methods introduced in this paper and repeated on $100$ simulated data sets for every $d$. Figure \ref{fig:LaplaceBias} shows that the variance parameter estimator has a small downward bias which vanishes as $d$ grows.
% \begin{figure}[ht!]
%     \centering
%     \includegraphics[width=0.8\linewidth]{plots/bias_analysis.png}
%     \caption{Estimated variance parameter $\sigma^2_1$ obtained with a Laplace approximation vs. different numbers of repeated observations $d$ per random effect for binary data. The red rhombi represent means and the whiskers are $\pm 2 \times$ standard errors. The dashed line indicates the true parameter $\sigma^2_1=0.25$.}
%     \label{fig:LaplaceBias}
% \end{figure}

For parameter estimation, the marginal likelihood $p(y|\beta,\theta,\xi)=\int p(y|\mu,\xi)p(b|\theta)\,\mathrm{d}b$ is typically maximized. For non-Gaussian likelihoods, there is no analytic expression for this marginal likelihood and an approximation has to be used. We use the Laplace approximation since it is computationally efficient and widely used in popular GLMM software libraries such as \texttt{lme4} \citep{bates2015fitting} and \texttt{glmmTMB} \citep{glmmtmb_cite}. \rev{Laplace approximations are also at the core of the integrated nested Laplace approximation (INLA) approach of \citet{rue2009approximate} for approximate Bayesian inference in latent Gaussian models, as implemented in the \texttt{R-INLA} package, which we additionally include in our experiments.} For the Laplace approximation, we assume that the likelihood $p(y|\mu,\xi)$ is log-concave and twice continuously differentiable in $\mu$. A Laplace approximation to the negative log-marginal likelihood is given by (ignoring constant terms)
\begin{equation}\label{mll_laplace}
    L^{LA}(y|\beta,\theta,\xi)=-\log p(y|\mu^*,\xi) + \frac{1}{2} {b^*}^T\Sigma^{-1} b^* + \frac{1}{2}\log\det\left(\Sigma\right) + \frac{1}{2}\log\det\left(\Sigma^{-1} + Z^T W Z\right),
\end{equation}
where $\mu^* = F(X)+Zb^*$, $b^*=\argmax_{b}\log p(y|\mu,\xi) - \frac{1}{2} b^T\Sigma^{-1} b$ is the mode of $p(y|\mu,\xi)p(b|\theta)$, and $W\in\mathbb{R}^{n\times n}$ is diagonal with $W_{ii}=-\frac{\partial^2 \log p(y_i| \mu_i,\xi)}{\partial \mu_i^2}\Big|_{\mu=\mu^*}$. \rev{Thus, $W_{ii}$ is the local negative curvature of observation $i$'s log-likelihood contribution with respect to $\mu_i$, evaluated at the mode.}
The mode $b^*$ is typically found with Newton's method, and one iteration is given by
\begin{equation}\label{newton}
b^{*t+1}
= b^{*t} + \left(Z^T W Z + \Sigma^{-1}\right)^{-1}\left(Z^T\frac{\partial\log p(y|\mu^{*t},\xi)}{\partial \mu}-\Sigma^{-1}b^{*t}\right), ~~ t = 0, 1, \dots
\end{equation}
If a first- or second-order optimization method is used for minimizing $L^{LA}(y|\beta,\theta,\xi)$, gradients with respect to $\theta$, $F=F(X)\in\mathbb{R}^n$, and $\xi$ are needed. These gradients can be found, e.g., in \cite{sigrist2022latent}. Note that gradients with respect to $F\in\mathbb{R}^n$ are used, for instance, for GLMMs since $\frac{\partial L^{LA}(y|\beta,\theta,\xi)}{\partial \beta}=X^T\frac{\partial L^{LA}(y|\beta,\theta,\xi)}{\partial F}$ by the chain rule when $F(X)=X\beta$.
% Note that $\log\det\left(\Sigma Z^T W Z + I_n\right) = \log\det\left(\Sigma\right) +  \log\det\left(\Sigma^{-1} + Z^T W Z\right)$.

% \subsubsection{Gaussian likelihoods}
% For Gaussian likelihoods, we can write
% \if0\jasa{$$ y = F(X) + Zb + \epsilon, ~~~ \epsilon \sim N(0,W^{-1}),~~ W = \frac{1}{\sigma^2} I_n,$$}\else{$ y = F(X) + Zb + \epsilon, ~~~ \epsilon \sim N(0,W^{-1}),~~ W = \frac{1}{\sigma^2} I_n,$}\fi
% where $\sigma^2$ is the error variance and $I_n \in \mathbb{R}^{n \times n}$ an identity matrix. In this case, 

For Gaussian likelihoods, the Laplace approximation is exact, and we have
\begin{equation}\label{lossGaussian}
    L(y|\beta,\theta,\xi) = \frac{n}{2}\log(2\pi) + \frac{1}{2} \log\det(\Psi) + \frac{1}{2}(y-F(X))^T\Psi^{-1}(y-F(X)), ~~ \Psi = Z \Sigma Z^T + W^{-1}.
\end{equation}
Since usually $m<n$ and for sparsity reasons, one uses the Woodbury identity 
\begin{equation}\label{woodbury}
    \Psi^{-1} =  W-WZ(\Sigma^{-1} + Z^TWZ)^{-1}Z^TW
\end{equation}
for solving linear systems and the matrix determinant lemma to calculate 
\begin{equation}\label{determinant_lemma}
\log\det(\Psi) = \log\det(\Sigma^{-1} + Z^T W Z) + \log\det(\Sigma) + \log\det(W^{-1}).
\end{equation}
Gradients, e.g., with respect to the variance parameters are given by
\begin{equation}
\begin{split}
    \frac{\partial L(y|\beta,\theta,\xi)}{\partial\theta_k}
    ={}& \frac{1}{2} \tr\left((\Sigma^{-1} + Z^T W Z)^{-1}\frac{\partial\Sigma^{-1}}{\partial\theta_k} + \Sigma^{-1}\frac{\partial\Sigma}{\partial\theta_k}\right)\\
    &- \frac{1}{2}(y-F(X))^T\Psi^{-1}\frac{\partial \Psi}{\partial\theta_k}\Psi^{-1}(y-F(X)).
\end{split}
\end{equation}
The Fisher information $\mathcal{I} \in \mathbb{R}^{K \times K}$ for $\theta$, which can be used for Fisher scoring and for obtaining asymptotic confidence sets, is given by
\begin{equation}\label{FisherI}
    (\mathcal{I})_{kl} = \frac{1}{2}\tr\left(\Psi^{-1}\frac{\partial\Psi}{\partial\theta_k}\Psi^{-1}\frac{\partial\Psi}{\partial\theta_l}\right), \; 1\leq k,l\leq K.
\end{equation}

\subsection{Prediction}\label{sec_prediction}
Predictions for latent variables $\mu_p \in\mathbb{R}^{n_p}$ and response variables $y_p \in\mathbb{R}^{n_p}$ are made using the posterior predictive distributions $p(\mu_p|y)$ and $p(y_p|y)$, respectively. Note that predictions can be made for random effects $b$ corresponding to levels present in the training data and new random effects $b_p\in\mathbb{R}^{m_p}$ that did not occur in the observed data. We denote by $Z_{po} \in \mathbb{R}^{n_p \times m}$ the matrix that relates the random effects $b$ from the training data to the prediction points, $Z_{pp} \in \mathbb{R}^{n_p \times m_p}$ the matrix that relates the new random effects $b_p$ to the prediction points, and $n_p$ and $m_p$ the numbers of prediction points and new random effects, respectively. Using this notation, we have
\begin{equation}\label{pred_mean}
    \mu_p = F(X_p) + Z_{po}b + Z_{pp}b_p \in\mathbb{R}^{n_p},
\end{equation}
where $X_p\in\mathbb{R}^{n_p\times p}$ is the predictor variable matrix of the prediction points, and 
\begin{equation}\label{pred_Laplace}
    \begin{split}
        \begin{pmatrix}  b\\  \mu_p \end{pmatrix}
        &= \begin{pmatrix}  0\\  F(X_p) \end{pmatrix}  +
        \begin{pmatrix}
            (I_m, 0_{m\times m_p})\\
            (Z_{po}, Z_{pp})
        \end{pmatrix}
        \begin{pmatrix}  b\\  b_p \end{pmatrix}\\
        &\sim \mathcal{N}\left( \begin{pmatrix}  0\\  F(X_p) \end{pmatrix},
        \begin{pmatrix}
            \Sigma       & \Sigma Z_{po}^T\\
            Z_{po}\Sigma & Z_{po} \Sigma Z_{po}^T + Z_{pp} \Sigma_p Z_{pp}^T
        \end{pmatrix}
        \right),
    \end{split}
\end{equation}
where $(I_{m},0_{m\times m_p})\in \mathbb{R}^{m\times (m + m_p)}$, $I_{m}\in \mathbb{R}^{m\times m}$ is an identity matrix, $0_{m\times m_p}\in\mathbb{R}^{m\times m_p}$ is a matrix of zeros, and $\Sigma_{p} = \text{Cov}(b_p)$. %While predictive distributions are less commonly used in classical statistical applications of GLMMs, they are important, for example, in mixed-effects machine learning models \citep{sigrist2022gaussian, simchoni2023integrating}, where probabilistic predictions are required.

For the latent variable $\mu_p$, the posterior predictive distribution $p(\mu_p|y)$ is given by 
$$
p(\mu_p|y) = \int p(\mu_p|b)p(b|y)\,\mathrm{d}b \approx \mathcal{N}\left(\omega_p,\Omega_p\right),
$$
where
% \begin{subequations}
% \label{postpred_Laplace}
\begin{align}
    \omega_p &= F(X_p) + Z_{po} b^*,\label{post_mean}\\
    \Omega_p &= Z_{po} \Sigma Z_{po}^T + Z_{pp} \Sigma_p Z_{pp}^T - Z_{po} \Sigma Z^T \Psi^{-1} Z\Sigma Z_{po}^T \label{postpred_Laplace_cov1}\\
             &= Z_{pp} \Sigma_p Z_{pp}^T + Z_{po} \left(\Sigma^{-1} + Z^TWZ\right)^{-1} Z_{po}^T. \label{postpred_Laplace_cov2}
\end{align}
% \end{subequations}
This can be derived using the fact that a Laplace approximation for $p(y|\beta,\theta,\xi)$ is equivalent to the approximation $p(b|y) \approx \mathcal{N}\left(b^*,\left(\Sigma^{-1} + Z^T W Z\right)^{-1}\right)$ since $ p(b|y) = p(y|b)p(b) / p(y)$, and by applying standard results for conditional distributions of multivariate Gaussian distributions. For predicting the observable response variables $y_p$, $y_p|\mu_p\sim p(y_p|\mu_p)$, an additional integral must be calculated:
\begin{equation}\label{pred_resp}
    p(y_p|y)=\int p(y_p|\mu_p)p(\mu_p|y)\,\mathrm{d}\mu_p.
\end{equation}
This is analytically intractable for most likelihoods but can be approximated using numerical integration or by simulating from $p(\mu_p|y)\approx \mathcal{N}\left(\omega_p,\Omega_p\right)$. Note that, for notational simplicity, we omit the dependence of $p(\mu_p|y)$, $p(b|y)$, $p(\mu_p|b)$, $p(y_p|y)$, and $p(y_p|\mu_p)$ on $\beta$, $\theta$, and $\xi$.

For a Gaussian likelihood, the above predictive distributions are exact and available in closed form. The predictive mean in \eqref{post_mean} can be equivalently calculated as 
\if0\jasa{$$\omega_p = F(X_p) + Z_{po} \Sigma Z^T \Psi^{-1} (y-F(X)),$$}\else{$\omega_p = F(X_p) + Z_{po} \Sigma Z^T \Psi^{-1} (y-F(X))$, }\fi
and $p(y_p|y)$ is obtained from $p(\mu_p|y)$ by adding the error variance to the diagonal of $\Omega_p$. The predictive mean $\omega_p$ is often called the best linear unbiased predictor (BLUP) in the mixed-effects literature.

\section{Krylov Subspace Methods for Mixed-Effects Models}\label{krylov_methods}
\subsection{\rev{Core computational operations}}
Parameter estimation and prediction for models with high-dimensional crossed random effects require the following time-consuming operations all involving operations with the sparse matrix $\Sigma^{-1} + Z^TWZ \in\mathbb{R}^{m\times m}$. First, calculating linear solves $(\Sigma^{-1} + Z^T W Z) u = v$, $v \in \mathbb{R}^m$, in four tasks: (i) in quadratic forms of log-marginal likelihoods in \eqref{lossGaussian} after applying the Woodbury identity given in \eqref{woodbury}, (ii) in Newton's method for finding the mode, see \eqref{newton}, (iii) for implicit derivatives of the log-marginal likelihood, see \cite[][Proposition 2.1]{sigrist2022latent}, and (iv) for predictive variances $\text{diag}(\Omega_p)$ given in \eqref{postpred_Laplace_cov1}. The latter is particularly challenging as the number of prediction points $n_p$ is typically large, and $n_p$ linear systems need to be solved to calculate $\text{diag}(\Omega_p)$. Note that predictive variances are required not only for predictive distributions, but also for predictive means of non-Gaussian response variables, see \eqref{pred_resp}. Second, calculating log-determinants $\log\det(\Sigma^{-1} + Z^T W Z)$ in (approximate) log-marginal likelihoods $L(y|\beta,\theta,\xi)$ given in \eqref{mll_laplace} and \eqref{determinant_lemma} for non-Gaussian and Gaussian likelihoods, respectively. And third, calculating trace terms such as $\tr((\Sigma^{-1} + Z^T W Z)^{-1}\frac{\partial(\Sigma^{-1} + Z^T W Z)}{\partial\theta_k})$ for the derivatives of log-determinants and $\tr(\Psi^{-1}\frac{\partial\Psi}{\partial\theta_k}\Psi^{-1}\frac{\partial\Psi}{\partial\theta_l})$ for entries of the Fisher information. Traditionally, these operations are performed using a Cholesky decomposition of $\Sigma^{-1} + Z^T W Z$. In the following, we show how these operations can be done using Krylov subspace methods. 

For linear solves with the matrix $\Sigma^{-1} + Z^T W Z$, we use the preconditioned conjugate gradient (CG) method, which solves a linear system $(\Sigma^{-1} + Z^T W Z) u = v$ by iteratively doing matrix-vector multiplications with $\Sigma^{-1} + Z^T W Z$. This can be done fast since $\Sigma^{-1} + Z^T W Z$ is sparse, and convergence is often achieved with $l \ll m$ iterations. Thus, linear solves can be calculated in $O(l(m+n))$ time complexity. For completeness, the preconditioned CG algorithm is included in Appendix~\ref{appendix:CGalgo}. %nnz(\Sigma^{-1} + Z^T W Z) <= K^2 n <= K m

%Several techniques have been proposed in the literature to calculate log-determinants of large, symmetric positive definite matrices. \citet{dong2017scalable} find that stochastic Lanczos quadrature (SLQ) \citep{ubaru2017fast} achieves the highest accuracy and fastest runtime compared to other methods. In this article

To calculate log-determinants, we use the preconditioned stochastic Lanczos quadrature (SLQ) method. Specifically, we first note that
\begin{equation}\label{logdetsplit}
    \log\det(\Sigma^{-1} + Z^T W Z) = \log\det(P) + \log\det(P^{-\frac{1}{2}}(\Sigma^{-1} + Z^T W Z)P^{-\frac{T}{2}}),
\end{equation}
where $P$ is a symmetric positive definite preconditioner matrix; see Section \ref{sec_prec} for more information on preconditioners. The last term is approximated with the SLQ method \citep{ubaru2017fast} as follows:
\begin{equation}\label{log_det_slq}
    \log\det(P^{-\frac{1}{2}}(\Sigma^{-1} + Z^T W Z)P^{-\frac{T}{2}})
    % \approx \frac{1}{t} \sum_{i=1}^t\|P^{-\frac{1}{2}}z_i\|_2^2 e_1^T \log(\tilde{T}_i) e_1
    \approx \frac{m}{t} \sum_{i=1}^t e_1^T \log(\tilde{T}_i) e_1,
\end{equation}
\rev{Here, $e_1=(1,0,\dots,0)^T$, and $\tilde{T}_i\in\mathbb{R}^{l\times l}$ is the tridiagonal matrix obtained after $l$ Lanczos steps. More precisely, $\tilde{Q}_i\tilde{T}_i\tilde{Q}_i^T$ is a partial Lanczos decomposition of $P^{-1/2}(\Sigma^{-1}+Z^TWZ)P^{-T/2}$, where $\tilde{Q}_i\in\mathbb{R}^{m\times l}$ has orthonormal columns. The initial vector is $P^{-1/2}z_i/\|P^{-1/2}z_i\|_2$, and $z_1,\dots,z_t\in\mathbb{R}^m$ are independent random vectors satisfying $\mathbb{E}[z_i]=0$ and $\mathbb{E}[z_i z_i^T]=P$; we use $z_i\sim\mathcal{N}(0,P)$. Equation~\eqref{log_det_slq} follows by applying stochastic trace estimation and a Lanczos decomposition to $A=P^{-1/2}(\Sigma^{-1}+Z^TWZ)P^{-T/2}$; see \citet{ubaru2017fast} and \citet[Appendix A.2]{kundig2024iterative} for a derivation. The preconditioned decomposition reduces variance because a good preconditioner makes $A$ close to the identity and hence makes $\log(A)$ small.}

% $\log\det(\vecmat{A}) = \tr\log(\vecmat{A}) \approx \tr(\tilde{\vecmat{Q}}\log(\tilde{\vecmat{T}}) \tilde{\vecmat{Q}}^{T}) \approx \frac{1}{\ell} \sum_{i=1}^\ell {z}_i^{\top} \vecmat{\Tilde{Q}}_i\log(\vecmat{\Tilde{T}}_i)\vecmat{\Tilde{Q}}_i^{\top} {z}_i \approx \frac{n}{\ell} \sum_{i=1}^\ell {e}_1^{\top}\log(\vecmat{\Tilde{T}}_i) {e}_1$, where $\vecmat{A}=P^{-\frac{1}{2}}(\Sigma^{-1} + Z^T W Z)P^{-\frac{T}{2}}$. 

As in \citet{gardner2018gpytorch}, we use a technique from \citet{saad2003iterative} to calculate the partial Lanczos tridiagonal matrices $\tilde{T}_1,\dots,\tilde{T}_t$ from the coefficients of the preconditioned CG algorithm when solving $(\Sigma^{-1} + Z^T W Z)^{-1}z_1,\dots, (\Sigma^{-1} + Z^T W Z)^{-1}z_t$ $t$ times; see Appendix~\ref{appendix:CGalgo}. In doing so, we avoid running the Lanczos algorithm, which brings multiple advantages: Numerical instabilities of the Lanczos algorithm due to loss of orthogonality \citep{saad2003iterative} are not an issue, storing $\tilde{Q}_i$ is not necessary, and the linear solves $(\Sigma^{-1} + Z^T W Z)^{-1}z_i$ can be reused in stochastic trace estimation (STE) for calculating derivatives of the log-determinant, e.g., as follows:
\begin{equation*}
    \frac{\partial\log\det(\Sigma^{-1} + Z^T W Z)}{\partial\theta_k} 
    \approx \frac{1}{t}\sum_{i=1}^t ((\Sigma^{-1} + Z^T W Z)^{-1}z_i)^T\frac{\partial(\Sigma^{-1} + Z^T W Z)}{\partial\theta_k}P^{-1} z_i.
\end{equation*}
% The derivation of this approximation is also provided in Appendix A.2 in \cite{kundig2024iterative}. 
Gradients can thus be calculated with minimal computational overhead once the likelihood is calculated. In Appendix \ref{PAppendix}, we show in detail how to calculate derivatives of log-determinants with respect to $\theta$, $F$, $b^*$, and $\xi$, using STE and a form of variance reduction with control variates based on the SSOR preconditioner introduced in the next section. We choose the Lanczos rank $l$ adaptively by running the preconditioned CG algorithm until it has converged. The number of random vectors $t$ controls the stochastic trace estimation variance and is studied empirically in Section \ref{prec_compare}.% then recovering the corresponding Lanczos tridiagonal matrix from the CG coefficients.

Trace terms of the Fisher information given in \eqref{FisherI} can also be approximated with STE as follows:
\begin{equation}\label{STE_FisherI}
\begin{split}
    \tr\left(\Psi^{-1}\frac{\partial\Psi}{\partial\theta_k}\Psi^{-1}\frac{\partial\Psi}{\partial\theta_l}\right)
    =&\tr     \left(\left(W-WZ(\Sigma^{-1} + Z^TWZ)^{-1}Z^TW\right)Z_kZ_k^T\right.\\
     &\left.~~~~    \left(W-WZ(\Sigma^{-1} + Z^TWZ)^{-1}Z^TW\right)Z_lZ_l^T\right)\\
    \approx& \frac{1}{t} \sum_{i=1}^t \left(\left(Z_kZ_k^TW-Z_kZ_k^TWZ(\Sigma^{-1} + Z^TWZ)^{-1}Z^TW\right)z_i\right)^T\\
     &  ~~~~~~~~~                           \left(WZ_lZ_l^T-WZ(\Sigma^{-1} + Z^TWZ)^{-1}Z^TWZ_lZ_l^T\right)z_i,
\end{split}    
\end{equation}
where we use the Woodbury identity, $\frac{\partial\Psi}{\partial\theta_k} = Z_kZ_k^T$, and $z_i\sim \mathcal{N}(0,I_n)$. Linear solves $(\Sigma^{-1} + Z^TWZ)^{-1}Z^TWz_i$ and $(\Sigma^{-1} + Z^TWZ)^{-1}Z^TWZ_lZ_l^Tz_i$ in \eqref{STE_FisherI} can be computed efficiently using the CG method. In experiments, we have found that an alternative STE for the Fisher information given in Appendix \ref{STE_FisherI_alternative} leads to less accurate estimates than the STE in \eqref{STE_FisherI} (results not shown).

\subsection{Preconditioners}\label{sec_prec}
Preconditioners reduce the variance of stochastic log-determinant estimators and their derivatives and accelerate the convergence speed of the CG method. The convergence speed and accuracy of the CG and SLQ methods are determined by the spectrum of the preconditioned matrix $P^{-1/2}(\Sigma^{-1} + Z^T W Z)P^{-T/2}$; see Section \ref{sec_convergence} for more details. To practically use a matrix $P$ as a preconditioner, we need to construct it, perform linear solves with it, calculate $\log \det(P)$, and sample from $\mathcal{N}(0, P)$ in a computationally efficient manner. In the following, we describe three preconditioners that we analyze in this article. \rev{Let}
$$\rev{D = \operatorname{diag}((\Sigma^{-1} + Z^TWZ)_{ii})}$$ \rev{denote the diagonal part of $\Sigma^{-1} + Z^TWZ$. First, we consider the diagonal preconditioner given by}
\rev{$D.$} Furthermore, we analyze the ``\textbf{s}ymmetric \textbf{s}uccessive \textbf{o}ver-\textbf{r}elaxation" (SSOR) preconditioner given by
\begin{equation}
    P_{\text{SSOR}} = (L + D) D^{-1} (L + D)^T,
\end{equation}
where $L$ is a strictly lower-triangular matrix with entries $(L)_{ij} = \mathbf{1}_{\{i>j\}}(\Sigma^{-1} + Z^T W Z)_{ij} = \mathbf{1}_{\{i>j\}}(Z^T W Z)_{ij}$ such that $\Sigma^{-1} + Z^T W Z = L + L^T + D$. Since $L + D$ is lower-triangular, linear solves with $P_{\text{SSOR}}$ and its log-determinant can be calculated efficiently. In addition, the \textbf{z}ero fill-in \textbf{i}ncomplete \textbf{C}holesky (ZIC) factorization described in Appendix \ref{appendix:CGalgo} gives the ZIC preconditioner
$$P_{\text{ZIC}} = \tilde L \tilde L^T \approx \Sigma^{-1} + Z^T W Z,$$
where $\tilde L \in \mathbb{R}^{m\times m}$ is a sparse lower triangular matrix that has the sparsity pattern of the lower-triangular part of $\Sigma^{-1} + Z^T W Z$. Unfortunately, we sometimes observe breakdowns \citep{scott2014positive}, i.e., (clearly) negative numbers in the calculations of square roots for the ZIC preconditioner. In Section \ref{sec_convergence}, we theoretically analyze the SSOR and diagonal preconditioners, and in Section \ref{sec_simulation}, we empirically compare the different preconditioners.

\subsection{Predictive variances}\label{section:iterativePred}
The computational bottleneck for prediction is the calculation of the posterior predictive variances $\text{diag}(\Omega_p)$ in \eqref{postpred_Laplace_cov1}. These are required for predictive distributions of the latent and response variables as well as for predictive means of non-Gaussian response variables. Calculating $\text{diag}(\Omega_p)$ requires solving $n_p$ linear equation systems with the matrix $\Sigma^{-1} + Z^TWZ$. For high-dimensional crossed random effects and large $n_p$, solving these linear systems can thus be computationally prohibitive despite the use of the CG method. %In the following, we present a stochastic estimator relying on the CG method to calculate predictive variance computationally efficiently.

% \subsubsection{Predictive variances using a stochastic estimator for the diagonal of a matrix}
Algorithm \ref{alg:approach6} presents a method for calculating predictive variances $\text{diag}(\Omega_p)$ computationally efficiently by stochastically approximating $\text{diag}(Z_{po} (\Sigma^{-1} + Z^TWZ)^{-1} Z_{po}^T)$ in \eqref{postpred_Laplace_cov2}. The latter is based on an approach of \citet{bekas2007estimator} for approximating the diagonal of a matrix $A \in \mathbb{R}^{n_p\times n_p}$ as $\text{diag}(A) \approx \frac{1}{s}\sum_{i=1}^s z_i \odot A z_i$, where $\odot$ denotes the Hadamard product, and $z_i \in \mathbb{R}^{n_p}$ are Rademacher random vectors with entries $\pm1$. Linear solves for the stochastic estimator can be computed with the preconditioned CG method. We additionally apply variance reduction in Algorithm \ref{alg:approach6} by using a control variate based on $\text{diag}(Z_{po}P^{-1}Z_{po}^T)$: 
\DIFaddbegin
\begin{equation*}
\begin{split}
    h(z_i)&=z_i \odot Z_{po} (\Sigma^{-1} + Z^TWZ)^{-1} Z_{po}^Tz_i,\\
    r(z_i)&=z_i \odot Z_{po}P^{-1}Z_{po}^Tz_i,\\
    \text{diag}(Z_{po} (\Sigma^{-1} + Z^TWZ)^{-1} Z_{po}^T) 
    &\approx c \odot \text{diag}(Z_{po}P^{-1}Z_{po}^T)
      + \frac{1}{s}\sum_{i=1}^s\big(h(z_i)-c\odot r(z_i)\big),
\end{split}
\end{equation*}
\DIFaddend
\rev{where $c \in \mathbb{R}^{n_p}$ contains componentwise control-variate weights, each estimated as the sample covariance of $h(z_i)$ and $r(z_i)$ divided by the sample variance of $r(z_i)$.} This variance reduction is important for this stochastic estimator, as otherwise the variance of the stochastic approximation is much larger (results not shown). Algorithm \ref{alg:approach6} can be trivially parallelized and results in \rev{a consistent} approximation for $\text{diag}(\Omega_p)$; see Appendix~\ref{sec_prop_proofs} for a proof of Proposition~\ref{prop_approach6}. Assuming that $P^{-\frac{1}{2}}Z_{po}^T$ is precomputed, $P=P^{\frac{1}{2}}P^{\frac{T}{2}}$, the computational complexity of the algorithm is $O(sn_p + sl(m+n) + s\psi)$, where $l$ denotes the number of CG iterations, $s$ the number of simulation iterations, and $\psi$ the number of non-zero entries in $P^{-\frac{1}{2}}Z_{po}^T$. For instance, for a two-way crossed random-intercept model, we have $\psi\leq 2 n_p$ for the diagonal preconditioner, and $\psi\leq n_p(2 + \max_i(Z^T Z)_{ii})$ for the SSOR preconditioner. Note that \citet{gyger2024iterative} also use an approach based on \citet{bekas2007estimator} for predictive covariances in Gaussian process regression.
%time complexity
% line 3: O(sl(m+n) + s n_p)
% line 4: O(s n_p m) for given Z_po P^{-T/2} which is no longer n_p K sparse
% line 4 - direct: O(s n_p + s m) where m holds for sparse P^{1/2}, otherwise m^2
% before loop: calculating P^{-1/2} requires m forward solves of O(m) if P^{1/2} is sparse (max nnz(P^{1/2}) < K * m ) => O(m^2)
% before loop: calculating Z_po P^{-T/2} costs O(n_p m)
\begin{algorithm}[ht!]
    \caption{Predictive variances using stochastic estimator for the diagonal of a matrix}
    \label{alg:approach6}
    \begin{algorithmic}[1]
        \Require Matrices $Z_{po}, Z, \Sigma, Z_{pp}, \Sigma_p$, and $W$
        \Ensure Approximate predictive variances $\text{diag}(\hat{\Omega}_p)$
        \For{$i \gets 1$ to $s$}
            \State{Sample $z_i^{(1)}\overset{\text{i.i.d.}}{\sim} \text{Rademacher}$, $z_i^{(1)} \in \mathbb{R}^{n_p}$}
            \State{$z_i^{(2)} \gets Z_{po} (\Sigma^{-1} + Z^TWZ)^{-1} Z_{po}^T z_i^{(1)}$}
            \State{$z_i^{(3)} \gets Z_{po}P^{-1}Z_{po}^T  z_i^{(1)}$}
        \EndFor
        \State{\rev{Set $c_j\gets0$ if $\widehat{\text{Var}}((z_i^{(1)} \odot z_i^{(3)})_j)=0$; otherwise, set}}
        \Statex{\hspace{\algorithmicindent}$c_j \gets \widehat{\text{Cov}}((z_i^{(1)} \odot z_i^{(2)})_j,(z_i^{(1)} \odot z_i^{(3)})_j) / \widehat{\text{Var}}((z_i^{(1)} \odot z_i^{(3)})_j)$}
        \Statex{\hspace{\algorithmicindent}$j=1,2,\dots,n_p$, and $c=(c_1,\dots,c_{n_p})^T$}
        \State{$\text{diag}(\hat{\Omega}_p) \gets 
        \text{diag}(Z_{pp} \Sigma_p Z_{pp}^T)
        + c \odot \text{diag}(Z_{po}P^{-1}Z_{po}^T)
        + \frac{1}{s}\sum_{i=1}^s (z_i^{(1)} \odot z_i^{(2)}
        - c \odot z_i^{(1)} \odot z_i^{(3)})$}
    \end{algorithmic}
\end{algorithm}

\begin{proposition}\label{prop_approach6}
Algorithm \ref{alg:approach6} produces \rev{a consistent estimator} $\text{diag}(\hat{\Omega}_p)$ for the predictive variances $\text{diag}(\Omega_p)$ given in \eqref{postpred_Laplace_cov2}.
\end{proposition}

We have investigated several other algorithms for predictive variances reported in Appendix~\ref{app_pred_var} (Algorithm \ref{alg:approach5}, Algorithm \ref{alg:approach2}, and a Lanczos-based approach). In particular, Algorithm~\ref{alg:approach5} in Appendix~\ref{app_pred_var} presents an approach that allows for approximating predictive covariances $\Omega_p$ by sampling from a Gaussian distribution with covariance matrix $Z_{po} (\Sigma^{-1} + Z^TWZ)^{-1} Z_{po}^T$. However, none of these alternative methods resulted in more accurate estimators (see Section~\ref{section:sim_prediction}).

\subsection{Software implementation}
The methods presented in this article are implemented in the C++ library \texttt{GPBoost} \rev{\citep{sigrist2021gpboost}}, with accompanying Python and R packages\if1\blind{, see \url{https://github.com/fabsig/GPBoost}}\fi.%\footnote{Krylov subspace methods can be enabled via the parameter \texttt{matrix\_inversion\_method = "iterative"} and the preconditioner is chosen via the parameter \texttt{cg\_preconditioner\_type}.} %For linear algebra calculations, we use the \texttt{Eigen} library version 3.4.99 and its sparse matrix algebra operations whenever possible. Multi-processor parallelization is done using \texttt{OpenMP}.

\section{Convergence analysis}\label{sec_convergence}
We theoretically analyze the properties of the Krylov subspace methods presented in this article. Denote by $\lambda_{\min}=\lambda_m\leq \dots\leq \lambda_1 = \lambda_{\max}$ the eigenvalues of a symmetric matrix $\vecmat{A}\in \mathbb{R}^{m\times m}$. The spectrum of the preconditioned matrix $P^{-1/2}(\Sigma^{-1} + Z^T W Z)P^{-T/2}$ affects the CG and SLQ methods in different ways. For the CG method, a large condition number does not necessarily imply slow convergence; what matters is whether all but a few eigenvalues are clustered, leading to a small effective condition number $\kappa_{m-l,k}=\lambda_{k}/\lambda_{m-l}$ \citep[e.g.,][]{van2003iterative, nishimura2022prior}. For SLQ-based log-determinant estimation, the extremal eigenvalues are more directly relevant, both through available error bounds and through the variance of stochastic trace estimators as described below. We therefore analyze the extremal eigenvalues for the accuracy of SLQ-approximated log-determinants in Section \ref{slq_theory} and effective condition numbers for the convergence of the CG method in Section \ref{cg_theory}.

The Lanczos algorithm factorizes a symmetric matrix as $\vecmat{A} = \vecmat{Q}\vecmat{T}\vecmat{Q}^T$, where $\vecmat{Q} \in \mathbb{R}^{m\times m}$ is orthonormal, and $\vecmat{T}\in \mathbb{R}^{m\times m}$ is a tridiagonal matrix. This decomposition is computed iteratively, and $l$ iterations result in an approximation $\vecmat{A} \approx \tilde{\vecmat{Q}}\tilde{\vecmat{T}} \tilde{\vecmat{Q}}^{T}$, where $\tilde{\vecmat{Q}} \in \mathbb{R}^{m\times l}$ contains the first $l$ columns of $\vecmat{Q}$, and $\tilde{\vecmat{T}} \in \mathbb{R}^{l\times l}$ the corresponding coefficients of $\vecmat{T}$. In brief, SLQ is based on $\log\det(\vecmat{A}) = \tr\log(\vecmat{A})$, $\log(\vecmat{A}) \approx \tilde{\vecmat{Q}}\log(\tilde{\vecmat{T}}) \tilde{\vecmat{Q}}^{T}$, and stochastic trace estimation. The logarithm function used in $\log(\vecmat{A}) \approx \tilde{\vecmat{Q}}\log(\tilde{\vecmat{T}}) \tilde{\vecmat{Q}}^{T}$ is strongly affected by the extremal eigenvalues of $A=P^{-1/2}(\Sigma^{-1} + Z^T W Z)P^{-T/2}$. It is known that the Lanczos algorithm approximates the smallest and largest eigenvalues $\lambda_m$ and $\lambda_1$ the faster, the smaller $\lambda_1 - \lambda_m$ and the larger $(\lambda_{m-1} - \lambda_m)/(\lambda_1 - \lambda_{m-1})$ and $(\lambda_{1} - \lambda_2)/(\lambda_2 - \lambda_{m})$, respectively \citep[][Theorem 10.1.2, Corollary 10.1.3]{golub2013matrix}. Similar results can also be obtained for other eigenvalues. Moreover, even if the Lanczos approximation were exact, stochastic trace estimators have high variance when the spectrum of $\log(\vecmat{A})$ is poorly conditioned. The extremal eigenvalues $\lambda_{m}$ and $\lambda_{1}$ are thus expected to be important for the accuracy of an SLQ-approximated log-determinant. Besides, there is convergence theory for SLQ-approximated log-determinants with stochastic error bounds that depend on the condition number $\kappa = \lambda_1/\lambda_m$ of $P^{-1/2}(\Sigma^{-1} + Z^T W Z)P^{-T/2}$. In Theorem \ref{acc_SLQ}, we restate Theorem 3.2 of \citet{kundig2024iterative} for mixed-effects models to illustrate why extremal eigenvalues can affect SLQ approximations. %We are not aware of any convergence result that relates an effective condition number $\kappa_{m-l,k}=\lambda_{k}/\lambda_{m-l}$, $l,k>1$, to the approximation error of an SLQ-approximated log-determinant.
\begin{theorem}\label{acc_SLQ}
Let $\kappa = \lambda_1/\lambda_m$ denote the condition number of $P^{-1/2}(\Sigma^{-1} + Z^T W Z)P^{-T/2}$, and let $C_{mt} = \frac{1}{mt}Q_{\chi^2_{mt}}({1-\eta/2})$, where $Q_{\chi^2_{mt}}(\cdot)$ is the quantile function of a $\chi^2$-distribution with $mt$ degrees of freedom. Further, let $\widehat \Gamma = \frac{1}{t} \sum_{i=1}^t\|\vecmat{P}^{-\frac{1}{2}}\vecmat{z}_i\|_2^2 \vecmat{e}_1^T \log(\tilde{\vecmat{T}}_i) \vecmat{e}_1$ be an SLQ approximation. Suppose that SLQ is run with
\begin{align*}
l &\geq \frac{\sqrt{3\kappa}}{4}\log\left( \frac{20C_{mt}\log(2(\kappa+1))\sqrt{2\kappa+1}}{\epsilon}\right) &&\text{preconditioned CG steps},\\
t &\geq \frac{32}{\epsilon^2}\log(\kappa + 1)^2\log\left(\frac{4}{\eta}\right) &&\text{random vectors}.
\end{align*}
Then, for $\epsilon,\eta\in(0,1)$, the following holds for any preconditioner $\vecmat{P}$:
\begin{equation*}\label{bound_slq}
    P\!\left(\left|\widehat \Gamma - \log\det\!\left(P^{-1/2}(\Sigma^{-1} + Z^T W Z)P^{-T/2}\right)\right|\leq \epsilon m\right) \geq 1 - \eta.
\end{equation*}
 % where $\epsilon, \eta \in (0,1)$.
\end{theorem}
% In line with this, we find in our experiments in Section \ref{sec_simulation} that preconditioners for which the difference between the largest and smallest eigenvalues is small, such as the SSOR preconditioner, result in much more accurate log-determinant approximations compared to the diagonal preconditioner for which this difference is larger. 
\rev{For the remainder of Section \ref{sec_convergence}, we focus on grouped random-intercept designs, i.e., \(Z_k\in\{0,1\}^{n\times m_k}\) with exactly one non-zero entry in each row for every random-effect component \(k\). We also assume \((Z^\top WZ)_{ii}>0\) for every random effect.}

\subsection{Extremal eigenvalues and convergence of the SLQ method}\label{slq_theory}
Motivated by the SLQ error bound in Theorem \ref{acc_SLQ} and by the role of extremal eigenvalues in SLQ-approximated log-determinants outlined above, we next analyze the extremal eigenvalues $\lambda_{\min}$ and $\lambda_{\max}$ and the condition numbers $\lambda_{\max}/\lambda_{\min}$ of the preconditioned matrices $P^{-1/2}(\Sigma^{-1} + Z^T W Z)P^{-T/2}$ for the SSOR and diagonal preconditioners and the unpreconditioned matrix $\Sigma^{-1} + Z^T W Z$. We use the following notation.

\begin{notation}
We denote by $\lambda_{\min}^{\text{P}}=\lambda_m^{\text{P}}\leq \dots\leq \lambda_1^{\text{P}} =\lambda_{\max}^{\text{P}}$ the eigenvalues of $P^{-1/2}(\Sigma^{-1} + Z^T W Z)P^{-T/2}$ for a preconditioner $P$. The eigenvalues of the unpreconditioned matrix $\Sigma^{-1} + Z^T W Z$ are denoted by $\lambda_{\min}^{\text{none}}=\lambda_m^{\text{none}}\leq \dots\leq \lambda_1^{\text{none}} =\lambda_{\max}^{\text{none}}$. Furthermore, let $d_{\max} = \max_i(Z^T Z)_{ii}$ and $d_{\min} = \min_i(Z^T Z)_{ii}$ denote the largest and smallest number of occurrences per random effect, and $\sigma^2_{\min} = \min(\sigma^2_1,\dots,\sigma^2_K)~~~~\text{and} ~~~~\sigma^2_{\max} = \max(\sigma^2_1,\dots,\sigma^2_K).$ We call a random effects design ``balanced" if every random effect $(b_k)_j\in\mathbb{R}$, $j=1,\dots,m_k$, $k=1,\dots K$, occurs exactly $d$ times, i.e., $d = n/m_k = \sum_{i=1}^n (Z_k)_{ij}$ for all $1\leq j \leq m_k$ and $k=1,\dots,K$.
\end{notation}

In Lemmas \ref{EVs_SSOR} and \ref{EVs_Diag_no_PC}, we first derive bounds for the extremal eigenvalues in the general case including Gaussian and non-Gaussian likelihoods. Proofs of Lemmas \ref{EVs_SSOR} and \ref{EVs_Diag_no_PC} can be found in Appendix \ref{sec_proofs}. We then derive in Theorem \ref{cond_number_gaussian} upper and lower bounds for the condition numbers if $K=2$ and the likelihood is Gaussian. If the random effects design is balanced, the condition numbers for the SSOR and diagonal preconditioners are known exactly. Theorem \ref{cond_number_gaussian} also analyzes and compares these condition numbers asymptotically for large $d_{\max}$ and $d$, i.e., when the matrix $\Sigma^{-1} + Z^T W Z$ is usually less sparse and preconditioners are thus even more important given that computations per CG iteration are more expensive. Furthermore, in Theorem \ref{cond_number_binary}, we derive upper bounds for the condition numbers for Bernoulli likelihoods.

\begin{lemma}\label{EVs_SSOR}
If $K=2$, the following \rev{statements} hold for the SSOR preconditioner:
\begin{align}
 % & P_{\text{SSOR}}^{-1/2}(\Sigma^{-1} + Z^T W Z)P_{\text{SSOR}}^{-T/2} = I_m - \begin{pmatrix}
 %           0       & 0\\
 %            0 & D_2^{-1/2}Z_2^T W Z_1 D_1^{-1}Z_1^T W Z_2 D_2^{-1/2}
 %        \end{pmatrix},\label{pced_ssor}\\
\lambda_{m_1}^{\text{SSOR}} = \dots =\lambda_1^{\text{SSOR}} &=\lambda_{\max}^{\text{SSOR}} = 1,\label{large_ev_ssor}\\
 1 - \frac{1}{(\min(\Sigma^{-1}_{ii}(Z^TWZ)_{ii}^{-1}) + 1)^2}  &\leq  \lambda_{\min}^{\text{SSOR}} \leq  1 - \frac{1}{(\max(\Sigma^{-1}_{ii}(Z^TWZ)_{ii}^{-1}) + 1)^2}, \label{smallest_ev_ssor}\\
 1 - \frac{\max((Z^TWZ)_{ii})^2}{\min((D_1)_{ii})\min((D_2)_{ii})}  &\leq \lambda_{\min}^{\text{SSOR}} \leq  1 - \frac{\min((Z^TWZ)_{ii})^2}{\max((D_1)_{ii})\max((D_2)_{ii})}, \label{smallest_ev_ssor2}
\end{align}
\rev{where $D=\operatorname{blockdiag}(D_1,D_2)$, with $D_1\in\mathbb{R}^{m_1\times m_1}$ and $D_2\in\mathbb{R}^{m_2\times m_2}$.}
\end{lemma}

\begin{lemma}\label{EVs_Diag_no_PC}
For the diagonal preconditioner $\rev{D}$, it holds that
\begin{align}
 1 + \frac{K-1}{\max(\Sigma_{ii}^{-1}(Z^TWZ)_{ii}^{-1})+1} &\leq  \lambda_{\max}^{\text{Diag}} \leq 1 + \frac{K-1}{\min(\Sigma_{ii}^{-1}(Z^TWZ)_{ii}^{-1})+1},\label{largest_ev_diag} \\
 \frac{1}{1 + \max(\Sigma_{ii}(Z^TWZ)_{ii})} &\leq  \lambda_{m+1-k}^{\text{Diag}} \leq \frac{1}{1 + \min(\Sigma_{ii}(Z^TWZ)_{ii})}, ~~ k=1,\dots,K-1.\label{smallest_ev_diag}
\end{align}
Moreover, the following \rev{statements} hold for the unpreconditioned matrix $\Sigma^{-1} + Z^T W Z$:
\begin{align}
\frac{1}{\max(\Sigma_{ii})} + K \min((Z^TWZ)_{ii}) &\leq  \lambda_{\max}^{\text{none}} \leq \frac{1}{\min(\Sigma_{ii})} + K \max((Z^TWZ)_{ii}),\label{largest_ev_no_pc} \\
\frac{1}{\max(\Sigma_{ii})} &\leq  \lambda_{m+1-k}^{\text{none}} \leq \frac{1}{\min(\Sigma_{ii})}, ~~ k=1,\dots,K-1.\label{smallest_ev_no_pc}
\end{align}
\end{lemma}

\begin{theorem}\label{cond_number_gaussian}
If $K=2$ and the likelihood is Gaussian, it holds that
\begin{align}
\frac{1}{1 - \left(\frac{\sigma^2_{\min}d_{\min}}{\sigma^2+\sigma^2_{\min}d_{\min}}\right)^2} &\leq \frac{\lambda_{\max}^{\text{SSOR}}}{\lambda_{\min}^{\text{SSOR}}} \leq \frac{1}{1 - \left(\frac{\sigma^2_{\max}d_{\max}}{\sigma^2+\sigma^2_{\max}d_{\max}}\right)^2} = \frac{1}{2}\frac{\sigma^2_{\max}d_{\max}}{\sigma^2} + \frac{3}{4} + o(1),\quad d_{\max}\to\infty, \label{cond_numb_ssor_gauss_non_balanced}\\
2\frac{\sigma^2_{\min}d_{\min}}{\sigma^2} + 1 &\leq \frac{\lambda_{\max}^{\text{Diag}}}{\lambda_{\min}^{\text{Diag}}} \leq 2\frac{\sigma^2_{\max}d_{\max}}{\sigma^2} + 1\label{cond_numb_diag_gauss_non_balanced},\\
2\frac{\sigma^2_{\min}d_{\min}}{\sigma^2} + \frac{\sigma^2_{\min}}{\sigma^2_{\max}} &\leq \frac{\lambda_{\max}^{\text{none}}}{\lambda_{\min}^{\text{none}}} \leq 2\frac{\sigma^2_{\max}d_{\max}}{\sigma^2} + \frac{\sigma^2_{\max}}{\sigma^2_{\min}}.\label{none_cond_number}
\end{align}

If, in addition, the design of the random effects is balanced, the following hold:
\begin{align}
 \frac{\lambda_{\max}^{\text{SSOR}}}{\lambda_{\min}^{\text{SSOR}}} &= \rev{\left[1-\left(1+\frac{\sigma^2}{\sigma_1^2d}\right)^{-1}\left(1+\frac{\sigma^2}{\sigma_2^2d}\right)^{-1}\right]^{-1}} = \frac{\sigma_1^2 \sigma_2^2}{\sigma^2(\sigma_1^2+\sigma_2^2)}\, d + \rev{O(1)},\quad d\to\infty,\label{cond_number_ssor_balanced}\\
    % \frac{\lambda_{\max}^{\text{Diag}}}{\lambda_{\min}^{\text{Diag}}} &= \frac{1 + \frac{1}{\sqrt{\left(\frac{\sigma^2}{\sigma_1^2d}+ 1\right)\left(\frac{\sigma^2}{\sigma_2^2d}+ 1\right)}}}{ 1 - \frac{1}{\sqrt{\left(\frac{\sigma^2}{\sigma_1^2d}+ 1\right)\left(\frac{\sigma^2}{\sigma_2^2d}+ 1\right)}}}\\
        \frac{\lambda_{\max}^{\text{Diag}}}{\lambda_{\min}^{\text{Diag}}} &= \frac{\sqrt{\left(\frac{\sigma^2}{\sigma_1^2d}+ 1\right)\left(\frac{\sigma^2}{\sigma_2^2d}+ 1\right)}+1}{\sqrt{\left(\frac{\sigma^2}{\sigma_1^2d}+ 1\right)\left(\frac{\sigma^2}{\sigma_2^2d}+ 1\right)}-1} = \frac{4\sigma^2_1\sigma^2_2}{\sigma^2(\sigma^2_1+\sigma^2_2)}d + \rev{O(1)},\quad d\to\infty. \label{cond_number_diag_balanced}
\end{align}
\end{theorem}
A proof of Theorem \ref{cond_number_gaussian} can be found in Appendix \ref{sec_proofs}. This theorem has the following implications. First, Theorem \ref{cond_number_gaussian} shows that the bounds on the condition numbers $\lambda_{\max}^{\text{P}}/\lambda_{\min}^{\text{P}}$ and, for balanced designs, the exact SSOR- and diagonal-preconditioned condition numbers do not directly depend on either the sample size $n$ or the dimension of the random effects $m$ but rather only on (i) the number of occurrences per random effect $(Z^TZ)_{ii}$ (=$d=n/m_1=n/m_2$ when having a balanced design) and (ii) the signal-to-noise ratios $\sigma^2_k/\sigma^2$ (and $\sigma^2_{\max}/\sigma^2_{\min}$ when no preconditioning is used). Specifically, the bounds and the exact condition numbers for general and balanced designs, respectively, increase with the number of repeated occurrences and the signal-to-noise ratios. Concerning the former, the theorem shows that the condition numbers $\lambda_{\max}^{\text{P}}/\lambda_{\min}^{\text{P}}$ of the SSOR and diagonal preconditioners both grow asymptotically linearly in $d$ when having a balanced random effects design. The SSOR preconditioner improves the leading constant in this full-condition-number comparison; in the balanced case, the leading term for the diagonal preconditioner is four times larger than that for SSOR; see \eqref{cond_number_ssor_balanced} and \eqref{cond_number_diag_balanced}. Regarding the signal-to-noise ratios, similar arguments as in the derivations of the asymptotic results in \eqref{cond_number_ssor_balanced} and \eqref{cond_number_diag_balanced} can be made that $\lambda_{\max}^{\text{SSOR}}/\lambda_{\min}^{\text{SSOR}}$ and $\lambda_{\max}^{\text{Diag}}/\lambda_{\min}^{\text{Diag}}$ both have linear asymptotes in the ``average" signal-to-noise ratio $2\sigma^2_1\sigma^2_2/(\sigma^2(\sigma^2_1+\sigma^2_2))$, and the SSOR preconditioner improves the leading constant relative to the diagonal preconditioner.

%However, $\lambda_{\max}^{\text{Diag}}/\lambda_{\min}^{\text{Diag}}$ asymptotically grows four times faster in $d$ than $\lambda_{\max}^{\text{SSOR}}/\lambda_{\min}^{\text{SSOR}}$ ; see \eqref{cond_number_ssor_balanced} and \eqref{cond_number_diag_balanced}. Regarding the signal-to-noise ratios, similar arguments as in the derivations of the asymptotic results in \eqref{cond_number_ssor_balanced} and \eqref{cond_number_diag_balanced} can be made that $\lambda_{\max}^{\text{SSOR}}/\lambda_{\min}^{\text{SSOR}}$ and $\lambda_{\max}^{\text{Diag}}/\lambda_{\min}^{\text{Diag}}$ both have linear asymptotes in the ``average" signal-to-noise ratio $2\sigma^2_1\sigma^2_2/(\sigma^2(\sigma^2_1+\sigma^2_2))$. Again, $\lambda_{\max}^{\text{SSOR}}/\lambda_{\min}^{\text{SSOR}}$ grows four times slower in this signal-to-noise ratio compared to $\lambda_{\max}^{\text{Diag}}/\lambda_{\min}^{\text{Diag}}$. Analogous arguments concerning $d_{\max}$ and the signal-to-noise ratio can be made for the upper bounds of the condition numbers for general, non-balanced designs.

Moreover, Theorem \ref{cond_number_gaussian} implies $\lambda_{\max}^{\text{SSOR}} < \lambda_{\max}^{\text{Diag}}$ and $\lambda_{\min}^{\text{SSOR}} > \lambda_{\min}^{\text{Diag}}.$ This holds for all random effects designs and parameter choices. For the balanced design, we can additionally obtain an explicit expression for the condition number of the unpreconditioned matrix $\Sigma^{-1} + Z^T W Z$ if the two random effects have equal variances since the inequalities in \eqref{none_cond_number} then become equalities. In this special case, the condition numbers of the unpreconditioned and the diagonal-preconditioned matrix are equal and given by $2\sigma_1^2d/\sigma^2 + 1$.%, whereas the condition number of the SSOR-preconditioned matrix asymptotically grows linearly with a rate of $\sigma^2_1/(2\sigma^2)$ in $d$. %We conjecture that similar results also hold for the general case $\sigma^2_1 \neq \sigma^2_2$.

Next, in Theorem~\ref{cond_number_binary}, we derive upper bounds for the condition numbers of $P^{-1/2}(\Sigma^{-1} + Z^T W Z)P^{-T/2}$ for a Bernoulli likelihood with logit and probit links used in binary classification. Similarly to the Gaussian case, the upper bounds for all three preconditioners grow linearly in (i) the maximal number of repeated random effects occurrences $d_{\max}$ and (ii) the maximal random effects variance $\sigma^2_{\max}$. Compared to no preconditioning and the diagonal preconditioner, the SSOR preconditioner improves the leading constant. A proof of Theorem~\ref{cond_number_binary} can be found in Appendix~\ref{sec_proofs}.

% for the condition number $\lambda_{\max}^{\text{SSOR}}/\lambda_{\min}^{\text{SSOR}}$ of the SSOR preconditioner given in \eqref{cond_numb_ssor_binary_non_balanced} is smaller and grows four times slower in both (i) the maximal number of repeated random effects occurrences $d_{\max}$ and (ii) the maximal random effects variance $\sigma^2_{\max}$ compared to the diagonal preconditioner and the unpreconditioned matrix given in \eqref{cond_numb_diag_binary_non_balanced} and \eqref{none_cond_number_binary}, respectively. 

\begin{theorem}\label{cond_number_binary}
If $K=2$ and the likelihood is Bernoulli with a logit or probit link, it holds that
\begin{align}
\frac{\lambda_{\max}^{\text{SSOR}}}{\lambda_{\min}^{\text{SSOR}}} &\leq \frac{1}{1 - \left(\frac{\sigma^2_{\max}d_{\max}C}{1+\sigma^2_{\max}d_{\max}C}\right)^2} = \frac{1}{2}\sigma^2_{\max}d_{\max}C  + \frac{3}{4} + o(1) ~~ ( d_{\max}\to\infty) \label{cond_numb_ssor_binary_non_balanced},\\
\frac{\lambda_{\max}^{\text{Diag}}}{\lambda_{\min}^{\text{Diag}}} &\leq 2\sigma^2_{\max}d_{\max}C + 1\label{cond_numb_diag_binary_non_balanced},\\
\frac{\lambda_{\max}^{\text{none}}}{\lambda_{\min}^{\text{none}}} &\leq 2\sigma^2_{\max}d_{\max}C + \frac{\sigma^2_{\max}}{\sigma^2_{\min}},\label{none_cond_number_binary}
\end{align}
where $C$ equals $0.25$ and $1$ for a logit and a probit link, respectively.
\end{theorem}

In summary, the results in Theorems \ref{acc_SLQ}, \ref{cond_number_gaussian}, and \ref{cond_number_binary} imply that the SSOR preconditioner is expected to yield more accurate SLQ-approximated log-determinants compared to no preconditioning and to the diagonal preconditioner for both Gaussian and Bernoulli likelihoods. These findings are empirically confirmed in the experiments in Section \ref{sec_simulation}. In the next section, we address the convergence of the CG method by analyzing effective condition numbers.
% In particular, this holds in challenging cases when the number of repeated random effect occurrences are large, and thus $\Sigma^{-1} + Z^T W Z$ is less sparse, and when the signal-to-noise ratios $\sigma^2_k/\sigma^2$ are large. 

\subsection{Effective condition numbers and convergence of the CG method}\label{cg_theory}
Under additional assumptions on the random effects design, Theorem \ref{second_smallest_EV_SSOR_gaussian} gives a stochastic upper bound for the effective condition number \(\kappa_{m-1,1}^{\mathrm{SSOR}}=\lambda_{1}^{\mathrm{SSOR}}/\lambda_{m-1}^{\mathrm{SSOR}}\) of the SSOR-preconditioned matrix after removing the smallest eigenvalue.

\begin{theorem}\label{second_smallest_EV_SSOR_gaussian}
If $K=2$, the likelihood is Gaussian, and $\begin{pmatrix} 0  & Z_1^T Z_2\\  Z_2^T Z_1 & 0 \end{pmatrix} \in \{0,1\}^{m\times m}$ is an adjacency matrix of a random bipartite, biregular graph with uniform distribution over all bipartite, biregular graphs, the following holds:
\begin{equation}\label{cond_second_smalles_EV_SSOR}
\kappa_{m-1,1}^{\text{SSOR}}=\frac{\lambda_{1}^{\text{SSOR}}}{\lambda_{m-1}^{\text{SSOR}}} \leq \frac{1}{ 1 - \left(\frac{1}{d_1} + \frac{1}{d_2} + \frac{2}{\sqrt{d_1d_2}}\right) - \epsilon_m}
\end{equation}
asymptotically almost surely with \rev{$\epsilon_m\rightarrow 0$ as $\min(m_1,m_2)\rightarrow \infty$}, where $d_k = \frac{n}{m_k}$ for $k=1,2$. \rev{Here, the degrees $d_1$ and $d_2$ are held fixed as $\min(m_1,m_2)\rightarrow\infty$.} \rev{The bound applies whenever the denominator in \eqref{cond_second_smalles_EV_SSOR} is positive.}
\end{theorem}

%This theorem shows that $\kappa_{m-1,1}^{\text{SSOR}}$ decreases fast to one as the number of occurrences per random effect $d_k$ increases. 
The assumptions of Theorem \ref{second_smallest_EV_SSOR_gaussian} imply that the random effects design is balanced \rev{within each component}, with every random effect in component \(k\) occurring \(d_k=n/m_k\) times, and that pairs of levels from the two random-effect components co-occur at most once, i.e., \(Z_2^\top Z_1\in\{0,1\}^{m_2\times m_1}\). A proof is given in Appendix \ref{sec_proofs}. For comparison, under the same design assumptions, the corresponding diagonal-preconditioned effective condition number satisfies
\begin{equation}\label{eff_cond_diag}
    \kappa_{m-1,2}^{\text{Diag}}=\frac{\lambda_{2}^{\text{Diag}}}{\lambda_{m-1}^{\text{Diag}}} \leq \frac{ 1 + \frac{1}{\sqrt{d_1}} + \frac{1}{\sqrt{d_2}} + \epsilon_m}{ 1 - \left(\frac{1}{\sqrt{d_1}} + \frac{1}{\sqrt{d_2}}\right) - \epsilon_m}.
\end{equation}
The \rev{equal-dimension} case \(m_1=m_2\) in Corollary~\ref{eff_cond_SSOR_Diag_gaussian_balanced} makes the difference between the two bounds explicit.

\begin{corollary}\label{eff_cond_SSOR_Diag_gaussian_balanced}
If $K=2$, the likelihood is Gaussian, $m_1=m_2$, and $\begin{pmatrix} 0  & Z_1^T Z_2\\  Z_2^T Z_1 & 0 \end{pmatrix} \in \{0,1\}^{m\times m}$ is an adjacency matrix of a bipartite, biregular random graph with uniform distribution over all bipartite, biregular random graphs, \rev{the following two statements hold, where $d=n/m_1=n/m_2$. For every fixed $d$, asymptotically almost surely as $m\rightarrow\infty$,}
\begin{align}
    \kappa_{m-1,1}^{\text{SSOR}} = \frac{\lambda_{1}^{\text{SSOR}}}{\lambda_{m-1}^{\text{SSOR}}} & \leq \frac{1}{ 1 - \frac{4}{d} - \epsilon_m}\label{ineq_asym_eff_cond_ssor},\\
    \kappa_{m-1,2}^{\text{Diag}} = \frac{\lambda_{2}^{\text{Diag}}}{\lambda_{m-1}^{\text{Diag}}}& \leq \frac{ 1 + \frac{2}{\sqrt{d}} + \epsilon_m'}{ 1 - \frac{2}{\sqrt{d}} - \epsilon_m'}\label{ineq_asym_eff_cond_diag},
\end{align}
\rev{with $\epsilon_m\rightarrow 0$ and $\epsilon_m'\rightarrow 0$ as $m\rightarrow \infty$, provided the denominators are positive. Moreover, the deterministic degree-dependent parts of the bounds in \eqref{ineq_asym_eff_cond_ssor} and \eqref{ineq_asym_eff_cond_diag}, obtained by setting the vanishing random-graph remainders $\epsilon_m$ and $\epsilon_m'$ to zero, satisfy, as $d\rightarrow\infty$,}
\begin{align}
    \frac{1}{ 1 - \frac{4}{d}} &= 1 + \frac{4}{d} + O\left(\frac{1}{d^2}\right),~~~~
    \frac{ 1 + \frac{2}{\sqrt{d}}}{ 1 - \frac{2}{\sqrt{d}}} = 1 + \frac{4}{\sqrt{d}} + O\left(\frac{1}{d}\right)\label{exp_eff_cond}.
\end{align}
\end{corollary}

%\rev{The two limits in Corollary \ref{eff_cond_SSOR_Diag_gaussian_balanced} are deliberately kept separate. The spectral-gap result of \citet{brito2022spectral} underlying Theorem \ref{second_smallest_EV_SSOR_gaussian} is established in a bounded-degree regime, as the number of vertices tends to infinity for fixed degrees. It therefore justifies part (i) for each fixed $d$, but it does not establish a joint limit in which $d\rightarrow\infty$ while the remainders $\epsilon_m$ and $\epsilon_m'$ remain negligible uniformly. Part (ii) is thus a purely algebraic statement about the degree-dependent parts of the bounds and should be read as such.}

\rev{These results complement \citet{pandolfi2024conjugate}, who show that diagonal preconditioning can yield fast CG convergence when all but a few eigenvalues cluster. By \eqref{exp_eff_cond}, the degree-dependent parts of the SSOR and diagonal bounds differ from one by $O(d^{-1})$ and $O(d^{-1/2})$, respectively, so the SSOR bound approaches one at a faster asymptotic rate in \(d\). This difference arises because diagonal preconditioning is controlled by nontrivial singular values of a normalized bipartite adjacency block, whereas SSOR depends on their squares. SSOR is therefore expected to converge faster in this design regime. Note that \citet{pandolfi2024conjugate} include a random intercept and consider $\kappa_{m-2,3}^{\text{Diag}}$ instead of $\kappa_{m-1,2}^{\text{Diag}}$; Theorem~\ref{second_smallest_EVs_diag_gaussian} in Appendix~\ref{sec_proofs} establishes \eqref{eff_cond_diag} without a global random intercept.}

% Specifically, if $K=2$, the likelihood is Gaussian, and  $\begin{pmatrix}
%            0  & Z_1^T Z_2\\
%             Z_2^T Z_1 & 0
%         \end{pmatrix} \in \{0,1\}^{m\times m}$ an adjacency matrix of a bipartite, biregular random graph, we have the result below. The  {pandolfi2024conjugate}.

\section{Experiments using simulated data}\label{sec_simulation}
In the following, we conduct experiments using simulated data. We compare the different preconditioners and methods for calculating predictive variances. Furthermore, we analyze our proposed methods concerning the accuracy and runtime for parameter estimation and prediction. Krylov subspace methods are compared to computations relying on sparse Cholesky decompositions. \rev{Note that Cholesky decompositions are subject to finite-precision round-off errors, which may accumulate in sequential Schur-complement updates, and hence are used as the deterministic numerical baseline rather than as an error-free mathematical quantity}. Code to reproduce the simulated and the real-world data experiments of this article is available at \url{https://github.com/pkuendig/KrylovGMMs}.

\subsection{Experimental setting}\label{exp_setting}
We simulate data from a model with two crossed random effect components. Unless stated otherwise, we use a balanced random effects design, i.e., each random effect $b_{k,j}$, $j=1,\dots,m_k$, $k=1,2$, occurs the same number of times, $\sum_{i=1}^n (Z_1)_{ij}=\sum_{i=1}^n (Z_2)_{ij}=n/m_1=n/m_2$, both random effect components $b_1$ and $b_2$ have the same dimension $m_1=m_2$, the crossing structure between the two random effect components is random, and the random effects variances are $\sigma_1^2=\sigma_2^2=0.25$. A random crossing structure is generated by generating incidence matrices $Z_k\in\{0,1\}^{n\times m_k}$ as follows. We first assign exactly $n/m_k$ ones to column $j$ in consecutive rows, $(Z_k)_{ij} = 1 \quad \text{for } 
i = (j-1)n/m_k+ 1,\dots, jn/m_k$, $j=1,\dots,m_k$, and zero otherwise, and we randomly permute the rows of $Z_2$. We consider various sample sizes $n$ and dimensions $m$ of the random effects as described in the subsections below. The response variable $y$ follows either a Gaussian likelihood with error variance $\sigma^2=0.25$ or a Bernoulli likelihood with a logit link function. We additionally conduct experiments for less regular unbalanced designs, unequal dimensions of the random effect components ($m_1\neq m_2$), and other signal-to-noise ratios. See the subsections below for more information. For the experiments regarding the accuracy and runtimes of parameter estimation and prediction in Sections \ref{sec:times} and \ref{acc_param_estim}, we include a fixed effects linear predictor term with five covariates and an intercept. The regression coefficients are all set to one except for the intercept, which is set to zero, and the covariates are sampled from a normal distribution with mean zero and variance chosen such that the fixed and random effects have equal variance.

To analyze prediction accuracy, we randomly split the data into training and test sets. We measure both the accuracy of point and probabilistic predictions for the test random effects $\tilde{b}_p=Z_{po}b + Z_{pp}b_p$; see \eqref{pred_mean}. \rev{Their approximate predictive distribution is $\mathcal{N}(\tilde{\omega}_p,\Omega_p)$, where $\tilde{\omega}_p=Z_{po}b^*=\omega_p-F(X_p)$, and $\omega_p$ and $\Omega_p$ are defined in Section \ref{sec_prediction}. We evaluate the predictive means $\tilde{\omega}_p$ using the root mean squared error (RMSE) and probabilistic predictions using the log score (LS) $-\frac{1}{n_p}\sum_{i=1}^{n_p}\log\mathcal{N}((\tilde{b}_p)_i;\tilde{\omega}_{p,i},(\Omega_p)_{ii})$.} We also measure the runtime in seconds.
% such that most group levels are observed in both the training and test data
% Note that the probability that a group-level is missing in the training or test set is very small, when doing randomly splitting.

All calculations are done on a laptop with an Intel i7-12800H processor and 32 GB of random-access memory. \rev{We use \texttt{GPBoost} 1.7.4, the \texttt{R} packages \texttt{lme4} 2.0-6, \texttt{glmmTMB} 1.1.14, and \texttt{R-INLA} 26.08.07, and the Julia package \texttt{MixedModels.jl} 5.8.2. \texttt{GPBoost}, \texttt{lme4}, \texttt{glmmTMB}, and \texttt{MixedModels.jl} use the empirical Bayes approach maximizing the marginal likelihood, and, for non-Gaussian likelihoods, all four use Laplace approximations. \texttt{R-INLA} is run with its empirical-Bayes integration strategy (\texttt{int.strategy = 'eb'}); see Appendix \ref{appendix:inla} for more information. Where supported, we use the same initial parameter values and optimization is performed using package-specific internal optimizers. All methods were run using their default parallelization settings on the same machine.}
% , which is the default option in all packages except for \texttt{lme4}, where the restricted maximum likelihood criterion is optimized by default.
% Consequently, the different implementations rely on the same statistical methodology and differ only in their specific software implementations. We therefore expect very similar parameter estimates and predictions.
% In particular, to minimize the negative log-marginal likelihood, the `limited-memory BFGS' algorithm is used in the \texttt{GPBoost} package, and regression coefficients for Gaussian likelihoods are found using coordinate descent and weighted least squares. The \texttt{glmmTMB} package uses the `nlminb' optimizer, and the \texttt{lme4} package uses the `nloptwrap' optimizer for Gaussian likelihoods and the `bobyqa' and `Nelder-Mead' optimizers for non-Gaussian likelihoods.

For the Krylov subspace methods, if not stated otherwise, we use the SSOR preconditioner and $t=50$ random vectors for STE and SLQ. For the CG algorithm, we use an absolute convergence tolerance on the Euclidean norm of the unpreconditioned residual (see Appendix~\ref{appendix:CGalgo}) of $10^{-2}$ for calculating marginal likelihoods and gradients and $10^{-3}$ for predictive variances. Unless stated otherwise, we use Algorithm \ref{alg:approach6} and $s=1000$ samples for calculating predictive variances. Further, we adopt a sample average approximation approach \citep{kim2015guide} when maximizing the marginal likelihood.

\subsection{Preconditioner comparison}\label{prec_compare}
We compare the preconditioners introduced in Section \ref{sec_prec} with regard to the accuracy and runtimes for approximating log-marginal likelihoods. We use $m=100{,}000$ random effects and a sample size of $n=1{,}000{,}000$. These dimensions are approximately the largest such that Cholesky-based computations can be run in a reasonable amount of time (approx. four hours for one likelihood evaluation). We use both a balanced random effects design as described in Section \ref{exp_setting} and a less regular unbalanced design. The latter is generated as described in Appendix \ref{appendix:unbalance_design}. The calculation of the marginal likelihoods is done at the data-generating parameters and repeated $100$ times with different random vectors for the SLQ method. 

First, we analyze how the accuracy and runtimes depend on the number of random vectors $t$ in the SLQ method for a Gaussian likelihood. Figure \ref{fig:Pcomparison_Gaussian} reports standard deviations of the stochastic estimates and average wall-clock times for the different preconditioners for both the balanced and unbalanced random effects designs. The figure also reports the number of iterations (bottom row, right axis) of the CG method used for calculating the quadratic \rev{form} in the log-marginal likelihood. The wall-clock time and negative log-marginal likelihood when using a Cholesky decomposition are annotated in the plot. \rev{First, we observe that the Krylov subspace-based methods are much faster than Cholesky-based calculations while yielding essentially equal results. For instance, when using $t=50$ random vectors, they provide a speedup of approximately four orders of magnitude. At this setting, the standard deviation of the SLQ-approximated log-marginal likelihood with the SSOR preconditioner is approximately one, which is about $10^{-6}$ times the magnitude of the Cholesky-based negative log-marginal likelihood.} Moreover, we observe that the SSOR and ZIC preconditioners yield the most accurate log-marginal likelihood approximations. Compared to the diagonal preconditioner and the unpreconditioned SLQ method, the SSOR and ZIC preconditioners yield standard deviations that are approximately one order of magnitude lower for both random effects designs. This observation is consistent with the theoretical results presented in Section \ref{sec_convergence}. Furthermore, as expected from theory, the diagonal preconditioner does not result in any variance reduction for the balanced design, but for the unbalanced design, we observe smaller standard deviations compared to the unpreconditioned SLQ method. The runtimes and the number of CG iterations are similar for all preconditioners and also without preconditioning for the balanced design. For the unbalanced design, all preconditioners substantially reduce the runtimes and numbers of CG iterations. In line with the convergence theory in Section \ref{sec_convergence}, the diagonal preconditioner has a slightly higher number of CG iterations and runtimes compared to the SSOR preconditioner. Note that the SSOR preconditioner has a higher ``overhead" per CG iteration than the diagonal preconditioner since additional sparse triangular solves are done for the SSOR preconditioner. Note that SLQ-approximated log-marginal likelihoods show no apparent bias for all preconditioners in our experiments; see Figure \ref{fig:Pcomparison_Gaussian_hist} in Appendix \ref{appendix:add_Pcomparison}.
\begin{figure}[ht!]
    \centering
    \includegraphics[width=0.49\linewidth]{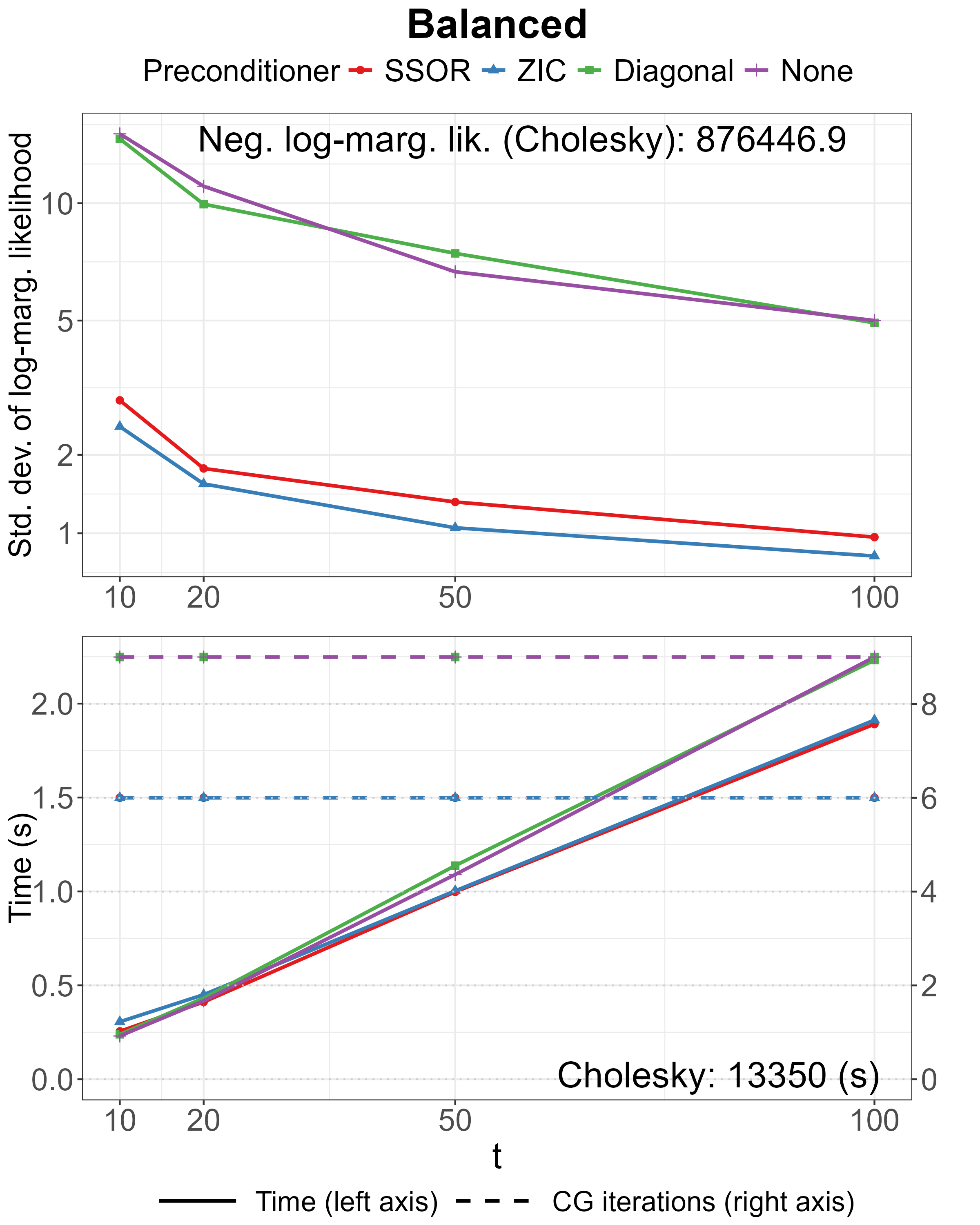}
    \includegraphics[width=0.49\linewidth]{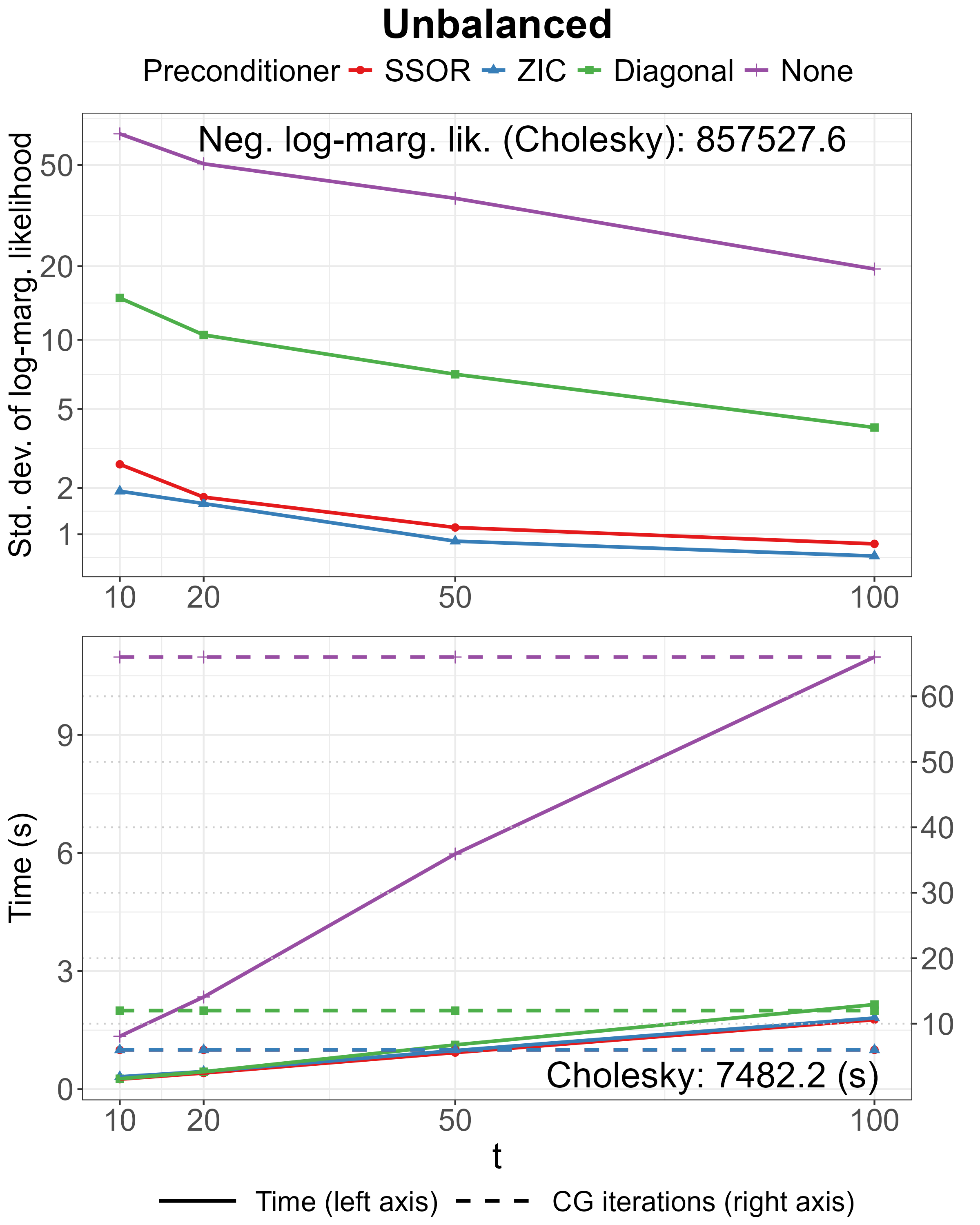}
    \caption{Accuracy (\rev{standard deviation of the stochastic log-marginal likelihood estimates}; top row), wall-clock time (\rev{bottom row, solid lines and left axis}), and number of CG iterations \rev{per linear solve} (\rev{bottom row, dashed lines and right axis}) for different preconditioners and varying numbers of random vectors $t$ in the SLQ method for calculating log-marginal likelihoods for a Gaussian likelihood and balanced and unbalanced random effects designs. \rev{The colored horizontal dashed lines therefore represent the CG iteration counts and the solid lines represent runtimes. The runtimes and negative log-marginal likelihood when using a Cholesky decomposition are annotated in the plot.}}
    \label{fig:Pcomparison_Gaussian}
\end{figure}

In Figures \ref{fig:Pcomparison_Laplace} and \ref{fig:Pcomparison_Laplace_hist} in Appendix~\ref{appendix:add_Pcomparison}, we report analogous results when the likelihood is Bernoulli with a logit link function and all other settings are identical. We find qualitatively very similar results to those in the Gaussian setting. In particular, \rev{Krylov subspace-based methods are approximately $30{,}000$ times faster than Cholesky-based computations while again yielding essentially equal results. For instance with $t=50$ and the SSOR preconditioner, the standard deviation is approximately $0.4$, or about \(6\times10^{-7}\) times the magnitude of the Cholesky-based negative log-marginal likelihood.} The SSOR and ZIC preconditioners yield substantially more accurate log-marginal likelihood approximations than the diagonal preconditioner and the unpreconditioned case. In Figure \ref{fig:Pcomparison_Gaussian_other_variance} in Appendix~\ref{appendix:add_Pcomparison}, we present additional results when having a higher signal-to-noise ratio obtained by using a lower error variance $\sigma^2=0.05$ for a Gaussian likelihood. The results are again qualitatively similar to those for the lower signal-to-noise ratio. We further analyze the impact of the irregularity of the random effects design. For this, we vary the size parameter $r$ in the unbalanced design generation approach explained in Appendix~\ref{appendix:unbalance_design}. The results are reported in the left column of Figure \ref{fig:Pcomparison_Gaussian_m_over_disp} in Appendix~\ref{appendix:add_Pcomparison}. We find that the runtime of the unpreconditioned CG method increases with increasing irregularity, whereas the runtimes remain constant for the preconditioned CG method for all preconditioners. The accuracy of the SLQ method is only marginally affected by the irregularity of the design with and without preconditioning. In the right column of Figure \ref{fig:Pcomparison_Gaussian_m_over_disp} in Appendix~\ref{appendix:add_Pcomparison}, we also report results when varying the random effects dimension $m$. An unbalanced design and a constant ratio $n/m=10$ (as above) are used for this experiment. We find that there is a small increase in the variance of the SLQ method, \rev{and the numbers of CG iterations are essentially constant, but the runtimes of the CG method increase because each iteration becomes more computationally expensive as $m$ increases.}

\rev{As a practical rule, the SSOR preconditioner is preferable when SLQ log-determinants are required since it substantially reduces stochastic variability at similar runtimes in our experiments. Diagonal preconditioning is a simpler alternative for applications that require only CG linear solves, but the runtime differences among preconditioners are comparatively small.}

% Note that the experiments in this section and the following Section \ref{section:sim_prediction} are conducted on only one simulated data set since we do not want to mix sampling variability and randomness of the simulation-based approximations. However, the results do not change when using other samples (results not shown).

% Cholesky-based computations are faster for the unbalanced design, since $Z^T W Z$ is more sparse, but still three orders of magnitude slower than Krylov subspace methods.

\subsection{Accuracy and runtimes of predictive variance approximations}\label{section:sim_prediction}
Next, we compare the accuracy of the different methods for calculating predictive variances introduced in Section \ref{section:iterativePred} and Appendix \ref{app_pred_var}, using simulated data with $n=n_p=100{,}000$ training and test points and a total of $m+m_p=10{,}000$ group levels. The response variable follows a Gaussian likelihood, and we calculate the RMSE between the approximate predictive variances and the ``exact" Cholesky-based predictive variances. We also measure the runtime for prediction, which includes the calculation of the latent predictive means and variances. Predictive distributions are calculated at the data-generating variance parameters. 

Figure \ref{fig:RMSE_vs_TIME} visualizes the accuracy (RMSE) versus the wall-clock time for different numbers of random vectors $s$ for the stochastic methods and different ranks $k$ for the Lanczos-based approximation. For the former, we average the results over $100$ independent repetitions for every $s$ and add confidence intervals for the RMSE obtained as $\pm 2 \times$ standard errors (not visible since they are very small). The runtime of Cholesky-based calculations is annotated in the plot. By far the most accurate predictive variances for a given runtime are obtained with the simulation-based Algorithm \ref{alg:approach6}. For instance, for $s=200$, predictive variances are very accurate with an RMSE of approximately $2.6\times10^{-4}$, and the method is more than $100$ times faster than the Cholesky-based calculation. %Algorithm \ref{alg:approach5} achieves the second highest prediction accuracy for a given runtime, and Algorithm \ref{alg:approach2} is considerably slower and less accurate. Finally, Lanczos-based approximations are very inaccurate and the RMSE decreases only very slowly with increasing rank $k$.

\begin{figure}[ht!]
\centering
    \includegraphics[width=0.6\linewidth]{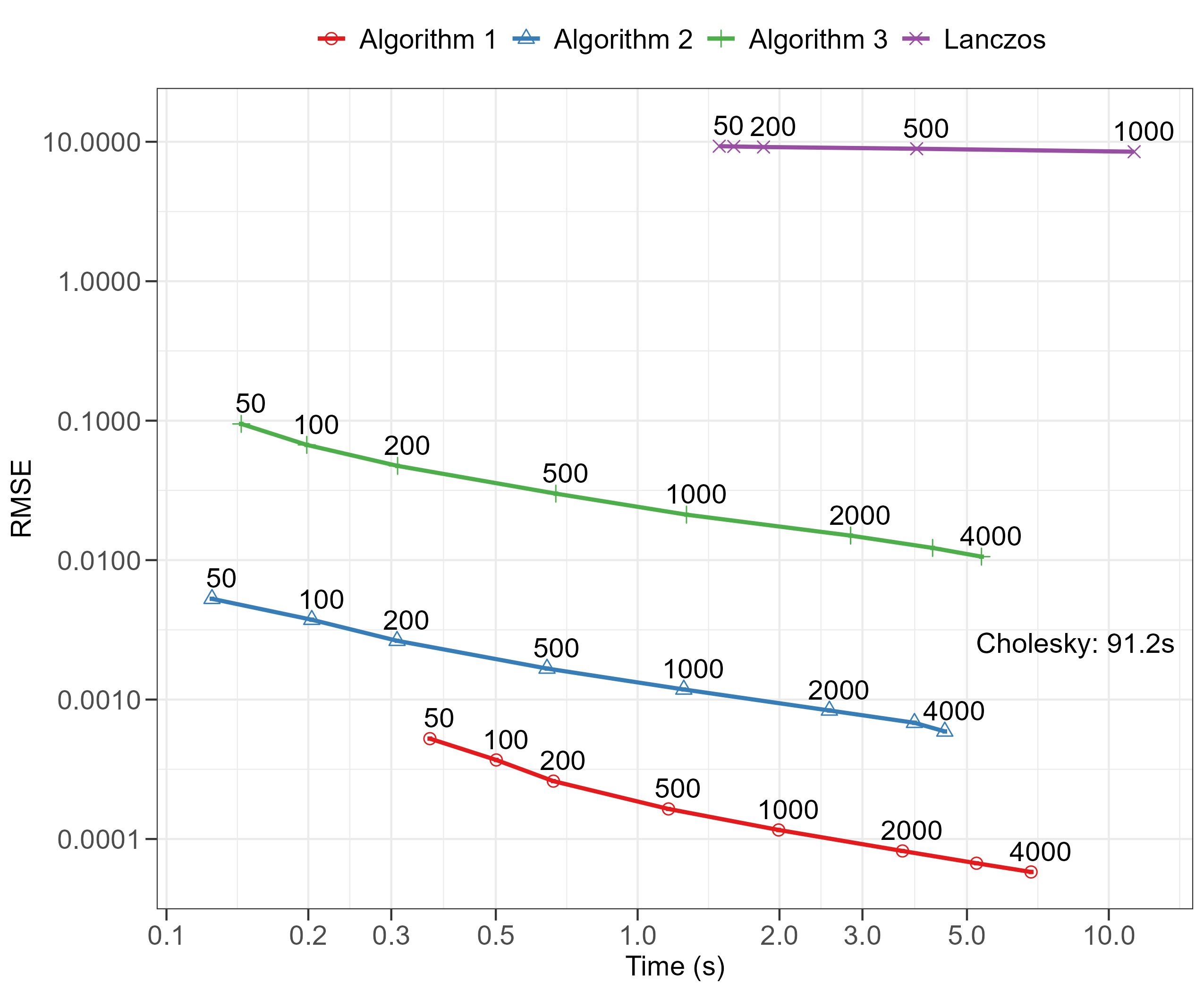}
    \caption{Accuracy-runtime comparison of different methods for predictive variances. The number of random vectors $s$ and the Lanczos rank $k$ are annotated in the plot.}
    \label{fig:RMSE_vs_TIME}
\end{figure}

\subsection{Runtime for parameter estimation} \label{sec:times}
\rev{We compare parameter-estimation runtimes for the Krylov approach, a Cholesky baseline using the same software implementation (\texttt{GPBoost}), \texttt{lme4}, \texttt{glmmTMB}, \texttt{MixedModels.jl} \citep{bates2025mixed}, and \texttt{R-INLA} (EB), for $n=10m$ and $m=1{,}000$ to $1{,}000{,}000$ using the setting described in Section \ref{exp_setting}. In particular,
we use a model with two crossed random effect components and five covariates plus an intercept in a linear predictor. Figure \ref{fig:runtime} shows that Krylov methods are fastest throughout. For example, at $m=20{,}000$ with a Gaussian likelihood, they are approximately three orders of magnitude faster than the same-backend Cholesky baseline. At the largest $m$ reached by each competing package, the Krylov approach is roughly three orders of magnitude faster, with speedups approaching four orders in some cases. Because the competing runtimes increase more steeply with $m$, these observed speedups likely understate the advantage at larger dimensions.}
\begin{figure}[ht!]
    \centering
    \includegraphics[width=\linewidth]{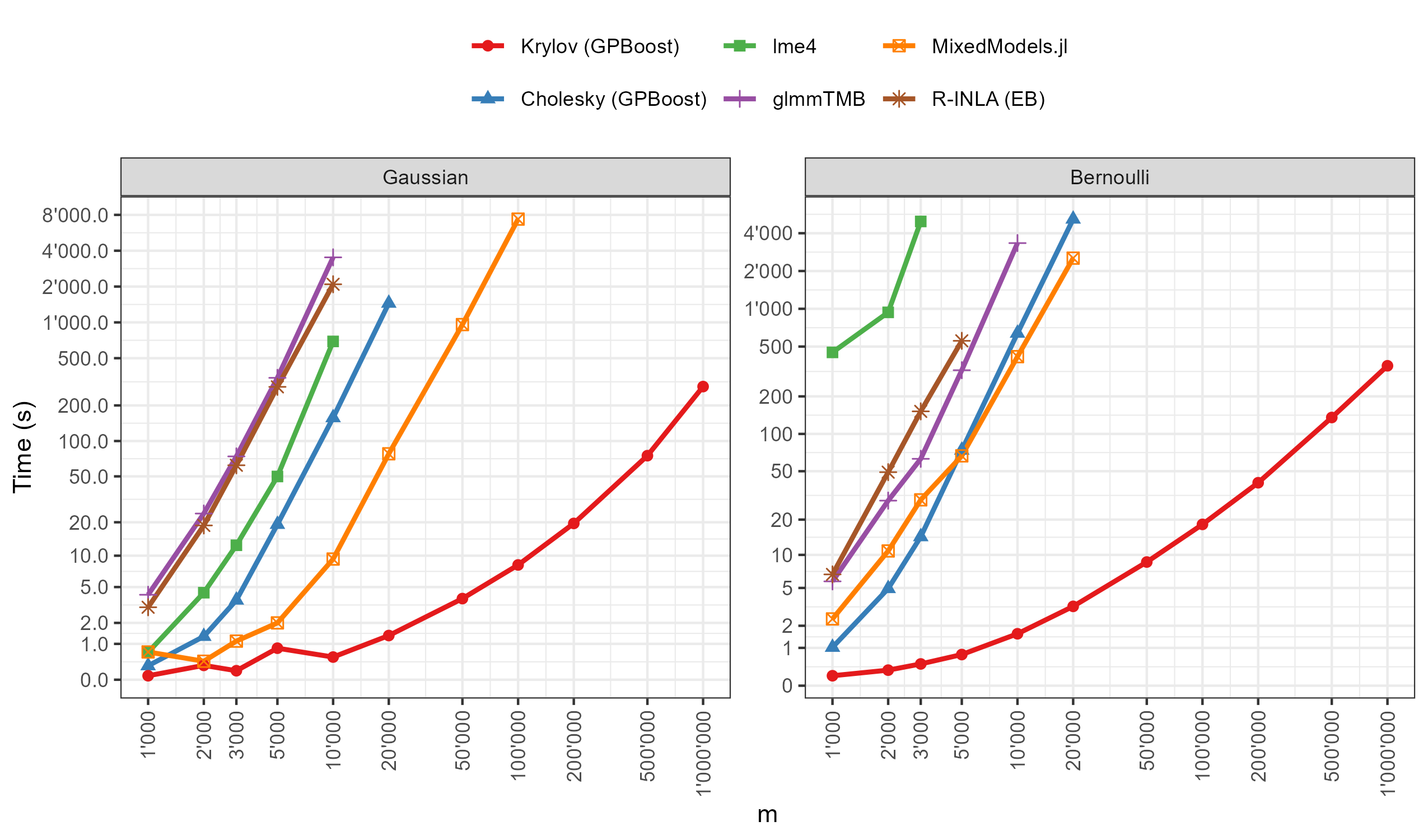}
    \caption{\rev{Parameter-estimation wall-clock times (s) for different $m$ and Gaussian or Bernoulli likelihoods.}}
    \label{fig:runtime}
\end{figure}
\rev{Figure \ref{fig:runtime_unbalanced} in Appendix \ref{appendix:add_res_est_pred} gives analogous results for $m_2=m_1/2$. The increased sparsity makes Cholesky calculations faster than in the balanced design, but Krylov methods clearly remain fastest throughout.} %Note that it would also be interesting to perform a runtime comparison for calculating log-marginal likelihoods as this would eliminate the impact of using different optimizers across different libraries, but the \texttt{lme4} and \texttt{glmmTMB} packages do not support the explicit calculation of marginal likelihoods. However, `Krylov (GPBoost)' and `Cholesky (GPBoost)' differ only in the use of Krylov subspace methods vs. Cholesky-based computations, and this comparison therefore clearly demonstrates the speedup with Krylov subspace methods.

\subsection{Accuracy of parameter estimation and predictive distributions}\label{acc_param_estim}
In the following, we analyze the properties of the variance parameter and linear regression coefficient estimates and the accuracy of predictive distributions for Gaussian and Bernoulli likelihoods when using Krylov subspace methods. We set the total number of training and prediction random effects to $m+m_p=4{,}000$ and simulate $n=n_p=40{,}000$ training and test points. This is roughly the largest number of group levels such that we can run calculations with the \texttt{glmmTMB} and \texttt{lme4} packages in a reasonable amount of time. \rev{\texttt{MixedModels.jl} and \texttt{R-INLA} are included only in the runtime comparison and real-world comparisons.} We perform $100$ simulation repetitions, and prediction is done using estimated parameters. Figure \ref{fig:var_estimates} and Tables \ref{table:estimates_Gaussian} and \ref{table:estimates_Bernoulli} in Appendix \ref{appendix:add_res_est_pred} report the estimates for the variance parameters for both Gaussian and Bernoulli likelihoods, and Figure \ref{fig:RMSE_LS} and Tables \ref{table:prediction_Gaussian} and \ref{table:prediction_Bernoulli} in Appendix \ref{appendix:add_res_est_pred} report the RMSE for predictive means and the log score for predictive distributions. \rev{We observe virtually no differences among Krylov subspace methods and Cholesky-based computations for both parameter estimates and predictive distributions. At the reported precision, their predictive RMSEs and log scores are identical, and the differences in parameter-estimation RMSEs are smaller than the Monte Carlo variability across simulation repetitions.}

\section{Real-world applications}\label{sec:real_world}
Next, we conduct experiments on six real-world data sets with crossed random effects, analyzing parameter-estimation runtimes and the agreement among parameter estimates. \rev{We compare Krylov and same-backend (\texttt{GPBoost}) Cholesky calculations with Cholesky-based computations in the packages \texttt{lme4}, \texttt{glmmTMB}, \texttt{MixedModels.jl}, and \texttt{R-INLA} (EB).} Unless stated otherwise, we use the computational settings described in Section~\ref{exp_setting}. We consider four regression (`cars', `building\_permits', `instEval', `MovieLens\_32m') and two classification (`employee\_access', `upselling') data sets. See Table \ref{table:real_data_sets} and Appendix \ref{section:add_real_world} for details on the data sets. In Figure \ref{fig:real_world_ZtZ} in Appendix \ref{section:add_real_world}, we additionally visualize the non-zero entries of the matrix $Z^TZ$ for the different data sets. We use the ZIC preconditioner for the `cars', `building\_permits', and `instEval' data sets and the SSOR preconditioner for the other data sets.
\begin{table}[ht]
\centering
\ifrevision\color{red}\fi
\begin{tabular}{lrrrlrr}
  \hline
  data set & $n$ & $p$ & $K$ & $y$ & $m$ & $\text{nnz}(Z^TZ)$\\
  \hline
  cars & 97{,}729 & 66 & 2 & log(price) & 27{,}461 & 185{,}807\\
  building\_permits & 527{,}168 & 6 & 3 & log(cost) & 510{,}729 & 2{,}544{,}893\\
  instEval & 73{,}421 & 22 & 2 & rating & 4{,}100 & 150{,}942\\
  MovieLens\_32m & 32{,}000{,}204 & 2 & 2 & rating & 285{,}380 & 64{,}285{,}788\\
  employee\_access & 32{,}769 & 0 & 9 & approval & 15{,}626 & 469{,}262\\
  upselling & 50{,}000 & 34 & 4 & up-selling & 25{,}372 & 454{,}852\\
  \hline
\end{tabular}
\caption{Summary of real-world data sets. $n$: sample size; $p$: number of predictors; $y$: response; $K$: number of categorical random-effect variables; $m=\sum_{k=1}^K m_k$: total random-effect dimension; $\operatorname{nnz}(Z^TZ)$: number of nonzero entries in $Z^TZ$. The categorical variables and their individual $m_k$ are reported in Table \ref{table:real_data_sets_detailed} in Appendix \ref{section:add_real_world}.}
% `\% 1' is the frequency of the 1's in the response variable for the classification data sets,
\label{table:real_data_sets}
\end{table}

% Non-zero entries in $Z^TZ$ occur if at least one observation belongs to a particular group level combination of any two grouped random effects. Thus, the random occurrence of non-zero entries in $Z^TZ$ in Figure \ref{fig:real_world_ZtZ} indicates random crossing among the random effects of data sets. For the `building\_permits' data set, building sites are visible in Figure \ref{fig:real_world_ZtZ} since the longitude and latitude variables are included into the random effect.

% \subsection{Results}
\rev{Figure \ref{fig:real_world_times} and Tables \ref{table:real_world_times}--\ref{table:real_world_coefficients} in Appendix \ref{section:add_real_world} report runtimes and estimates. Krylov methods are fastest on every data set and up to about $50$ times faster than same-backend Cholesky calculations on the data sets where Cholesky-based calculations do not crash. For the lower-dimensional regression data sets `cars' and `instEval', the differences among the methods are less pronounced, but Krylov methods still provide significant speed-ups. On the `MovieLens\_32m' data set, Krylov estimation takes under four minutes, \texttt{lme4} takes over 38 hours, and all other implementations crash. For the `upselling' data set, \texttt{glmmTMB} and \texttt{MixedModels.jl} return missing values, \texttt{R-INLA} crashes, \texttt{lme4} does not converge properly after more than 67 hours, whereas Krylov methods run in approximately 15 seconds. On the `employee\_access' data set, \texttt{MixedModels.jl} runs for more than 85 hours without fully converging, and the Krylov approach requires 30 seconds. The parameter estimates generally agree when the optimizers converge. Overall, Krylov methods are both substantially faster and more computationally robust.}
\begin{figure}[ht!]
    \centering
    \includegraphics[width=0.7\linewidth]{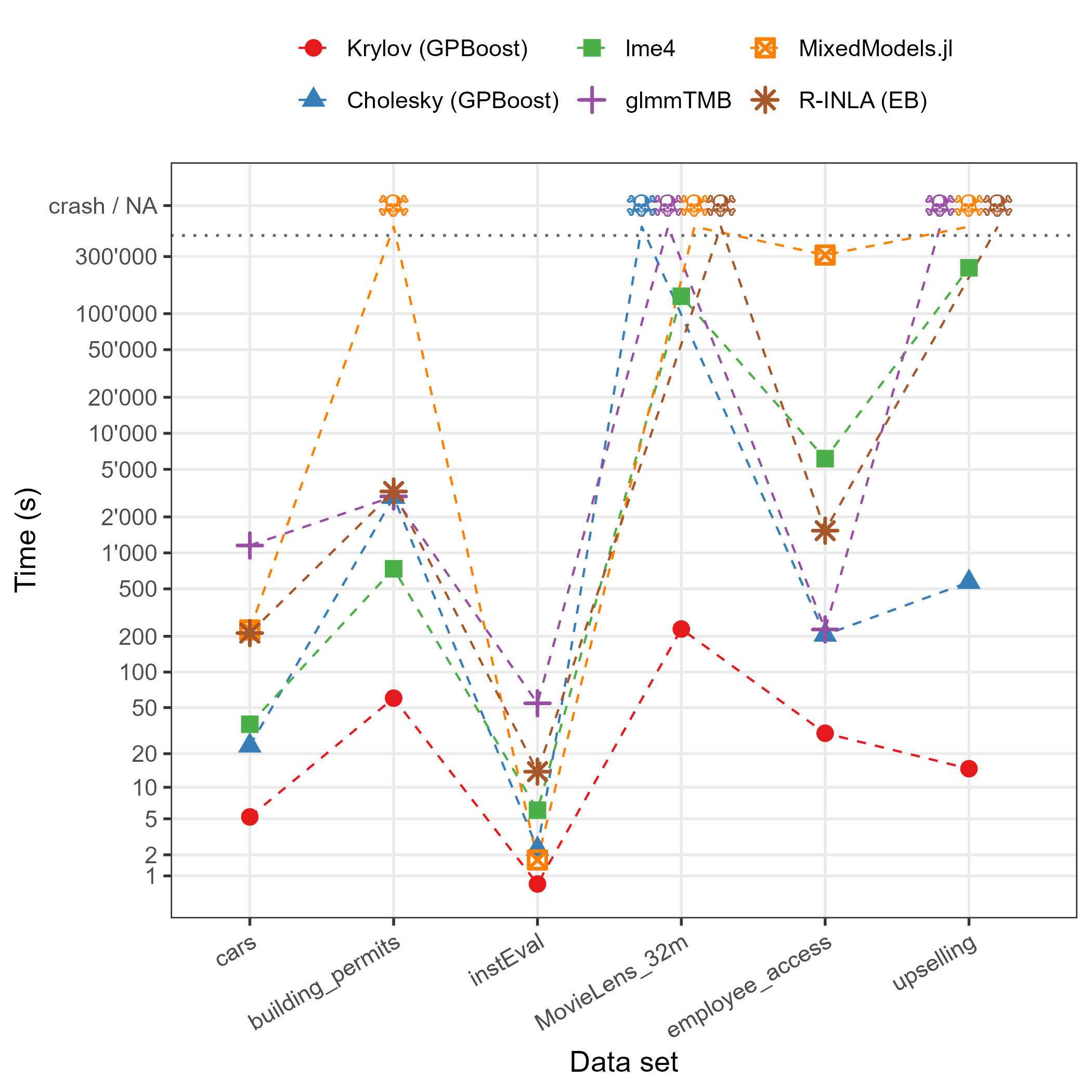}
    \caption{\rev{Parameter-estimation runtimes for the real-world data sets; symbols above the dotted line denote crashes or missing estimates.}}
    \label{fig:real_world_times}
\end{figure}

\section{Conclusion}
We present Krylov subspace methods for mixed-effects models with high-dimensional crossed random effects. We derive novel theoretical results on the convergence and accuracy of these methods, and we analyze them empirically. In experiments, we find that Krylov subspace-based methods are much faster and more computationally robust than Cholesky-based calculations and other state-of-the-art software implementations for GLMMs. \rev{For instance, we are able to estimate a model with two crossed random effects on the large-scale MovieLens\_32m data set in under four minutes, whereas \texttt{lme4} requires over 38 hours and Cholesky-based calculations in \texttt{GPBoost}, \texttt{glmmTMB}, \texttt{MixedModels.jl}, and \texttt{R-INLA} crash.}

\rev{While Krylov methods achieved faster runtimes than the Cholesky-based alternatives in all settings considered in our experiments, they are particularly attractive for high-dimensional crossed random effects for which Cholesky factorization results in substantial fill-in. Nevertheless, these results do not imply that Krylov methods are invariably optimal. For sufficiently small problems, sparse Cholesky methods may remain competitive, depending on the sparsity pattern, fill-in, and whether a factorization can be reused.} We use the Laplace approximation for non-Gaussian likelihoods following a common practice in widely used software implementations such as \texttt{lme4} and \texttt{glmmTMB}. A potential caveat of this is that, depending on the asymptotic setting for mixed-effects models, the Laplace approximation can be biased \citep{shun1995laplace}. Directions for future research include extending our methods to other marginal-likelihood and posterior approximations, analyzing alternative log-determinant approximations, systematically comparing the settings in which the ZIC and SSOR preconditioners are preferable, and developing novel preconditioners.

\if0\jasa{
\section*{Acknowledgments}
This research was partially supported by the Swiss Innovation Agency - Innosuisse (grant number `55463.1 IP-ICT').
}\fi

\if1\jasa{
\section*{Disclosure statement}
There are no competing interests to declare. The authors acknowledge the use of ChatGPT version 5 for language editing and assistance in improving clarity and structure.
}\fi

\bibliographystyle{abbrvnat}
\bibliography{bib_Iterative_GMMs}
\clearpage

% !TeX root = JCGS.tex

\appendix
\begin{bibunit}

% \begin{appendices}
\setcounter{figure}{0}
\renewcommand{\thefigure}{A\arabic{figure}}
\setcounter{table}{0}
\renewcommand{\thetable}{A\arabic{table}}

\if0\jasa{
\section*{Appendix}
}\fi

\section{Proofs of results in Section \ref{sec_convergence}}\label{sec_proofs}
\begin{notation}
We denote by $\sigma_{\min}(B)=\sigma_{\min(n_1,n_2)}(B)\leq \dots\leq \sigma_1(B) = \sigma_{\max}(B)$ the singular values of a matrix $B\in \mathbb{R}^{n_1\times n_2}$, and by $\lambda_{\min}(A)=\lambda_m(A)\leq \dots\leq \lambda_1(A) = \lambda_{\max}(A)$ the eigenvalues of a symmetric matrix $\vecmat{A}\in \mathbb{R}^{m\times m}$. Further, let 
$$D = \Sigma^{-1} + \text{diag}(Z^T W Z) \in \mathbb{R}^{m\times m},$$ 
$$\Delta = \text{diag}(Z^T W Z)\in \mathbb{R}^{m\times m},$$ 
and 
$$\Lambda = Z^T W Z - \text{diag}(Z^T W Z) \in \mathbb{R}^{m\times m}.$$
\end{notation}

We will use the following two lemmas for proving Lemma \ref{EVs_SSOR}. 
\begin{lemma}\label{evs_A_adjacency}
\rev{For any $B\in \mathbb{R}^{n_1\times n_2}$ and $A= \begin{pmatrix} 0  &  B^T\\  B & 0 \end{pmatrix}\in \mathbb{R}^{(n_1+n_2)\times (n_1+n_2)}$, the eigenvalues of $A$ consist of
$$\{\pm \sigma_i(B); i=1,\dots,\min(n_1,n_2)\}$$
and $|n_1-n_2|$ additional zero eigenvalues.}

In particular, the largest singular value of $B$ and the largest eigenvalue of $A$ are equal, $\lambda_{\max}(A)= \sigma_{\max}(B)$. Likewise, whenever it exists, the second-largest positive eigenvalue equals the second-largest singular value.

% The largest modulus singular value of any matrix $B\in \mathbb{R}^{n_1\times n_2}$ and the largest modulus eigenvalue of the matrix $A= \begin{pmatrix} 0  &  B^T\\  B & 0 \end{pmatrix} $ are equal, $|\lambda_{\max}(A)|= |\sigma_{\max}(B)|$.
\end{lemma}

\begin{proof}[Proof of Lemma \ref{evs_A_adjacency}]
The eigenvalues of $A$ are defined as solutions to
$$\text{det}(A - \lambda I)=0.$$
\rev{For $\lambda\neq 0$, the Schur-complement identity gives, if $n_1\geq n_2$,
$$
\det(A-\lambda I)=(-\lambda)^{n_1-n_2}
\det\left(\lambda^2 I_{n_2}-B^TB\right),
$$
whereas, if $n_2\geq n_1$,
$$
\det(A-\lambda I)=(-\lambda)^{n_2-n_1}
\det\left(\lambda^2 I_{n_1}-BB^T\right).
$$
Since both sides are polynomials in $\lambda$, these identities also extend to $\lambda=0$. Hence, the eigenvalues of $A$ consist of
$$
\{\pm\sigma_i(B); i=1,\dots,\min(n_1,n_2)\}
$$
and $|n_1-n_2|$ additional zero eigenvalues.}

% The matrix $A^2 = \begin{pmatrix} B^TB  &  0\\  0 & BB^T \end{pmatrix}$ has the eigenvalues 
% $$\{\lambda_i(B^TB), \lambda_j(BB^T); i = 1,\dots, \min(n_1,n_2), j = 1,\dots, \min(n_1,n_2) \} = \{\sigma_i(B)^2; i=1,\dots,\min(n_1,n_2)\}.$$
% It follows that $$|\lambda_{\max}(A)| = \sqrt{\lambda_{\max}(A^2)}= \sqrt{\sigma_{\max}(B)^2} = |\sigma_{\max}(B)|.$$
\end{proof}

\begin{lemma}\label{svs_diag}
For any matrix $B\in \mathbb{R}^{n_1\times n_2}$ and diagonal matrices $C_1^{n_1\times n_1}$ and $C_2^{n_2\times n_2}$, it holds that
$$\sigma_{\max}(B)\min(|(C_1)_{ii}|)\min(|(C_2)_{jj}|) \leq \sigma_{\max}(C_1BC_2) \leq \sigma_{\max}(B)\max(|(C_1)_{ii}|)\max(|(C_2)_{jj}|).
$$
\end{lemma}

\begin{proof}[Proof of Lemma \ref{svs_diag}]
$$\sigma_{\max}(C_1BC_2)=\|C_1BC_2\|_2\leq \|C_1\|_2\|B\|_2\|C_2\|_2 = \sigma_{\max}(B)\max(|(C_1)_{ii}|)\max(|(C_2)_{jj}|).$$
\rev{If either $C_1$ or $C_2$ has a zero diagonal entry, the lower bound is trivial. Otherwise, $C_1$ and $C_2$ are invertible, and
$$\sigma_{\max}(B) = \|B\|_2 = \|C_1^{-1}C_1BC_2C_2^{-1}\|_2 \leq \sigma_{\max}(C_1BC_2)\frac{1}{\min_i |(C_1)_{ii}|}\frac{1}{\min_j |(C_2)_{jj}|}$$
which gives the first inequality.}
\end{proof}

\begin{proof}[Proof of Lemma \ref{EVs_SSOR}]
For the SSOR preconditioner $P_{\text{SSOR}} = (L + D) D^{-1} (L + D)^T$, we have
$$L = \begin{pmatrix}
           0  &  0\\
            Z_2^T W Z_1 & 0
        \end{pmatrix}, ~~ D = \begin{pmatrix}
           D_1  &  0\\
            0 & D_2
        \end{pmatrix}, ~~ \Sigma^{-1} + Z^T W Z = \begin{pmatrix}
           D_1  & Z_1^T W Z_2\\
            Z_2^T W Z_1 & D_2
        \end{pmatrix}.$$
It can easily be seen that 
$$(D + L)^{-1} = \begin{pmatrix}
           D_1^{-1}       & 0\\
            -D_2^{-1}Z_2^T W Z_1 D_1^{-1} & D_2^{-1}
        \end{pmatrix}, $$
and we thus obtain
\begin{equation}\label{SSORE_pced_M}
    \begin{split}
        P_{\text{SSOR}}^{-1/2}(\Sigma^{-1} + Z^T W Z)P_{\text{SSOR}}^{-T/2} &= D^{1/2}(D + L)^{-1}(\Sigma^{-1} + Z^T W Z)(D + L)^{-T}D^{1/2}\\
        &= I_m - \begin{pmatrix}
           0       & 0\\
            0 & D_2^{-1/2}Z_2^T W Z_1 D_1^{-1}Z_1^T W Z_2D_2^{-1/2}
        \end{pmatrix}.
    \end{split}
\end{equation}
Because $D_2^{-1/2}Z_2^T W Z_1 D_1^{-1}Z_1^T W Z_2D_2^{-1/2}$ is a symmetric positive semi-definite matrix and using the fact that $\lambda_i(I_m + A) = 1 + \lambda_i(A)$ for any normal matrix $A$, we obtain \eqref{large_ev_ssor}.

Using Lemma \ref{evs_A_adjacency}, we have
\begin{equation*}
    \begin{split}
        \sigma_{\max}(D_2^{-1/2}Z_2^T W Z_1 D_1^{-1/2}) & = \lambda_{\max}\begin{pmatrix}
           0       & D_1^{-1/2}Z_1^T W Z_2 D_2^{-1/2}\\
            D_2^{-1/2}Z_2^T W Z_1 D_1^{-1/2} & 0
        \end{pmatrix}\\
        &=\lambda_{\max}(D^{-1/2}\Lambda D^{-1/2})\\
        &=\lambda_{\max}(D^{-1/2}\Delta^{1/2}\Delta^{-1/2}\Lambda\Delta^{-1/2}\Delta^{1/2} D^{-1/2}).
    \end{split}
\end{equation*}
Note that $\Delta^{-1/2}\Lambda\Delta^{-1/2}$ is a normalized adjacency matrix of a bipartite weighted graph. Its largest eigenvalue is thus $1$ \citep{chung1997spectral}. By Lemma \ref{svs_diag}, or Ostrowski's theorem \citep{horn2012matrix}, it follows that
$$
\min((D^{-1}\Delta)_{ii}) \leq  \sigma_{\max}(D_2^{-1/2}Z_2^T W Z_1 D_1^{-1/2}) \leq \max((D^{-1}\Delta)_{ii}).
$$
Since $D^{-1}\Delta = \text{diag}\left(\frac{1}{\Sigma^{-1}_{ii}\Delta^{-1}_{ii} + 1}\right)$ and 
\begin{equation}\label{ssor_eigen_singu}
    \lambda_{i}(D_2^{-1/2}Z_2^T W Z_1 D_1^{-1}Z_1^T W Z_2D_2^{-1/2}) = \sigma_{i}(D_2^{-1/2}Z_2^T W Z_1 D_1^{-1/2})^2, ~~1\leq i\leq\min(m_1,m_2),
\end{equation}
we have
$$
\frac{1}{(\max(\Sigma^{-1}_{ii}\Delta^{-1}_{ii}) + 1)^2} \leq  \lambda_{\max}(D_2^{-1/2}Z_2^T W Z_1 D_1^{-1}Z_1^T W Z_2D_2^{-1/2}) \leq \frac{1}{(\min(\Sigma^{-1}_{ii}\Delta^{-1}_{ii}) + 1)^2},
$$
from which follows \eqref{smallest_ev_ssor}. 

Furthermore, since 
\begin{equation*}
    \begin{split}
        \sigma_{\max}(Z_2^T W Z_1 ) & = \lambda_{\max}\begin{pmatrix}
           0       & Z_1^T W Z_2 \\
            Z_2^T W Z_1  & 0
        \end{pmatrix} 
        = \lambda_{\max}(\Lambda) 
        = \lambda_{\max}(\Delta^{1/2}\Delta^{-1/2}\Lambda\Delta^{-1/2}\Delta^{1/2}),
    \end{split}
\end{equation*}
and applying Lemma \ref{svs_diag}, we obtain 
$$ \min(\Delta_{ii})\leq \sigma_{\max}(Z_2^T W Z_1 )\leq \max(\Delta_{ii}).$$
By again using Lemma \ref{svs_diag}, we obtain
\begin{equation*}
    \begin{split}
        \min((D_1^{-1/2})_{ii})\min((D_2^{-1/2})_{ii})\min(\Delta_{ii})&\leq \sigma_{\max}(D_2^{-1/2}Z_2^T W Z_1 D_1^{-1/2})\\ &\leq \max((D_1^{-1/2})_{ii})\max((D_2^{-1/2})_{ii})\max(\Delta_{ii})
    \end{split}
\end{equation*}

and
\begin{equation*}
    \begin{split}
        \min((D_1)_{ii}^{-1})\min((D_2)_{ii}^{-1})\min(\Delta_{ii}^2) &\leq \lambda_{\max}(D_2^{-1/2}Z_2^T W Z_1 D_1^{-1}Z_1^T W Z_2D_2^{-1/2})\\ &\leq \max((D_1)_{ii}^{-1})\max((D_2)_{ii}^{-1})\max(\Delta_{ii}^2)
    \end{split}
\end{equation*}
which gives the second set of inequalities for $\lambda_{\min}^{\text{SSOR}}$ in \eqref{smallest_ev_ssor2}.
\end{proof}

\begin{proof}[Proof of Lemma \ref{EVs_Diag_no_PC}]
Consider the vectors 
$$v_k = D^{1/2}(\frac{1}{k}1_{m_1}^T,\dots,\frac{1}{k}1_{m_k}^T,-1_{m_{k+1}}^T,0,\dots,0)^T\in \mathbb{R}^m, ~~k=1,\dots,K-1,$$ 
where $1_{m_k}$ denotes a vector of ones of length $m_k$. We have $D^{-1/2}Z^T W ZD^{-1/2}v_k = 0$, and the vectors $v_k$ are linearly independent. Since $D^{-1/2}Z^T W ZD^{-1/2}$ is positive semi-definite, this shows that the $K-1$ smallest eigenvalues of $D^{-1/2}Z^T W ZD^{-1/2}$ are zero. Applying Weyl's inequality \citep{horn2012matrix} to $\rev{D^{-1/2}(\Sigma^{-1} + Z^T W Z)D^{-1/2}} = D^{-1/2}\Sigma^{-1}D^{-1/2} + D^{-1/2}Z^T W ZD^{-1/2}$ gives
$$\lambda_{\min}(D^{-1/2}\Sigma^{-1}D^{-1/2}) \leq \lambda_{m+1-k}(\rev{D^{-1/2}(\Sigma^{-1} + Z^T W Z)D^{-1/2}}) \leq \lambda_{\max}(D^{-1/2}\Sigma^{-1}D^{-1/2})$$
for $k=1,\dots,K-1$, which proves \eqref{smallest_ev_diag}. 

As in the proof of Lemma \ref{EVs_SSOR}, we can show, for general $K$, that
\begin{equation}\label{aux_ineq1}
    (K-1)\min((D^{-1}\Delta)_{ii}) \leq \lambda_{\max}(D^{-1/2}\Lambda D^{-1/2})\leq (K-1)\max((D^{-1}\Delta)_{ii}).
\end{equation}
For doing this, we first note that $\Lambda$ is an adjacency matrix of a $K$-partite weighted graph. Moreover,
\begin{equation*}
    \begin{split}
        \sum_{j=1}^m\Lambda_{ij} &= \sum_{k=1,k\neq\tilde k(i)}^K\sum_{j=1}^{m_k}(Z_{\tilde k(i)}^TWZ_k)_{ij}\\
        &=\sum_{k=1,k\neq\tilde k(i)}^K\sum_{j=1}^{m_k}\sum_{l=1}^n W_{ll}(Z_{\tilde k(i)})_{li}(Z_k)_{lj}\\
        &=\sum_{k=1,k\neq\tilde k(i)}^K\sum_{l=1}^nW_{ll}(Z_{\tilde k(i)})_{li}\\
        &=\sum_{k=1,k\neq\tilde k(i)}^K \Delta_{ii} = (K-1)\Delta_{ii}
    \end{split}
\end{equation*}
where $\tilde k(i) = k$ if $i\in \{\sum_{k'=0}^{k-1}m_{k'}+1,\dots,\sum_{k'=0}^{k}m_{k'}\}$ with the convention $m_0=0$. This shows that $\frac{1}{K-1}\Delta^{-1/2}\Lambda\Delta^{-1/2}$ is a normalized adjacency matrix of a $K$-partite weighted graph whose largest eigenvalue is $1$. Using
 $$\lambda_{\max}(D^{-1/2}\Lambda D^{-1/2})=(K-1)\lambda_{\max}(D^{-1/2}\Delta^{1/2}\frac{1}{K-1}\Delta^{-1/2}\Lambda\Delta^{-1/2}\Delta^{1/2} D^{-1/2})$$
 gives \eqref{aux_ineq1}, which together with
\begin{equation}\label{pced_diag}
    \rev{D^{-1/2}(\Sigma^{-1} + Z^T W Z)D^{-1/2}} = I_m + D^{-1/2} \Lambda D^{-1/2}
\end{equation}
proves the statement in \eqref{largest_ev_diag}.

Next, we have
\begin{equation*}
    \begin{split}
        \Sigma^{-1} + Z^TWZ = \Sigma^{-1} + \Delta^{1/2}(I_m + \Delta^{-1/2}\Lambda\Delta^{-1/2})\Delta^{1/2}.
    \end{split}
\end{equation*}
Applying Weyl's inequality \citep{horn2012matrix}, Ostrowski's theorem \citep{horn2012matrix}, and the fact that the largest eigenvalue of $\Delta^{-1/2}\Lambda\Delta^{-1/2}$ is $K-1$, see above, we obtain \eqref{largest_ev_no_pc}. \eqref{smallest_ev_no_pc} is shown analogously to \eqref{smallest_ev_diag} above.
\end{proof}

To prove parts of Theorem \ref{cond_number_gaussian}, we first derive the following slightly more general result.
\begin{theorem}\label{EVs_Diag_gaussian_balanced}
If the likelihood is Gaussian, it holds that
\begin{align}
&  1 + \frac{K-1}{\frac{\max(\Sigma^{-1})\sigma^2}{d_{\min}}+1} \leq  \lambda_{\max}^{\text{Diag}} \leq 1 + \frac{K-1}{\frac{\min(\Sigma^{-1})\sigma^2}{d_{\max}}+1},\label{largest_ev_diag_gauss} \\
& \frac{1}{\frac{\max(\Sigma)}{\sigma^2}d_{\max} + 1} \leq  \lambda_{m+1-k}^{\text{Diag}} \leq \frac{1}{\frac{\min(\Sigma)}{\sigma^2}d_{\min} + 1}, ~~ k=1,\dots,K-1,\label{smallest_ev_diag_gauss}
\end{align}
where \rev{$d_{\max} = \max_i(Z^T Z)_{ii}$ and $d_{\min} = \min_i(Z^T Z)_{ii}$}.

If, in addition, $K=2$ and the design of the random effects is balanced in the sense that every random effect occurs $d_k$ times for component $k=1,2$, i.e., $d_k = \frac{n}{m_k} = \sum_{i=1}^n (Z_k)_{ij}$ for all $1\leq j \leq m_k$, it holds that
\begin{equation}\label{largest_EV_diag_balanced}
    \lambda_{\max}^{\text{Diag}} = 1 + \frac{1}{\sqrt{\left(\frac{\sigma^2}{\sigma_1^2d_1}+ 1\right)\left(\frac{\sigma^2}{\sigma_2^2d_2}+ 1\right)}}
\end{equation}
\begin{equation}\label{smallesd_EV_diag_balanced}
    \lambda_{\min}^{\text{Diag}} = 1 - \frac{1}{\sqrt{\left(\frac{\sigma^2}{\sigma_1^2d_1}+ 1\right)\left(\frac{\sigma^2}{\sigma_2^2d_2}+ 1\right)}}
\end{equation}
\end{theorem}

\begin{proof}[Proof of Theorem \ref{EVs_Diag_gaussian_balanced}]
First, \eqref{largest_ev_diag_gauss} and \eqref{smallest_ev_diag_gauss} follow directly from \eqref{largest_ev_diag} and \eqref{smallest_ev_diag}. 

For deriving \eqref{largest_EV_diag_balanced}, note that 
\begin{equation*}
    \begin{split}
D^{-1/2}\Delta^{1/2} &= \text{diag}\left(\frac{1}{\Sigma^{-1}_{ii}\Delta^{-1}_{ii} + 1}\right)^{1/2}
\\ &=\text{diag}\left(\frac{1}{\sqrt{\frac{\sigma^2}{\sigma_1^2d_1}+ 1}},\dots,\frac{1}{\sqrt{\frac{\sigma^2}{\sigma_1^2d_1}+ 1}},\frac{1}{\sqrt{\frac{\sigma^2}{\sigma_2^2d_2}+ 1}},\dots,\frac{1}{\sqrt{\frac{\sigma^2}{\sigma_2^2d_2}+ 1}}\right),
    \end{split}
\end{equation*}
which implies 
\begin{equation}\label{pced_diag_min_1}
    D^{-1/2}\Delta^{1/2}\Delta^{-1/2}\Lambda\Delta^{-1/2}\Delta^{1/2} D^{-1/2}=\frac{1}{\sqrt{\frac{\sigma^2}{\sigma_1^2d_1}+ 1}}\frac{1}{\sqrt{\frac{\sigma^2}{\sigma_2^2d_2}+ 1}}\Delta^{-1/2}\Lambda\Delta^{-1/2}.
\end{equation}
Using $\lambda_{\max}(\Delta^{-1/2}\Lambda\Delta^{-1/2})=K-1$, see the proof of Lemma \ref{EVs_Diag_no_PC}, then gives
\begin{equation*}
    \begin{split}
        \lambda_{\max}(D^{-1/2}\Lambda D^{-1/2})  &=\lambda_{\max}(D^{-1/2}\Delta^{1/2}\Delta^{-1/2}\Lambda\Delta^{-1/2}\Delta^{1/2} D^{-1/2})\\
        &=\frac{K-1}{\sqrt{\frac{\sigma^2}{\sigma_1^2d_1}+ 1}\sqrt{\frac{\sigma^2}{\sigma_2^2d_2}+ 1}}.
    \end{split}
\end{equation*}
This and \eqref{pced_diag} prove \eqref{largest_EV_diag_balanced}.

By analogous arguments, again using \eqref{pced_diag_min_1} and $\lambda_{\min}(\Delta^{-1/2}\Lambda\Delta^{-1/2})=-1$ due to Lemma \ref{evs_A_adjacency}, we have
\begin{equation*}
        \lambda_{\min}(D^{-1/2}\Lambda D^{-1/2})  =-\frac{1}{\sqrt{\frac{\sigma^2}{\sigma_1^2d_1}+ 1}}\frac{1}{\sqrt{\frac{\sigma^2}{\sigma_2^2d_2}+ 1}},
\end{equation*}
which gives \eqref{smallesd_EV_diag_balanced}.
\end{proof}

\begin{proof}[Proof of Theorem \ref{cond_number_gaussian}]
First, \eqref{smallest_ev_ssor} and 
$$\min(\Sigma^{-1}_{ii}(Z^TWZ)_{ii}^{-1}) \geq \min(\Sigma^{-1})_{ii}\min((Z^TWZ)_{ii}^{-1})=\frac{\sigma^2}{\sigma^2_{\max}d_{\max}}$$ 
imply 
$$1 - \left(\frac{\sigma^2_{\max}d_{\max}}{\sigma^2+\sigma^2_{\max}d_{max}}\right)^2 \leq  \lambda_{\min}^{\text{SSOR}} \leq  1 - \left(\frac{\sigma^2_{\min}d_{\min}}{\sigma^2+\sigma^2_{\min}d_{\min}}\right)^2$$ which gives \eqref{cond_numb_ssor_gauss_non_balanced}. In addition,
$$\frac{1}{1 - \left(\frac{\sigma^2_{\max}d_{\max}}{\sigma^2+\sigma^2_{\max}d_{\max}}\right)^2}=\frac{1}{2}\frac{\sigma^2_{\max}d_{\max}}{\sigma^2} + \frac{3}{4} + \frac{1}{8\sigma^2_{\max}d_{\max}/\sigma^2 + 4}$$
yields the asymptotic result in \eqref{cond_numb_ssor_gauss_non_balanced}.

Furthermore, plugging in $D_k=\text{diag}(1/\sigma^2_k + d_k / \sigma^2)$, $k=1,2$, in \eqref{smallest_ev_ssor2} gives
\begin{equation*}
    \begin{split}
\frac{\lambda_{\max}^{\text{SSOR}}}{\lambda_{\min}^{\text{SSOR}}} &= \frac{1}{1 - \left(\frac{1}{\frac{\sigma^2}{\sigma^2_1d}+1}\right)\left(\frac{1}{\frac{\sigma^2}{\sigma^2_2d}+1}\right)} =\frac{1}{\frac{(\sigma^2)^2 + \sigma^2(\sigma^2_1+\sigma^2_2)d}{(\sigma^2)^2 + \sigma^2(\sigma^2_1+\sigma^2_2)d + \sigma^2_1\sigma^2_2d^2}} =\frac{(\sigma^2)^2/d + \sigma^2(\sigma^2_1+\sigma^2_2) + \sigma^2_1\sigma^2_2d}{(\sigma^2)^2/d + \sigma^2(\sigma^2_1+\sigma^2_2)},
    \end{split}
\end{equation*}
\rev{Expanding this expression gives
$$
\frac{\lambda_{\max}^{\text{SSOR}}}{\lambda_{\min}^{\text{SSOR}}}
=\frac{\sigma_1^2\sigma_2^2}{\sigma^2(\sigma_1^2+\sigma_2^2)}d
+\frac{\sigma_1^4+\sigma_1^2\sigma_2^2+\sigma_2^4}{(\sigma_1^2+\sigma_2^2)^2}+O(d^{-1}),
$$
which proves \eqref{cond_number_ssor_balanced}.}

Next, Theorem \ref{EVs_Diag_gaussian_balanced} implies 
$$\left(1 + \frac{1}{\frac{\sigma^2}{\sigma^2_{\min}d_{\min}}+1}\right)\left(\frac{\sigma^2_{\min}d_{\min}}{\sigma^2} + 1\right) \leq \frac{\lambda_{\max}^{\text{Diag}}}{\lambda_{\min}^{\text{Diag}}} \leq \left(1 + \frac{1}{\frac{\sigma^2}{\sigma^2_{\max}d_{\max}}+1} \right)\left(\frac{\sigma^2_{\max}d_{\max}}{\sigma^2} + 1\right),$$
which gives
\eqref{cond_numb_diag_gauss_non_balanced}. Similarly, by Theorem \ref{EVs_Diag_gaussian_balanced}, we have
\begin{equation*}
         \frac{\lambda_{\max}^{\text{Diag}}}{\lambda_{\min}^{\text{Diag}} } = \frac{\left(1 + \frac{1}{\sqrt{\left(\frac{\sigma^2}{\sigma_1^2d}+ 1\right)\left(\frac{\sigma^2}{\sigma_2^2d}+ 1\right)}}\right)}{\left(1 - \frac{1}{\sqrt{\left(\frac{\sigma^2}{\sigma_1^2d}+ 1\right)\left(\frac{\sigma^2}{\sigma_2^2d}+ 1\right)}}\right)} = \frac{\sqrt{\left(\frac{\sigma^2}{\sigma_1^2 d} + 1\right)\left(\frac{\sigma^2}{\sigma_2^2 d} + 1\right)} + 1}{\sqrt{\left(\frac{\sigma^2}{\sigma_1^2 d} + 1\right)\left(\frac{\sigma^2}{\sigma_2^2 d} + 1\right)} - 1},
\end{equation*} 
which is the first equality in \eqref{cond_number_diag_balanced}. Using 
$$
\sqrt{\left(\frac{\sigma^2}{\sigma_1^2 d} +1 \right) \left(\frac{\sigma^2}{\sigma_2^2 d} + 1\right)}
= 1 + \frac{1}{2} \left(\frac{\sigma^2}{\sigma_1^2 d} + \frac{\sigma^2}{\sigma_2^2 d} \right) \rev{-\frac{1}{8}\left(\frac{\sigma^2}{\sigma_1^2d}-\frac{\sigma^2}{\sigma_2^2d}\right)^2+O(d^{-3})},
$$
\rev{one obtains
$$
\frac{\lambda_{\max}^{\text{Diag}}}{\lambda_{\min}^{\text{Diag}}}
=\frac{4\sigma_1^2\sigma_2^2}{\sigma^2(\sigma_1^2+\sigma_2^2)}d
+\frac{2(\sigma_1^4+\sigma_2^4)}{(\sigma_1^2+\sigma_2^2)^2}+O(d^{-1}),
$$
which proves \eqref{cond_number_diag_balanced}.}

Finally, \eqref{none_cond_number} follows from Lemma \ref{EVs_Diag_no_PC}.
\end{proof}

% We have the following results for the extremal eigenvalues of the unpreconditioned matrix $\Sigma^{-1} + Z^T W Z$. 

\begin{proof}[Proof of Theorem \ref{cond_number_binary}]
We first note that 
\begin{equation}\label{ineq_aux_binary}
\begin{split}
\min(\Sigma^{-1}_{ii}(Z^TWZ)_{ii}^{-1}) &\geq \frac{1}{\max(\Sigma)_{ii}\max((Z^TWZ)_{ii})}\\
    & \geq \frac{1}{\max(\Sigma)_{ii}\max((Z^TZ)_{ii})\max(W)} \\
    & \geq \frac{1}{\sigma^2_{\max}d_{\max}C},
\end{split}
\end{equation}
since 
\begin{equation}\label{ineq_w_binary}
\max(W) \leq C,
\end{equation}where $C$ equals $0.25$ and $1$ for a logit and a probit link, respectively. This and \eqref{smallest_ev_ssor} imply
$$1 - \left(\frac{\sigma^2_{\max}d_{\max}C}{1+\sigma^2_{\max}d_{\max}C}\right)^2 \leq  \lambda_{\min}^{\text{SSOR}},$$ which gives the inequality in \eqref{cond_numb_ssor_binary_non_balanced}. As in the proof of Theorem \ref{cond_number_gaussian}, we have
$$\frac{1}{1 - \left(\frac{\sigma^2_{\max}d_{\max}C}{1+\sigma^2_{\max}d_{\max}C}\right)^2}=\frac{\sigma^2_{\max}C}{2}d_{\max} + \frac{3}{4} + \frac{1}{8\sigma^2_{\max}d_{\max}C + 4},$$
which yields the asymptotic result in \eqref{cond_numb_ssor_binary_non_balanced}.

For the diagonal preconditioner, \eqref{largest_ev_diag} and the inequality in \eqref{ineq_aux_binary} imply 
$$\lambda_{\max}^{\text{Diag}} \leq 1 + \frac{1}{\frac{1}{\sigma^2_{\max}d_{\max}C}+1}=\frac{2\sigma^2_{\max}d_{\max}C + 1}{\sigma^2_{\max}d_{\max}C+1},$$
and \eqref{smallest_ev_diag} and \eqref{ineq_w_binary} imply
$$ \frac{1}{1 + \sigma^2_{\max} d_{\max}C} \leq  \lambda_{\min}^{\text{Diag}}.$$
This gives \eqref{cond_numb_diag_binary_non_balanced}. For the unpreconditioned matrix, Lemma \ref{EVs_Diag_no_PC} and \eqref{ineq_w_binary} imply \eqref{none_cond_number_binary}.
\end{proof}

% \section{Effective condition numbers for the diagonal preconditioner and proof of Corollary \ref{eff_cond_SSOR_Diag_gaussian_balanced}}\label{ineq_cond_second_smalles_EV_Diag}
\begin{proof}[Proof of Theorem \ref{second_smallest_EV_SSOR_gaussian}]
As in the proof of Lemma \ref{EVs_SSOR}, we can apply Lemma \ref{evs_A_adjacency} and obtain
\begin{equation}\label{second_ev_intermed}
    \begin{split}
        \sigma_{2}\left(\frac{1}{\sigma^2}D_2^{-1/2}Z_2^T Z_1 D_1^{-1/2}\right) & =\lambda_{2}\left(\frac{1}{\sigma^2}D^{-1/2} \begin{pmatrix}
           0  & Z_1^T Z_2\\
            Z_2^T Z_1 & 0
        \end{pmatrix} D^{-1/2}\right)\\
        % & \leq \lambda_{2} \begin{pmatrix}
        %    0  & Z_1^T Z_2\\
        %     Z_2^T Z_1 & 0
        % \end{pmatrix}\max(\sigma^{-2}D^{-1})\\
        & = \lambda_{2} \begin{pmatrix}
           0  & Z_1^T Z_2\\
            Z_2^T Z_1 & 0
        \end{pmatrix}\frac{1}{\sqrt{\sigma^2/\sigma_1^2 + d_1}}\frac{1}{\sqrt{\sigma^2/\sigma_2^2 + d_2}},
    \end{split}
\end{equation}    
where we have used 
$$
\sigma^{-1}D^{-1/2} = \text{diag}\left(\frac{1}{\sqrt{\sigma^2/\sigma_1^2 + d_1}},\dots,\frac{1}{\sqrt{\sigma^2/\sigma_1^2 + d_1}},\frac{1}{\sqrt{\sigma^2/\sigma_2^2 + d_2}},\dots,\frac{1}{\sqrt{\sigma^2/\sigma_2^2 + d_2}}\right).
$$ 
\rev{Theorem 3.2 of \citet{brito2022spectral} is stated for fixed degrees $d_1$ and $d_2$ as the number of vertices tends to infinity, which is the regime considered here.} Applying it to $\lambda_{2} \begin{pmatrix} 0  & Z_1^T Z_2\\ Z_2^T Z_1 & 0 \end{pmatrix}$ then shows that
\begin{equation*}
\begin{split}
\sigma_{2}\left(\frac{1}{\sigma^2}D_2^{-1/2}Z_2^T Z_1 D_1^{-1/2}\right) &\leq \frac{\sqrt{d_1 - 1} + \sqrt{d_2 - 1}}{\sqrt{\sigma^2/\sigma_1^2 + d_1}\sqrt{\sigma^2/\sigma_2^2 + d_2}}+\epsilon_m'\\
&\leq \frac{1}{\sqrt{d_1}} + \frac{1}{\sqrt{d_2}} +\epsilon_m'
\end{split}
\end{equation*}
asymptotically almost surely with \rev{$\epsilon_m'\rightarrow 0$ as $\min(m_1,m_2)\rightarrow \infty$}. Equations \eqref{SSORE_pced_M} and \eqref{ssor_eigen_singu} then give
\begin{equation}\label{second_smalles_EV_ineq}
\begin{split}
    \lambda_{m-1}^{\text{SSOR}} &\geq 1 - \frac{(\sqrt{d_1 - 1} + \sqrt{d_2 - 1})^2}{(\sigma^2/\sigma_1^2 + d_1)(\sigma^2/\sigma_2^2 + d_2)} -\epsilon_m
    \\
    &\geq 1 - \left(\frac{1}{d_1} + \frac{1}{d_2} + \frac{2}{\sqrt{d_1d_2}}\right) - \epsilon_m. 
    \end{split}
\end{equation}
\rev{Provided the lower bound in \eqref{second_smalles_EV_ineq} is positive, taking its reciprocal gives \eqref{cond_second_smalles_EV_SSOR}.}
\end{proof}

\begin{theorem}\label{second_smallest_EVs_diag_gaussian}
If $K=2$, the likelihood is Gaussian, and $\begin{pmatrix} 0  & Z_1^T Z_2\\  Z_2^T Z_1 & 0 \end{pmatrix} \in \{0,1\}^{m\times m}$ is an adjacency matrix of a bipartite, biregular random graph with uniform distribution over all bipartite, biregular random graphs, the following holds:
\begin{equation}\label{cond_second_smalles_EV_Diag}
    \kappa_{m-1,2}^{\text{Diag}}=\frac{\lambda_{2}^{\text{Diag}}}{\lambda_{m-1}^{\text{Diag}}} \leq \frac{ 1 + \frac{1}{\sqrt{d_1}} + \frac{1}{\sqrt{d_2}} + \epsilon_m}{ 1 - \left(\frac{1}{\sqrt{d_1}} + \frac{1}{\sqrt{d_2}}\right) - \epsilon_m}
\end{equation}
asymptotically almost surely with \rev{$\epsilon_m\rightarrow 0$ as $\min(m_1,m_2)\rightarrow \infty$}, where $d_k = \frac{n}{m_k}$ for $k=1,2$. \rev{Here, the degrees $d_1$ and $d_2$ are held fixed as $\min(m_1,m_2)\rightarrow\infty$.}
\rev{The bound applies whenever the denominator in \eqref{cond_second_smalles_EV_Diag} is positive.}
\end{theorem}

\begin{proof}[Proof of Theorem \ref{second_smallest_EVs_diag_gaussian}]
We have $\lambda_{2}^{\text{Diag}} = 1 + \lambda_2(D^{-1/2} \Lambda D^{-1/2})$ because of \eqref{pced_diag}. Next, \eqref{second_ev_intermed} gives
$$
\lambda_2(D^{-1/2} \Lambda D^{-1/2}) = \lambda_{2} \begin{pmatrix}
           0  & Z_1^T Z_2\\
            Z_2^T Z_1 & 0
        \end{pmatrix}\frac{1}{\sqrt{\sigma^2/\sigma_1^2 + d_1}}\frac{1}{\sqrt{\sigma^2/\sigma_2^2 + d_2}}.$$
Analogously to the proof of Theorem \ref{second_smallest_EV_SSOR_gaussian}, applying Theorem 3.2 of \citet{brito2022spectral} to $\lambda_{2} \begin{pmatrix} 0  & Z_1^T Z_2\\ Z_2^T Z_1 & 0 \end{pmatrix}$ then shows that 
\begin{equation*}
\begin{split}
\lambda_{2}^{\text{Diag}} &\leq 1 + \frac{\sqrt{d_1 - 1} + \sqrt{d_2 - 1}}{\sqrt{\sigma^2/\sigma_1^2 + d_1}\sqrt{\sigma^2/\sigma_2^2 + d_2}}+\epsilon_m\\
&\leq 1 + \frac{1}{\sqrt{d_1}} + \frac{1}{\sqrt{d_2}} +\epsilon_m
\end{split}
\end{equation*}
asymptotically almost surely with \rev{$\epsilon_m\rightarrow 0$ as $\min(m_1,m_2)\rightarrow \infty$}. Lemma \ref{evs_A_adjacency} and analogous arguments give
\begin{equation*}
\lambda_{m-1}^{\text{Diag}} \geq 1 - \left(\frac{1}{\sqrt{d_1}} + \frac{1}{\sqrt{d_2}}\right) -\epsilon_m.
\end{equation*}
\rev{Provided the lower bound is positive, taking the ratio of the above two inequalities gives \eqref{cond_second_smalles_EV_Diag}.}
\end{proof}

\begin{proof}[Proof of Corollary \ref{eff_cond_SSOR_Diag_gaussian_balanced}]
\rev{The first part for fixed $d$ follows by setting $d_1=d_2=d$ in Theorems \ref{second_smallest_EV_SSOR_gaussian} and \ref{second_smallest_EVs_diag_gaussian}, which gives $\frac{1}{d_1} + \frac{1}{d_2} + \frac{2}{\sqrt{d_1d_2}} = \frac{4}{d}$ and $\frac{1}{\sqrt{d_1}} + \frac{1}{\sqrt{d_2}} = \frac{2}{\sqrt d}$, and thus \eqref{ineq_asym_eff_cond_ssor} and \eqref{ineq_asym_eff_cond_diag}. Note that $m_1 = m_2$ and $d_1 = d_2$ are equivalent here since $d_k = n/m_k$. The asymptotic result in the second part is a deterministic statement about the right-hand sides of \eqref{ineq_asym_eff_cond_ssor} and \eqref{ineq_asym_eff_cond_diag} with $\epsilon_m = \epsilon_m' = 0$. Applying the expansion $(1-x)^{-1}=1 + x + O(x^2)$ with $x=4/d$ gives the first expansion in \eqref{exp_eff_cond}. Similarly, $\frac{1+x}{1-x}=1+2x+O(x^2)$ with $x=2/\sqrt d$ gives the second. Both expansions are taken as $d\rightarrow\infty$ and hold for the fixed-$d$ bounds of part (i) separately for each $d$; no joint limit in $m$ and $d$ is claimed.}
\end{proof}

\section{Additional algorithms for predictive variances} \label{app_pred_var}

\subsection{Predictive variances using a stochastic estimator for the diagonal of a matrix}
We have also considered an algorithm similar to Algorithm \ref{alg:approach6} using the alternative version in \eqref{postpred_Laplace_cov1} for $\Omega_p$ instead of \eqref{postpred_Laplace_cov2}. Specifically, $\text{diag}(Z_{po} \Sigma Z^T \Psi^{-1} Z\Sigma Z_{po}^T)$ is approximated stochastically with the approach of \cite{bekas2007estimator} using the Woodbury identity, and all other terms are computed deterministically. However, experiments show that such an algorithm results in less accurate estimates compared to Algorithm \ref{alg:approach6} (results not shown).

\subsection{Predictive covariances by simulating from a multivariate Gaussian distribution}

Algorithm \ref{alg:approach5} presents an approach that allows for approximating predictive covariances $\Omega_p$ by sampling from a Gaussian distribution with covariance matrix $Z_{po} (\Sigma^{-1} + Z^TWZ)^{-1} Z_{po}^T$. The linear solves $(\Sigma^{-1} + Z^T W Z)^{-1} z_i^{(3)}$ can be computed efficiently using the preconditioned CG method. The linear solves and generation of the \(s\) predictive samples thus have complexity \(O(sn+sl(m+n)+sn_p)\). If the full covariance matrix is explicitly formed, the outer products additionally require \(O(sn_p^2)\) operations and \(O(n_p^2)\) storage. If only the predictive variances are calculated, as described below, this quadratic cost is avoided and the complexity of Algorithm \ref{alg:approach5} is \(O(sn+sl(m+n)+sn_p)\). This algorithm can also be trivially parallelized and results in an unbiased and consistent approximation for $\Omega_p$; see Appendix \ref{sec_prop_proofs} for a proof of Proposition \ref{prop_approach5}. In contrast to Algorithm \ref{alg:approach6}, this method can be used to calculate the entire covariance matrix and not just the variances on the diagonal. Furthermore, Algorithm \ref{alg:approach5} can be adapted to compute only the predictive variances $\text{diag}(\boldsymbol{\Omega}_p)$ by summing $\boldsymbol{z}_i^{(4)} \odot \boldsymbol{z}_i^{(4)}$ in Line \ref{line5}. We have also considered applying variance reduction to the estimator in Algorithm \ref{alg:approach5} by using a control variate based on $Z_{po}P^{-1}Z_{po}^T$. However, we found that the variance reduction was very small for this estimator and control variate since $\widehat{\text{Cov}}((z_i^{(4)} (z_i^{(4)})^T)_{jk}, (z_i^{(5)} (z_i^{(5)})^T)_{jk})$, $z_i^{(5)} = Z_{po} P^{-\frac{T}{2}} z_i^{(1)}$ and $j,k=1,2,\dots,n_p$, is often close to zero. \cite{kundig2024iterative} use a similar approach in Gaussian process regression.

\begin{algorithm}[ht!]
    \caption{Approximate predictive covariances using simulation}
    \label{alg:approach5}
    \begin{algorithmic}[1]
        \Require Matrices $Z_{po}, Z, \Sigma, Z_{pp}, \Sigma_p$, and $W$
        \Ensure Approximate predictive covariances $\hat{\Omega}_p$
        \For{$i \gets 1$ to $s$}
            \State{Sample $z_i^{(1)}\overset{\text{i.i.d.}}{\sim} \mathcal{N}(0, I_m)$ and $z_i^{(2)}\overset{\text{i.i.d.}}{\sim} \mathcal{N}(0, I_n)$ and set $z_i^{(3)} \gets \Sigma^{-\frac{1}{2}}z_i^{(1)} + Z^T W^{\frac{1}{2}} z_i^{(2)}$}
            \State{$z_i^{(4)} \gets Z_{po} (\Sigma^{-1} + Z^T W Z)^{-1} z_i^{(3)}$}
        \EndFor
        \State{$\hat{\Omega}_p \gets Z_{pp} \Sigma_p Z_{pp}^T + \frac{1}{s}\sum_{i=1}^s z_i^{(4)} \left(z_i^{(4)}\right)^T$}\label{line5}
    \end{algorithmic}
\end{algorithm}

\begin{proposition}\label{prop_approach5}
Algorithm \ref{alg:approach5} produces an unbiased and consistent estimator $\hat{\Omega}_p$ for the predictive covariances $\Omega_p$ given in \eqref{postpred_Laplace_cov2}.
\end{proposition}

As an alternative to Algorithm \ref{alg:approach5}, we can use the expression in \eqref{postpred_Laplace_cov1} and approximate $Z_{po} \Sigma Z^T \Psi^{-1} Z\Sigma Z_{po}^T$ using simulation. Algorithm \ref{alg:approach2} presents this approach. The linear solve $\Psi^{-1} z_i^{(3)}$ can be computed efficiently using the Woodbury identity and the preconditioned CG method. If only the predictive variances are calculated, the computational complexity of Algorithm \ref{alg:approach2} is thus $O(sn + sl(m+n) + sn_p)$. This algorithm also results in an unbiased and consistent approximation for $\Omega_p$; see Appendix \ref{sec_prop_proofs} for a proof of Proposition \ref{prop_approach2}.

\begin{proposition}\label{prop_approach2}
Algorithm \ref{alg:approach2} produces an unbiased and consistent estimator $\hat{\Omega}_p$ for the predictive covariances $\Omega_p$ given in \eqref{postpred_Laplace_cov1}.
\end{proposition}
%Variance reduction is also not effective for alg:approach2.
\begin{algorithm}[ht!]
    \caption{Approximate predictive covariances using simulation}
    \label{alg:approach2}
    \begin{algorithmic}[1]
        \Require Matrices $Z_{po}, Z, \Sigma, Z_{pp}, \Sigma_p$, and $W$
         \Ensure Approximate predictive covariances $\hat{\Omega}_p$
        \For{$i \gets 1$ to $s$}
            \State{Sample $z_i^{(1)}\overset{\text{i.i.d.}}{\sim} \mathcal{N}(0, I_m)$ and $z_i^{(2)}\overset{\text{i.i.d.}}{\sim} \mathcal{N}(0, I_n)$ and set $z_i^{(3)} \gets Z\Sigma^{\frac{1}{2}}z_i^{(1)} + W^{-\frac{1}{2}} z_i^{(2)}$}
            \State{$z_i^{(4)} \gets Z_{po}\Sigma Z^T\Psi^{-1} z_i^{(3)}$}
        \EndFor
        \State{$\hat{\Omega}_p \gets Z_{po} \Sigma Z_{po}^T + Z_{pp} \Sigma_p Z_{pp}^T - \frac{1}{s}\sum_{i=1}^s z_i^{(4)} \left(z_i^{(4)}\right)^T$}
    \end{algorithmic}
\end{algorithm}

We have also considered another version of such an algorithm as in Algorithms \ref{alg:approach5} and \ref{alg:approach2}. First, applying the Woodbury identity in \eqref{postpred_Laplace_cov1}, we obtain
\begin{equation}\label{latentpostpredictiveWoodbury}
\begin{split}
Z_{po} \Sigma Z^T \Psi^{-1} Z\Sigma Z_{po}^T 
=  & Z_{po} \Sigma Z^T W Z \Sigma Z_{po}^T  - Z_{po} \Sigma Z^TWZ(\Sigma^{-1} + Z^TWZ)^{-1}Z^TW Z \Sigma Z_{po}^T.     
\end{split}
\end{equation}
We can then approximate the last term in \eqref{latentpostpredictiveWoodbury} by simulation, and all other terms in \eqref{postpred_Laplace_cov1} are calculated deterministically. However, we found that this algorithm leads to less accurate estimates than Algorithms \ref{alg:approach5} and \ref{alg:approach2} (results not shown).

\subsection{Predictive variances using preconditioned Lanczos algorithms} \label{section:Lanczos_pred}
\citet{pleiss2018constant} use the Lanczos algorithm to approximate predictive covariance matrices for Gaussian process regression. We adapt and extend this approach for generalized mixed-effects models. Specifically, we approximate $\text{diag}(Z_{po} \Sigma Z^TWZ(\Sigma^{-1} + Z^TWZ)^{-1}Z^TW Z \Sigma Z_{po}^T)$ in \eqref{latentpostpredictiveWoodbury} by running the Lanczos algorithm with the matrix $\Sigma^{-1} + Z^TWZ$ using the normalized average of the column vectors of $Z^T W Z \Sigma Z_{po}^T$ as initial value and approximate $(\Sigma^{-1} + Z^TWZ)^{-1}Z^TW Z\Sigma Z_{po}^T \approx \tilde{Q}\tilde{T}^{-1} \tilde{Q}^{T} Z^T W Z \Sigma Z_{po}^T$, where $\tilde{Q}\tilde{T} \tilde{Q}^{T} \approx \Sigma^{-1} + Z^TWZ$ is obtained after $k$ steps of the Lanczos algorithm, $\tilde{\vecmat{Q}} \in \mathbb{R}^{m\times k}$ has orthonormal columns, and $\tilde{\vecmat{T}} \in \mathbb{R}^{k\times k}$ is tridiagonal. $\text{diag}(Z_{po} \Sigma Z^T W Z \Sigma Z_{po}^T)$ and the other terms in \eqref{postpred_Laplace_cov1} are deterministically computed, and predictive variances are thus approximated as
\begin{equation}
\begin{split}
\text{diag}(\hat{\Omega}_p) 
\approx\; & \text{diag}(Z_{po} \Sigma Z_{po}^T)
+ \text{diag}(Z_{pp} \Sigma_p Z_{pp}^T) 
- \text{diag}(Z_{po} \Sigma Z^T W Z \Sigma Z_{po}^T) \\
& + \text{diag}(Z_{po} \Sigma Z^TWZ\tilde{Q}\tilde{T}^{-1}\tilde{Q}^{T}Z^TW Z \Sigma Z_{po}^T).    
\end{split}
\end{equation}
The computation of a rank $k$ Lanczos approximation requires one matrix-vector multiplication with $\Sigma^{-1} + Z^TWZ$ in each iteration. Since the Lanczos algorithm suffers from numerical stability issues due to loss of orthogonality, we use a full reorthogonalization scheme \citep{wu2000thick}. Computing the Lanczos approximation then has complexity $O(k(m+n)+k^2m)$, where the first term accounts for the matrix-vector multiplications and the second for full reorthogonalization. Furthermore, we consider a preconditioning approach similar to that of \cite{kundig2024iterative} by applying the partial Lanczos algorithm to $P^{-\frac{1}{2}}(\Sigma^{-1} + Z^TWZ)P^{-\frac{T}{2}} \approx \tilde{Q}\tilde{T}\rev{\tilde{Q}^{T}}$ and correcting for this by using $P^{-\frac{T}{2}}\tilde{Q}\tilde{T}^{-1}\rev{\tilde{Q}^{T}}P^{-\frac{1}{2}}Z^TW Z\Sigma Z_{po}^T \approx (\Sigma^{-1} + Z^TWZ)^{-1}Z^TW Z\Sigma Z_{po}^T $ as approximation. Note that the Lanczos algorithm requires a symmetric matrix and in contrast to the preconditioned CG method, we need to explicitly calculate a factor $P^{-\frac{1}{2}}$, where $P=P^{\frac{1}{2}}P^{\frac{T}{2}}$. This is possible for the diagonal, SSOR, and ZIC preconditioners. However, in our experiments, this form of preconditioning does not lead to more accurate predictive variance estimates (results not shown).

\section{Proofs of Propositions \ref{prop_approach6}, \ref{prop_approach5}, and \ref{prop_approach2}} \label{sec_prop_proofs}

\begin{proof}[Proof of Proposition \ref{prop_approach6}]
\rev{Let $\bar h_s=s^{-1}\sum_{i=1}^s h(z_i)$ and $\bar r_s=s^{-1}\sum_{i=1}^s r(z_i)$. By the law of large numbers,
\[
\bar h_s \xrightarrow{p} \operatorname{diag}\!\left(Z_{po}(\Sigma^{-1}+Z^TWZ)^{-1}Z_{po}^T\right)
\quad\text{and}\quad
\bar r_s \xrightarrow{p} \operatorname{diag}\!\left(Z_{po}P^{-1}Z_{po}^T\right).
\]
Algorithm~\ref{alg:approach6} sets the weight to zero whenever the sample variance is zero. For every component for which the population variance of $r(z_i)$ is positive, the sample variance is positive with probability tending to one, and the sample covariance and variance converge to their population values; hence, the estimated control-variate weight converges in probability to a finite constant. If a component of $r(z_i)$ has zero variance, its sample variance and control-variate correction are identically zero. Consequently,
\[
c \odot \left(\operatorname{diag}\!\left(Z_{po}P^{-1}Z_{po}^T\right)-\bar r_s\right) \xrightarrow{p} 0.
\]
Combining these results with the deterministic term $\operatorname{diag}(Z_{pp}\Sigma_pZ_{pp}^T)$ proves the consistency claimed in Proposition~\ref{prop_approach6}.}
\end{proof}

\begin{proof}[Proof of Proposition \ref{prop_approach5}]
First, observe that $z_i^{(3)} = \Sigma^{-\frac{1}{2}}z_i^{(1)} + Z^T W^{\frac{1}{2}} z_i^{(2)} \sim \mathcal{N}(\vecmat{0}, (\Sigma^{-1} + Z^T W Z))$. It follows that $z_i^{(4)} = Z_{po} (\Sigma^{-1} + Z^T W Z)^{-1} z_i^{(3)} \sim \mathcal{N}(\vecmat{0}, Z_{po} (\Sigma^{-1} + Z^T W Z)^{-1} Z_{po}^T)$.
By standard results, $\frac{1}{s}\sum_{i=1}^s z_i^{(4)} \left(z_i^{(4)}\right)^T$ in Algorithm \ref{alg:approach5} is an unbiased and consistent estimator for $Z_{po} (\Sigma^{-1} + Z^T W Z)^{-1} Z_{po}^T$, and the claim in Proposition \ref{prop_approach5} thus follows.
\end{proof}

\begin{proof}[Proof of Proposition \ref{prop_approach2}]
First, observe that $z_i^{(3)} = Z\Sigma^{\frac{1}{2}}z_i^{(1)} + W^{-\frac{1}{2}} z_i^{(2)} \sim \mathcal{N}(\vecmat{0}, \Psi)$. It follows that $z_i^{(4)} = Z_{po}\Sigma Z^T\Psi^{-1} z_i^{(3)} \sim \mathcal{N}(\vecmat{0}, Z_{po}\Sigma Z^T\Psi^{-1} Z \Sigma Z_{po}^T)$. 
By standard results, $\frac{1}{s}\sum_{i=1}^s z_i^{(4)} \left(z_i^{(4)}\right)^T$ in Algorithm \ref{alg:approach2} is an unbiased and consistent estimator for $Z_{po}\Sigma Z^T\Psi^{-1} Z \Sigma Z_{po}^T$, and the claim in Proposition \ref{prop_approach2} thus follows.
\end{proof}

\section{Derivatives of log-determinants using stochastic trace estimation and variance reduction for the SSOR preconditioner} \label{PAppendix}

In the following, $c=\widehat{\text{Cov}}(h(z_i),r(z_i))/\widehat{\text{Var}}(r(z_i))$ is the sample estimate of the optimal control-variate weight.

\subsection{Derivative for variance parameters}
Note that
\[
\frac{\partial D}{\partial\theta_k}
= \frac{\partial \Sigma^{-1}}{\partial\theta_k}
= -\Sigma^{-1}\frac{\partial\Sigma}{\partial\theta_k}\Sigma^{-1}.
\]
\begin{equation*}
    \frac{\partial\log\det(\Sigma^{-1} + Z^T W Z)}{\partial\theta_k} 
    \approx
    c \; \underbrace{\frac{\partial\log\det(P_{\text{SSOR}})}{\partial\theta_k}}_{\text{deterministic}}
    + \frac{\partial\log\det(\Sigma^{-1} + Z^T W Z)}{\partial\theta_k}
    - c \; \underbrace{\frac{\partial\log\det(P_{\text{SSOR}})}{\partial\theta_k}}_{\text{stochastic}} 
    % \frac{\partial}{\partial\theta_k} \left(
    % c \; \underbrace{\log\det(P_{\text{SSOR}})}_{\text{deterministic}} + \log\det(\Sigma^{-1} + Z^T W Z) - c\;\underbrace{\log\det(P_{\text{SSOR}}}_{\text{stochastic}})
    % \right)
\end{equation*}

\begin{equation*}
\begin{split}
    \underbrace{\frac{\partial\log\det(P_{\text{SSOR}})}{\partial\theta_k}}_{\text{deterministic}} 
    &= \frac{\partial\left(\log\det(L+D)+\log\det(D^{-1})+\log\det(L+D)\right)}{\partial\theta_k} \\
    &= \frac{\partial \log\det(D)}{\partial\theta_k} \\
    &= \tr\left(D^{-1}\frac{\partial D}{\partial\theta_k}\right)\\
    &= -\tr\left(D^{-1}\Sigma^{-1}\frac{\partial\Sigma}{\partial\theta_k}\Sigma^{-1}\right)
\end{split}
\end{equation*}

\begin{equation*}
\begin{split}
    \frac{\partial\log\det(\Sigma^{-1} + Z^T W Z)}{\partial\theta_k} 
    =& \tr\left((\Sigma^{-1} + Z^T W Z)^{-1}\frac{\partial(\Sigma^{-1} + Z^T W Z)}{\partial\theta_k}\right)\\
    =& \tr\left((\Sigma^{-1} + Z^T W Z)^{-1}\frac{\partial\Sigma^{-1}}{\partial\theta_k}\right)\\
    =& -\tr\left((\Sigma^{-1} + Z^T W Z)^{-1}\Sigma^{-1}
    \frac{\partial\Sigma}{\partial\theta_k}\Sigma^{-1}\right)\\
    \approx& \frac{1}{t}\sum_{i=1}^t \underbrace{-\left((\Sigma^{-1} + Z^T W Z)^{-1}z_i\right)^T
    \Sigma^{-1}\frac{\partial\Sigma}{\partial\theta_k}\Sigma^{-1}P_{\text{SSOR}}^{-1} z_i}_{=:h(z_i)}
\end{split}
\end{equation*}

\begin{equation*}
\begin{split}
    \underbrace{\frac{\partial\log\det(P_{\text{SSOR}})}{\partial\theta_k}}_{\text{stochastic}} 
    =& \tr\left(P_{\text{SSOR}}^{-1}\frac{\partial P_{\text{SSOR}}}{\partial\theta_k}\right)\\
    =& \tr\left(P_{\text{SSOR}}^{-1} \frac{\partial (L + D) D^{-1} (L + D)^T}{\partial\theta_k}\right)\\
    =& \tr\left(P_{\text{SSOR}}^{-1} \left(-2\; \Sigma^{-1}\frac{\partial\Sigma}{\partial\theta_k}\Sigma^{-1}D^{-1}(L + D)^T %since tr(A) = tr(A^T)
    + (L + D) D^{-1}\Sigma^{-1}\frac{\partial\Sigma}{\partial\theta_k}\Sigma^{-1}D^{-1}(L + D)^T\right)\right)\\
    \approx& \frac{1}{t}\sum_{i=1}^t 
    \underbrace{\left(P_{\text{SSOR}}^{-1}z_i\right)^T
                                     \left(-2\; \Sigma^{-1}\frac{\partial\Sigma}{\partial\theta_k}\Sigma^{-1}D^{-1}(L + D)^T
                                     + (L + D) D^{-1}\Sigma^{-1}\frac{\partial\Sigma}{\partial\theta_k}\Sigma^{-1}D^{-1}(L + D)^T\right)
                                     P_{\text{SSOR}}^{-1}z_i}_{=:r(z_i)}
\end{split}
\end{equation*}

\subsection{Further derivatives for Laplace approximations}
For efficient calculations of the derivatives of log-determinants with STE, note that $\frac{\partial\log\det(\Sigma^{-1} + Z^T W Z)}{\partial F_i} = \frac{\partial\log\det(\Sigma^{-1} + Z^T W Z)}{\partial \mu^*_i}$ and $\frac{\partial\log\det(\Sigma^{-1} + Z^T W Z)}{\partial b^*_j} = \left(\frac{\partial\log\det(\Sigma^{-1} + Z^T W Z)}{\partial\mu^*}\right)^T\frac{\partial\mu^*}{\partial b^*_j} = \left(\frac{\partial\log\det(\Sigma^{-1} + Z^T W Z)}{\partial\mu^*}\right)^T Z_j$. 
Therefore, we present in the following calculations for $\frac{\partial\log\det(\Sigma^{-1} + Z^T W Z)}{\partial \mu^*_i}$, and $\frac{\partial\log\det(\Sigma^{-1} + Z^T W Z)}{\partial \xi_l}$ is obtained by replacing
$$
\frac{\partial W}{\partial \mu_i^*}
=-\left.\frac{\partial^3\log p(y_i\mid\mu_i,\xi)}{\partial\mu_i^3}\right|_{\mu_i=\mu_i^*}e_i e_i^T
$$
with $\frac{\partial W}{\partial \xi_l}=\rev{\text{diag}(-\frac{\partial^3 \log p(y_a| \mu^*_a,\xi)}{\partial{\mu^*_a}^2\partial\xi_l})}$, where $e_i$ is the $i$th standard basis vector in $\mathbb{R}^n$.
Further, $\frac{\partial D}{\partial \mu^*_i}$ is a diagonal matrix with diagonal entries $\rev{\left(\frac{\partial D}{\partial \mu^*_i}\right)_{aa} = \left(Z^T\frac{\partial W}{\partial \mu^*_i}Z\right)_{aa}}$, and $\frac{\partial(L+D)}{\partial \mu^*_i}$ is a lower-triangular matrix with entries
$\rev{\left(\frac{\partial(L+D)}{\partial \mu^*_i}\right)_{ab} = \mathbf{1}_{\{a\geq b\}}\left(Z^T \frac{\partial W}{\partial \mu^*_i} Z\right)_{ab}}$.

\begin{equation*}
    \frac{\partial\log\det(\Sigma^{-1} + Z^T W Z)}{\partial \mu^*_i} 
    \approx
    c \; \underbrace{\frac{\partial\log\det(P_{\text{SSOR}})}{\partial \mu^*_i}}_{\text{deterministic}}
    + \frac{\partial\log\det(\Sigma^{-1} + Z^T W Z)}{\partial \mu^*_i}
    - c \; \underbrace{\frac{\partial\log\det(P_{\text{SSOR}})}{\partial \mu^*_i}}_{\text{stochastic}}
\end{equation*}

\begin{equation*}
\begin{split}
    \underbrace{\frac{\partial\log\det(P_{\text{SSOR}})}{\partial \mu^*_i}}_{\text{deterministic}} 
    =& \frac{\partial\left(\log\det(L+D)+\log\det(D^{-1})+\log\det(L+D)\right)}{\partial \mu^*_i} \\
    =& \frac{\partial\log\det(D)}{\partial \mu^*_i} \\
    =& \tr\left(D^{-1}\frac{\partial D}{\partial \mu^*_i}\right)
\end{split}
\end{equation*}

\begin{equation*}
\begin{split}
    \frac{\partial\log\det(\Sigma^{-1} + Z^T W Z)}{\partial \mu^*_i}
    =& \tr\left((\Sigma^{-1} + Z^T W Z)^{-1}\frac{\partial(\Sigma^{-1} + Z^T W Z)}{\partial \mu^*_i}\right)\\
    =& \tr\left((\Sigma^{-1} + Z^T W Z)^{-1}Z^T\frac{\partial W}{\partial \mu^*_i}Z\right)\\
    \approx& \frac{1}{t}\rev{\sum_{r=1}^t} \underbrace{\left((\Sigma^{-1} + Z^T W Z)^{-1}\rev{z_r}\right)^TZ^T\frac{\partial W}{\partial \mu^*_i}Z P_{\text{SSOR}}^{-1} \rev{z_r}}_{=:h(\rev{z_r})}
\end{split}
\end{equation*}

\begin{equation*}
\begin{split}
    \underbrace{\frac{\partial\log\det(P_{\text{SSOR}})}{\partial \mu^*_i}}_{\text{stochastic}}
    =& \tr\left(P_{\text{SSOR}}^{-1}\frac{\partial P_{\text{SSOR}}}{\partial \mu^*_i}\right)\\
    =& \tr\left(P_{\text{SSOR}}^{-1} \frac{\partial (L + D) D^{-1} (L + D)^T}{\partial \mu^*_i}\right)\\
    % =& \tr\left(P_{\text{SSOR}}^{-1} \left(\frac{\partial(L+D)}{\partial \mu^*_i} D^{-1}(L + D)^T - 
    %                                    (L + D)D^{-1} \frac{\partial D}{\partial \mu^*_i}D^{-1}(L + D)^T \right.\right.\\
    % & \left.\left.+ (L + D) D^{-1} \frac{\partial (L+D)^T}{\partial \mu^*_i} \right)\right)\\
    =& \tr\left(P_{\text{SSOR}}^{-1} \left(2\;\frac{\partial(L+D)}{\partial \mu^*_i}D^{-1}(L + D)^T 
    - (L + D)D^{-1}\frac{\partial D}{\partial \mu^*_i}D^{-1}(L + D)^T \right)\right)\\
    \approx& \frac{1}{t}\rev{\sum_{r=1}^t} \underbrace{\left(P_{\text{SSOR}}^{-1}\rev{z_r}\right)^T
    \left(2\;\frac{\partial(L+D)}{\partial \mu^*_i}D^{-1}(L + D)^T 
    - (L + D)D^{-1}\frac{\partial D}{\partial \mu^*_i}D^{-1}(L + D)^T \right)
    P_{\text{SSOR}}^{-1}\rev{z_r}}_{=:r(\rev{z_r})}
\end{split}
\end{equation*}

\section{Alternative stochastic trace estimation for Fisher information} \label{STE_FisherI_alternative}
For $1\leq k,l\leq K$ and $z_i\sim \mathcal{N}(0,P)$, the trace terms of the Fisher information in \eqref{FisherI} can alternatively be computed with STE as follows:
\begin{equation}
\begin{split}
    &\tr\left(\Psi^{-1}\frac{\partial\Psi}{\partial\theta_k}\Psi^{-1}\frac{\partial\Psi}{\partial\theta_l}\right)\\
    =&\tr     \left(\left(W-WZ(\Sigma^{-1} + Z^TWZ)^{-1}Z^TW\right)Z_kZ_k^T\right.\\
     &\left.~~~~    \left(W-WZ(\Sigma^{-1} + Z^TWZ)^{-1}Z^TW\right)Z_lZ_l^T\right)\\
    =&  \tr\left(WZ_kZ_k^TWZ_lZ_l^T\right) \\
    & - 2 \tr\left((\Sigma^{-1} + Z^TWZ)^{-1}Z^TWZ_kZ_k^TWZ_lZ_l^T{WZ}\right) \\
    & + ~~\tr\left((\Sigma^{-1} + Z^TWZ)^{-1}Z^TWZ_kZ_k^TWZ(\Sigma^{-1} + Z^TWZ)^{-1}Z^TWZ_lZ_l^T{WZ}\right)\\
    \approx& \tr\left(WZ_kZ_k^TWZ_lZ_l^T\right) \\
    & - 2 \frac{1}{t}\sum_{i=1}^t \left((\Sigma^{-1} + Z^T W Z)^{-1}z_i\right)^TZ^TWZ_kZ_k^TWZ_lZ_l^T{WZ}P^{-1}z_i\\
    & + ~~\frac{1}{t}\sum_{i=1}^t \left((\Sigma^{-1} + Z^T W Z)^{-1}z_i\right)^T Z^TWZ_kZ_k^TWZ(\Sigma^{-1} + Z^TWZ)^{-1}Z^TWZ_lZ_l^T{WZ}P^{-1}z_i.
\end{split}    
\end{equation}

% \section{Derivation for the posterior predictive covariance matrix} \label{postpredictiveDerivation}
% The identities for the posterior predictive covariance matrix can be derived as follows, where in the second last line we use the Woodbury identity:
% \begin{equation}
%     \begin{split}
%         \Omega_p &= Z_{po} \Sigma Z_{po}^T + Z_{pp} \Sigma_p Z_{pp}^T - Z_{po} \Sigma Z^T \Psi^{-1} Z\Sigma Z_{po}^T \\
%                  &= Z_{po} \Sigma Z_{po}^T + Z_{pp} \Sigma_p Z_{pp}^T
%                     - Z_{po}\Sigma Z^T
%                     \left(Z\Sigma Z^T + W^{-1}\right)^{-1}
%                     Z \Sigma Z_{po}^T \\
%                  &= Z_{po} \Sigma Z_{po}^T + Z_{pp} \Sigma_p Z_{pp}^T
%                     - Z_{po}\Sigma 
%                     \left(\Sigma + (Z^TWZ)^{-1}\right)^{-1}
%                     \Sigma Z_{po}^T \\
%                  &= Z_{po} \Sigma Z_{po}^T + Z_{pp} \Sigma_p Z_{pp}^T
%                      - Z_{po}\Sigma 
%                      \left(\Sigma^{-1} - \Sigma^{-1}\left(\Sigma^{-1} + Z^TWZ\right)^{-1}\Sigma^{-1} \right) 
%                      \Sigma Z_{po}^T \\
%                  &= Z_{pp} \Sigma_p Z_{pp}^T + Z_{po} \left(\Sigma^{-1} + Z^TWZ\right)^{-1} Z_{po}^T                         
%     \end{split}
% \end{equation}

\clearpage
\section{Preconditioned conjugate gradient algorithm and zero fill-in incomplete Cholesky factorization}\label{appendix:CGalgo}
\begin{algorithm}[ht!]
    \caption{Preconditioned conjugate gradient algorithm with Lanczos tridiagonal matrix}
    \begin{algorithmic}[1]
        \Require Matrix $A$, preconditioner matrix $P$, vector $b$, absolute residual tolerance $\epsilon$
        \Ensure $u_{l} \approx A^{-1}b$, tridiagonal matrix $\tilde{T}$
        \State{early-stopping $\gets$ false}
        \State{$\alpha_0 \gets 1$}
        \State{$\beta_0 \gets 0$}
        \State{$u_0 \gets 0$}
        \State{$r_0 \gets b - Au_0$}
        \State{$z_0 \gets P^{-1}r_0$}
        \State{$h_0 \gets z_0$}
        \For{$k \gets 0$ to $l-1$}
            \State{$v_k \gets Ah_k$}
            \State{$\alpha_{k+1} \gets \frac{r_k^Tz_k}{h_k^Tv_k}$}
            \State{$u_{k+1} \gets u_k + \alpha_{k+1} h_k$}
            \State{$r_{k+1} \gets r_k - \alpha_{k+1} v_k$}
            \If{$\lVert r_{k+1}\rVert_2 < \epsilon$}
                \State{early-stopping $\gets$ true}
            \EndIf
            \State{$z_{k+1} \gets P^{-1}r_{k+1}$}
            \State{$\beta_{k+1} \gets \frac{r_{k+1}^Tz_{k+1}}{r_k^Tz_k}$}
            \State{$h_{k+1} \gets z_{k+1} + \beta_{k+1} h_k$}
            \State{$\tilde{T}_{k+1,k+1} \gets \frac{1}{\alpha_{k+1}} + \frac{\beta_k}{\alpha_{k}}$}
            \If{$k > 0$}
                \State{$\tilde{T}_{k,k+1}, \tilde{T}_{k+1,k} \gets \frac{\sqrt{\beta_k}}{\alpha_{k}}$}
            \EndIf
            \If{early-stopping}
                \State{return $u_{k+1}, \tilde{T}$}
            \EndIf
        \EndFor
    \end{algorithmic}
\end{algorithm}

\begin{algorithm}[ht!]
    \caption{Zero fill-in incomplete Cholesky algorithm}
    \begin{algorithmic}[1]
        \Require Matrix $A\in\mathbb{R}^{m\times m}$, matrix $S\in\mathbb{R}^{m\times m}$ with sparsity pattern
        \Ensure Sparse lower triangular matrix $L$ with $A \approx LL^T$
        \State{$L \gets 0_{m\times m}$}

        \For{$i\gets 1$ to $m$}
            \For{\rev{each $j \leq i$ such that $(i,j) \in \vecmat{S}$}}
                \State{$s \gets \vecmat{L}_{i\cdot} \vecmat{L}_{j\cdot}^T$}
                \If{$i == j$}
                    \State{$L_{ii} \gets \sqrt{A_{ii}-s}$}
                \Else
                    \State{$L_{ij} \gets \frac{A_{ij}-s}{L_{jj}}$}
                \EndIf
            \EndFor
        \EndFor
    \end{algorithmic}
\end{algorithm}

\clearpage
\section{Simulating unbalanced random effects designs}\label{appendix:unbalance_design}
The unbalanced designs are obtained as follows. We first generate  numbers of repeated occurrences $N_{jk}$ (=$\sum_{i=1}^n (Z_k)_{ij}$) of random effect $b_{k,j}$ for level $j$ of component $k$ by simulating from a negative binomial distribution with mean $n/m_k-1$ and size parameter $r$ and then adding $1$. The choice ensures that every random effect $b_{k,j}$ occurs at least once, and there is variability in the numbers of repeated occurrences $N_{jk}$. Note that the size, or inverse-overdispersion, parameter $r$ controls the regularity of the random-effects design, with smaller values of $r$ yielding less regular designs with more variability in the number of occurrences per random effect. Unless stated otherwise, we use a size parameter $r=1$. The resulting total numbers of occurrences $T_k=\sum_{j=1}^{m_k}N_{jk}$ are typically close to $n$, but not exactly equal. Thus, \(n\) denotes the target sample size in this construction, while the realized sample size can differ slightly from \(n\). Furthermore, the simulated totals $T_1$ and $T_2$ for the two random effect components usually differ slightly. To ensure that both components have the same total number of occurrences, we adjust the smaller of the two totals. If, say, the simulated total number of occurrences of the first component is smaller than that of the second component, $T_1 < T_2$, we add $T_2-T_1$ occurrences to the first component by simulating $T_2-T_1$ samples from a multinomial distribution with $m_1$ categories and probabilities $1/m_1$ and adding the resulting count vector to
$(N_{1,1}, \dots, N_{m_1,1})$. Once the final values $N_{jk}$ are obtained, we construct the incidence matrices $Z_k$ by assigning exactly $N_{jk}$ ones to column $j$ in consecutive rows:
\[
(Z_k)_{ij} = 1 \quad \text{for } 
i = \Bigl(\sum_{j'=1}^{j-1} N_{j'k}\Bigr) + 1,\dots, \sum_{j'=1}^{j} N_{j'k},
\]
and zero otherwise (we use the convention $\sum_{j'=1}^{0} N_{j'k} = 0$). Finally, we randomly permute the rows of $Z_2$, which induces a random crossing structure between the two random effect components.

\clearpage
\section{Additional results for preconditioner comparison}\label{appendix:add_Pcomparison}
In the following, we present additional results comparing the preconditioners presented in Section \ref{sec_prec}. Besides these preconditioners, we have also considered two low-rank preconditioners $P = \Sigma^{-1} + L_k L_k^T \approx \Sigma^{-1} + Z^T W Z$, where $L_k \in \mathbb{R}^{m\times k}$ is obtained either (i) using a pivoted Cholesky decomposition \citep{harbrecht2012low, gardner2018gpytorch} to approximate $Z^T W Z \approx L_k L_k^T$ or (ii) using a partial Lanczos algorithm to approximate $Z^T W Z \approx L_k L_k^T$. However, both these low-rank preconditioners lead to slow CG convergence and high variances of the SLQ method (results not shown). %See Section \ref{section:Lanczos_pred} for more information on the partial Lanczos algorithm.
%For these low-rank preconditioners, linear solves and log-determinants can be computed using the Woodbury matrix identity and the matrix determinant lemma, respectively, and the costs thus scale linearly with $m$.

\begin{figure}[ht!]
    \centering
    \includegraphics[width=0.49\linewidth]{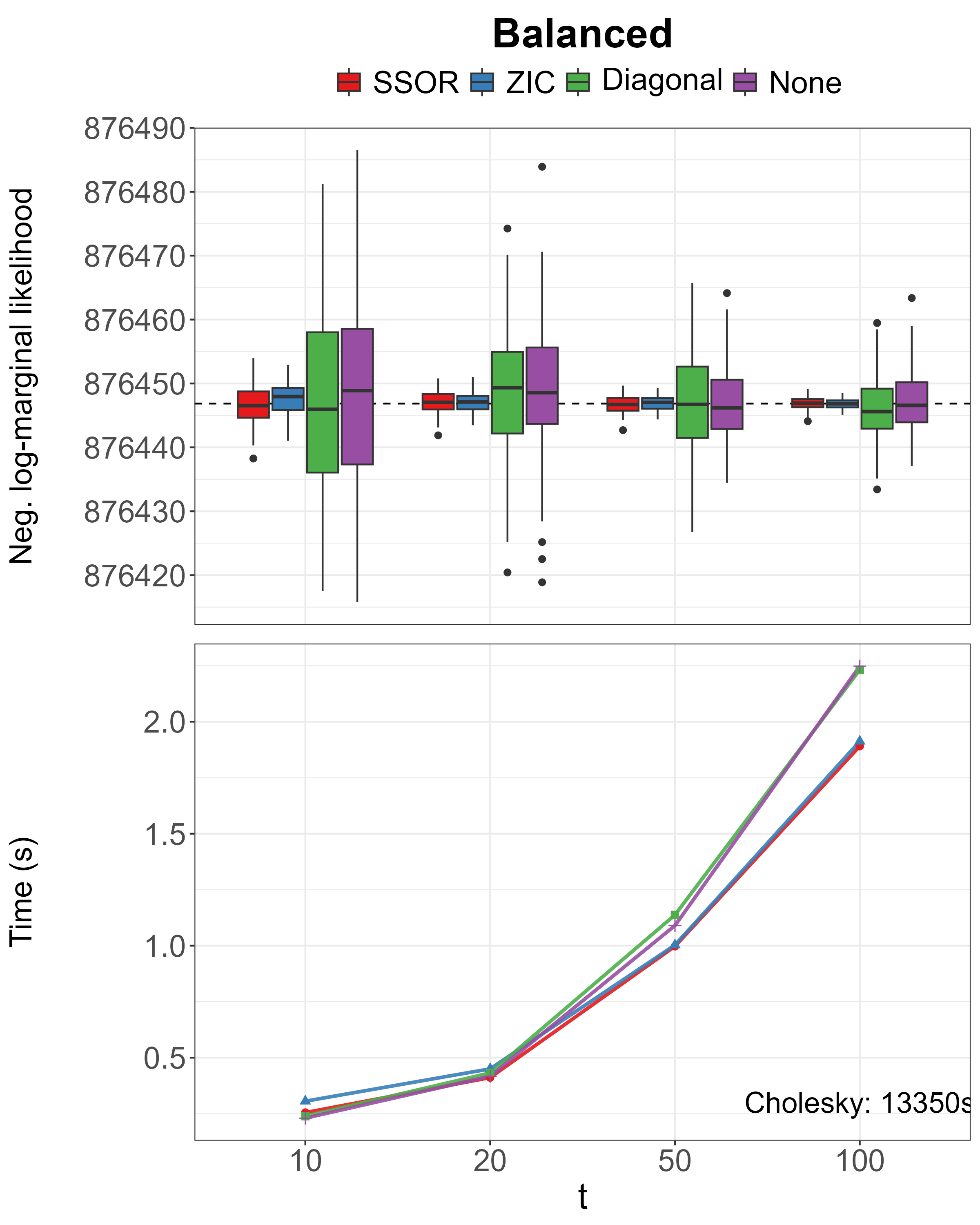}
     \includegraphics[width=0.49\linewidth]{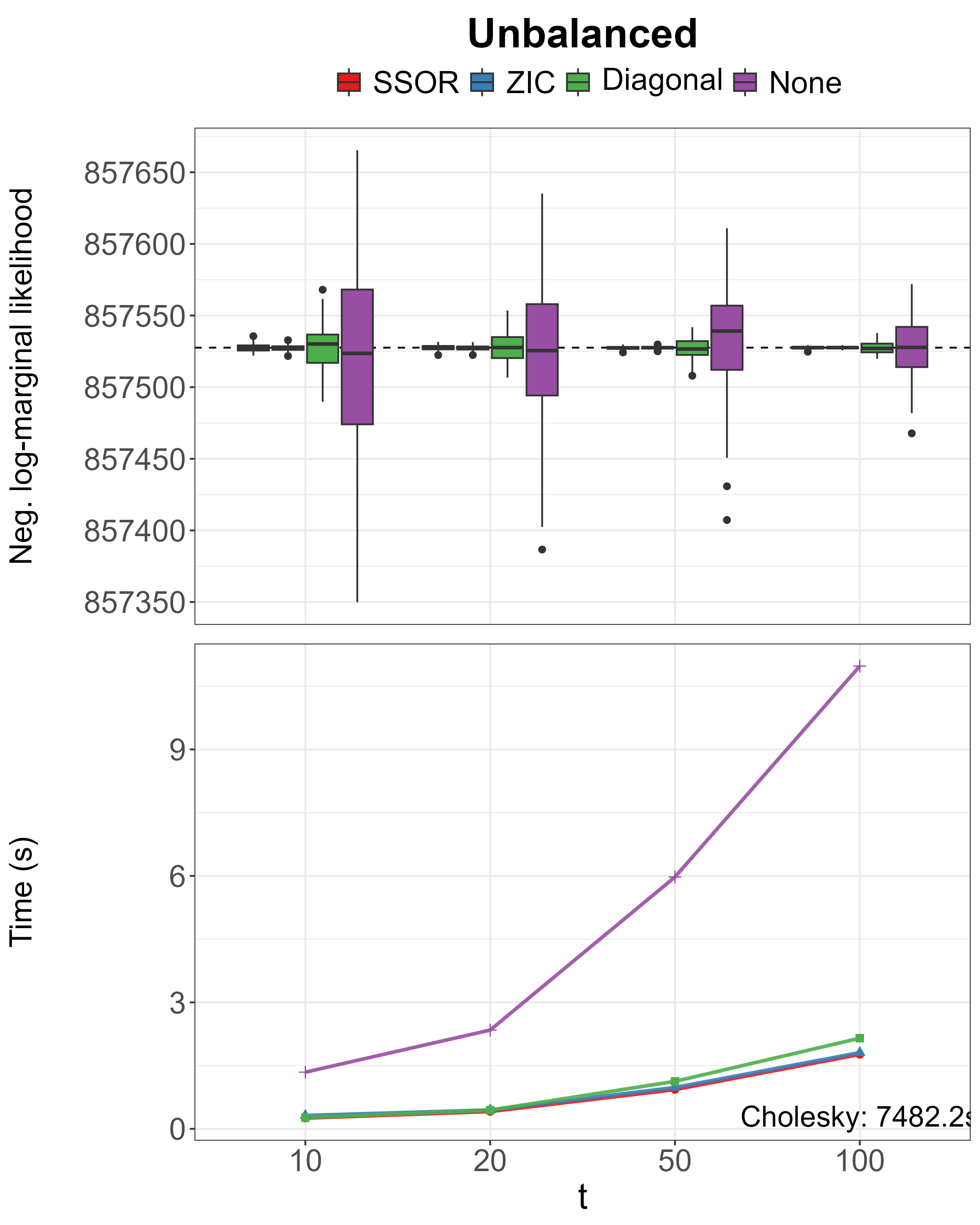}
    \caption{Negative log-marginal likelihood and wall-clock time in seconds for different preconditioners and varying numbers of random vectors $t$ in the SLQ method for a Gaussian likelihood. The dashed line represents the result for the Cholesky decomposition. }
    \label{fig:Pcomparison_Gaussian_hist}
\end{figure}

\begin{figure}[ht!]
    \centering
    \includegraphics[width=0.49\linewidth]{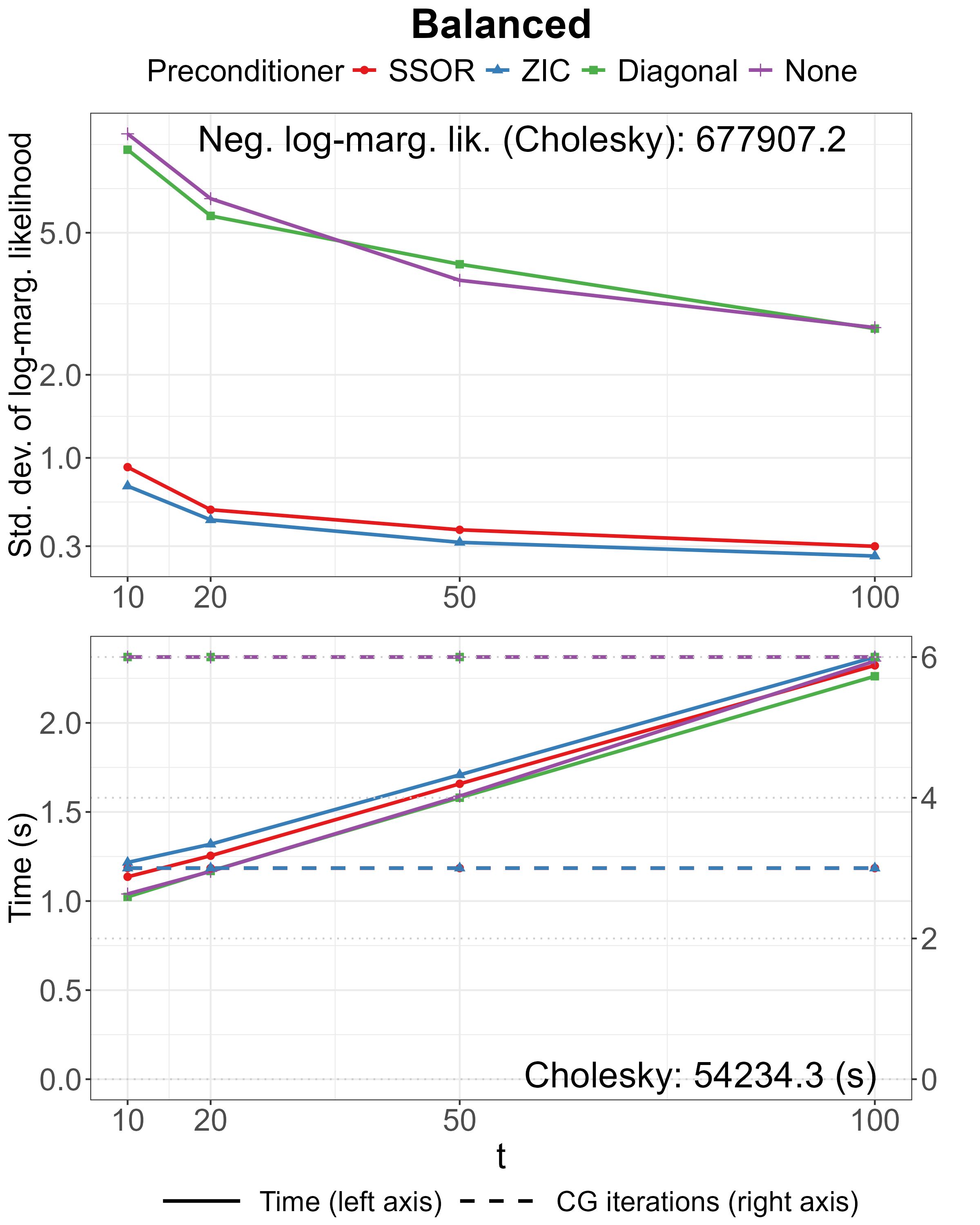}
    \includegraphics[width=0.49\linewidth]{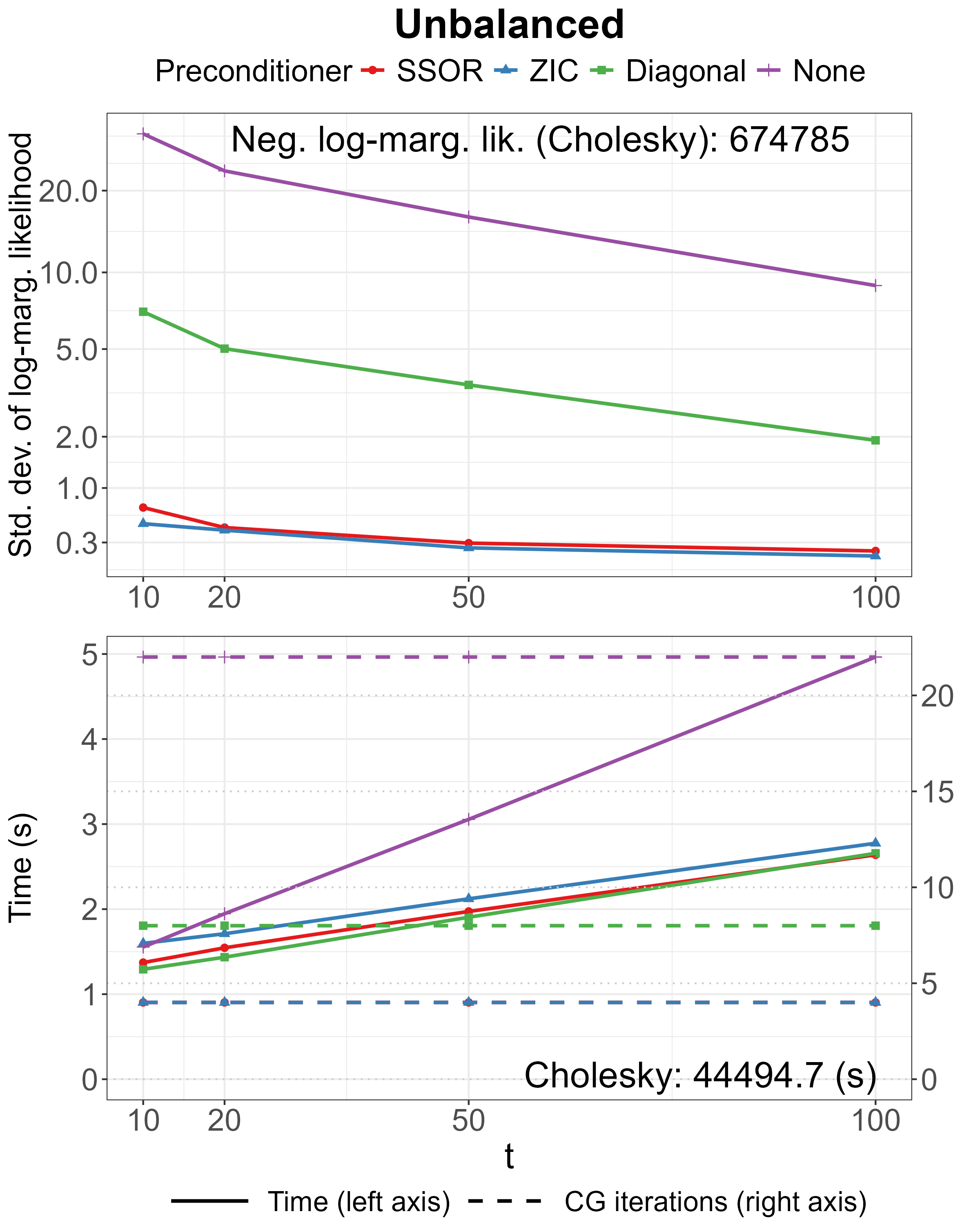}
    \caption{Accuracy (in standard deviations; top row), wall-clock time (in seconds; bottom row, left axis), and number of CG iterations (bottom row, right axis) for varying numbers of random vectors $t$ in the SLQ method for calculating log-marginal likelihoods for a Bernoulli likelihood with a logit link.}
    \label{fig:Pcomparison_Laplace}
\end{figure}

\begin{figure}[ht!]
    \centering
    \includegraphics[width=0.49\linewidth]{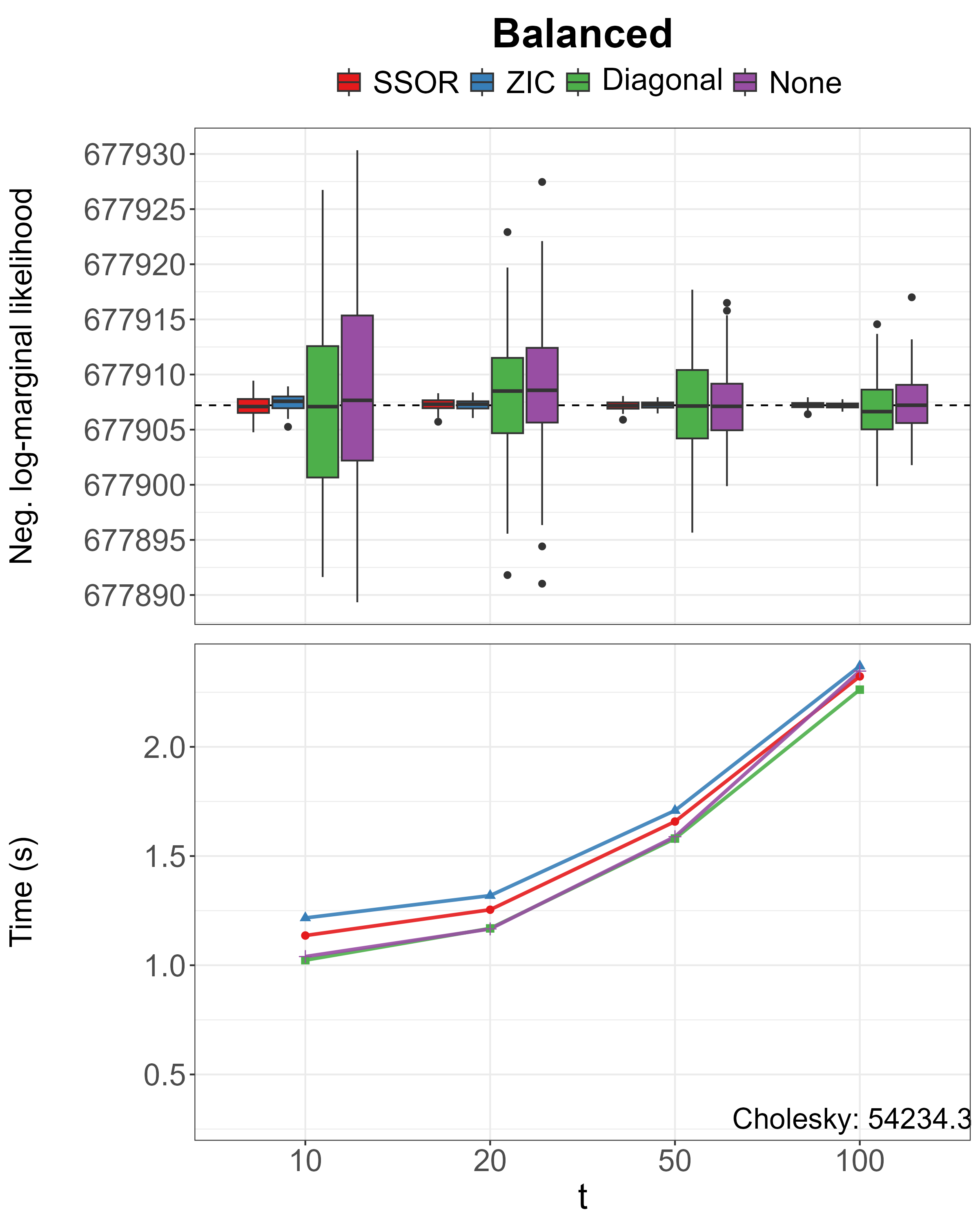}
     \includegraphics[width=0.49\linewidth]{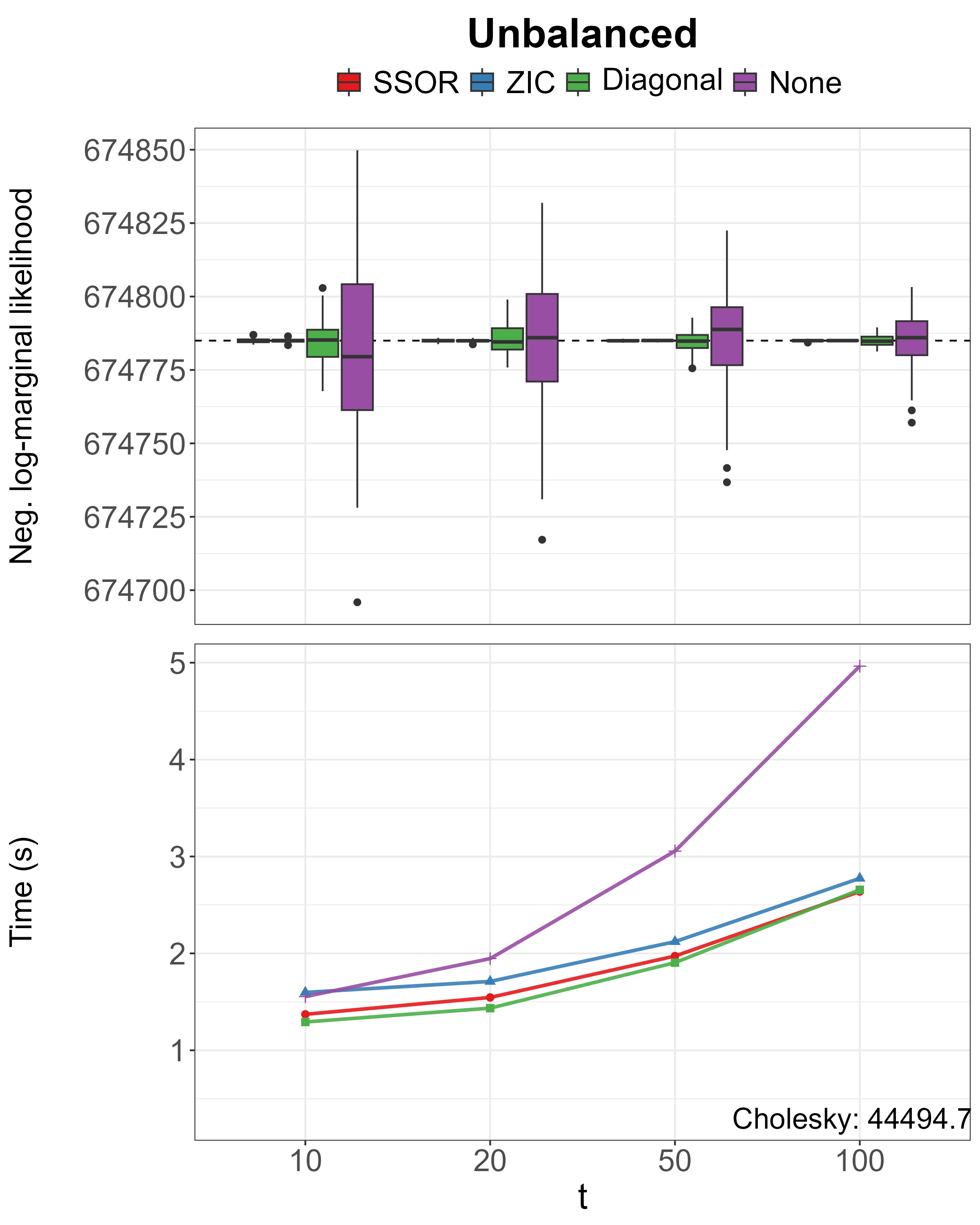}
    \caption{Negative log-marginal likelihood and wall-clock time in seconds for varying numbers of random vectors $t$ in the SLQ method for a Bernoulli likelihood with a logit link. The dashed line represents the result for the Cholesky decomposition. }
    \label{fig:Pcomparison_Laplace_hist}
\end{figure}

\begin{figure}[ht!]
\centering
    \includegraphics[width=0.49\linewidth]{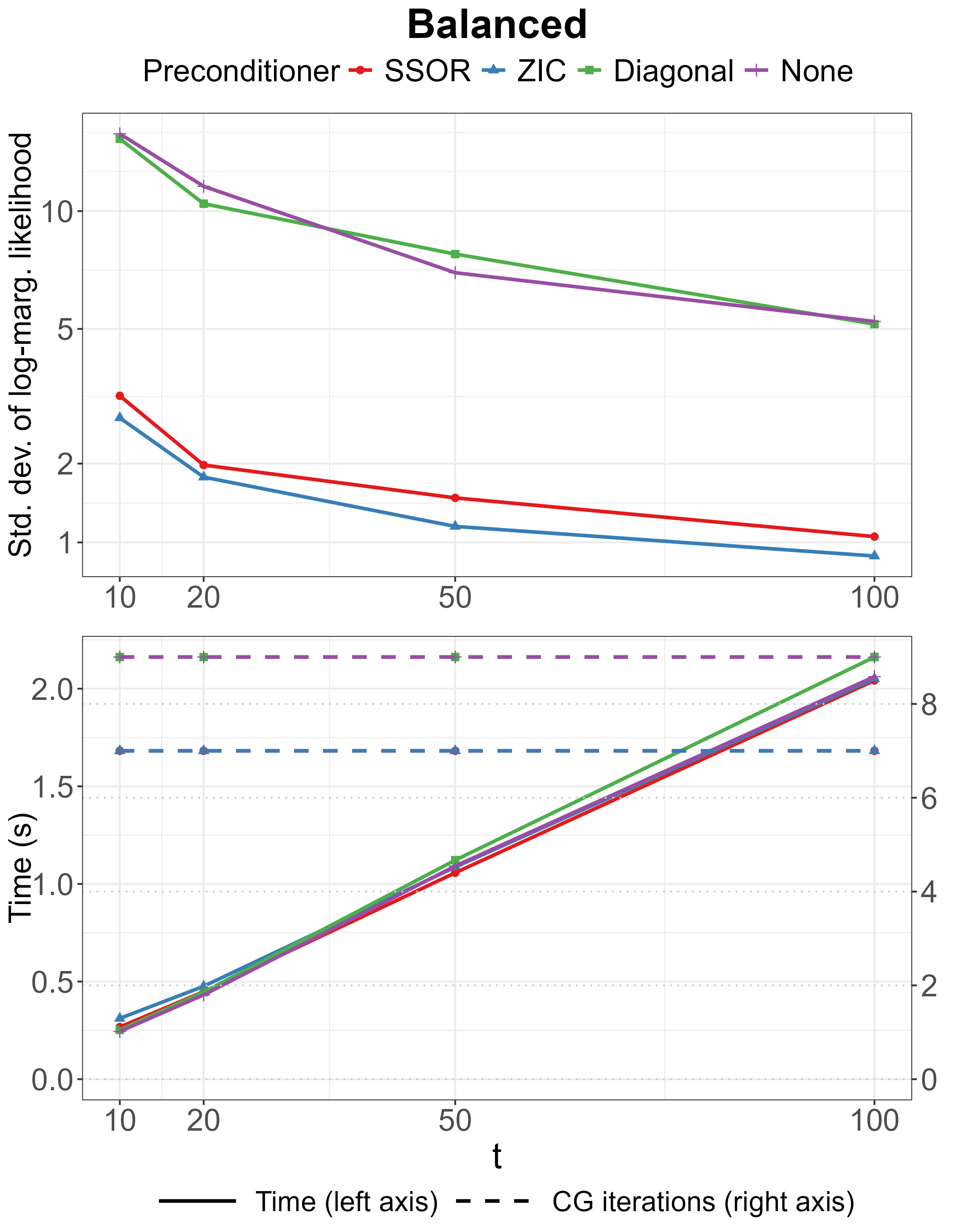}
    \includegraphics[width=0.49\linewidth]{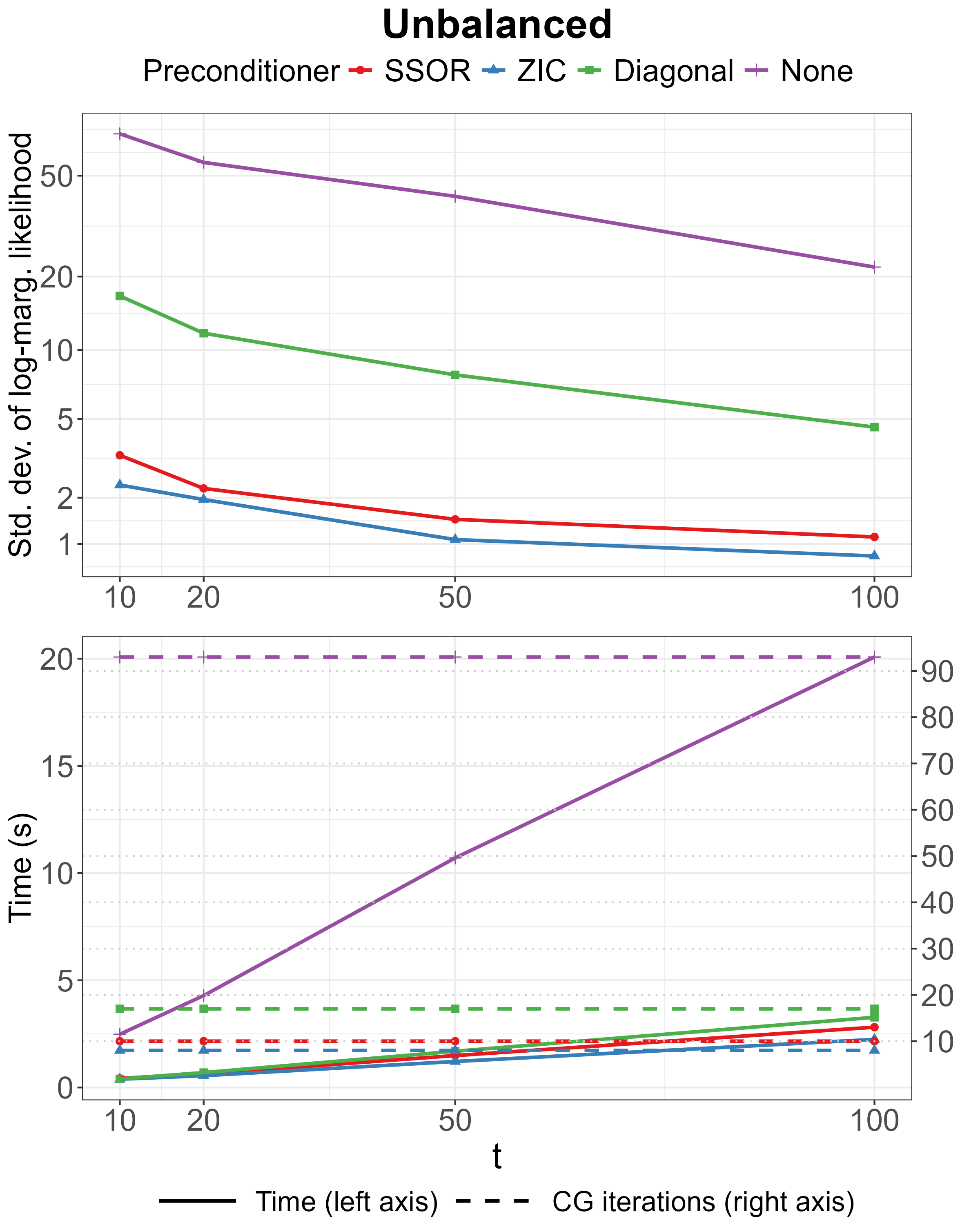}
    \caption{Accuracy (in standard deviations; top row), wall-clock time (in seconds; bottom row, left axis), and number of CG iterations (bottom row, right axis) for varying numbers of random vectors $t$ in the SLQ method for calculating log-marginal likelihoods for a Gaussian likelihood with an error variance $\sigma^2=0.05$ (= higher signal-to-noise ratio).}
    \label{fig:Pcomparison_Gaussian_other_variance}
\end{figure}

\begin{figure}[ht!]
\centering
    \includegraphics[width=0.49\linewidth]{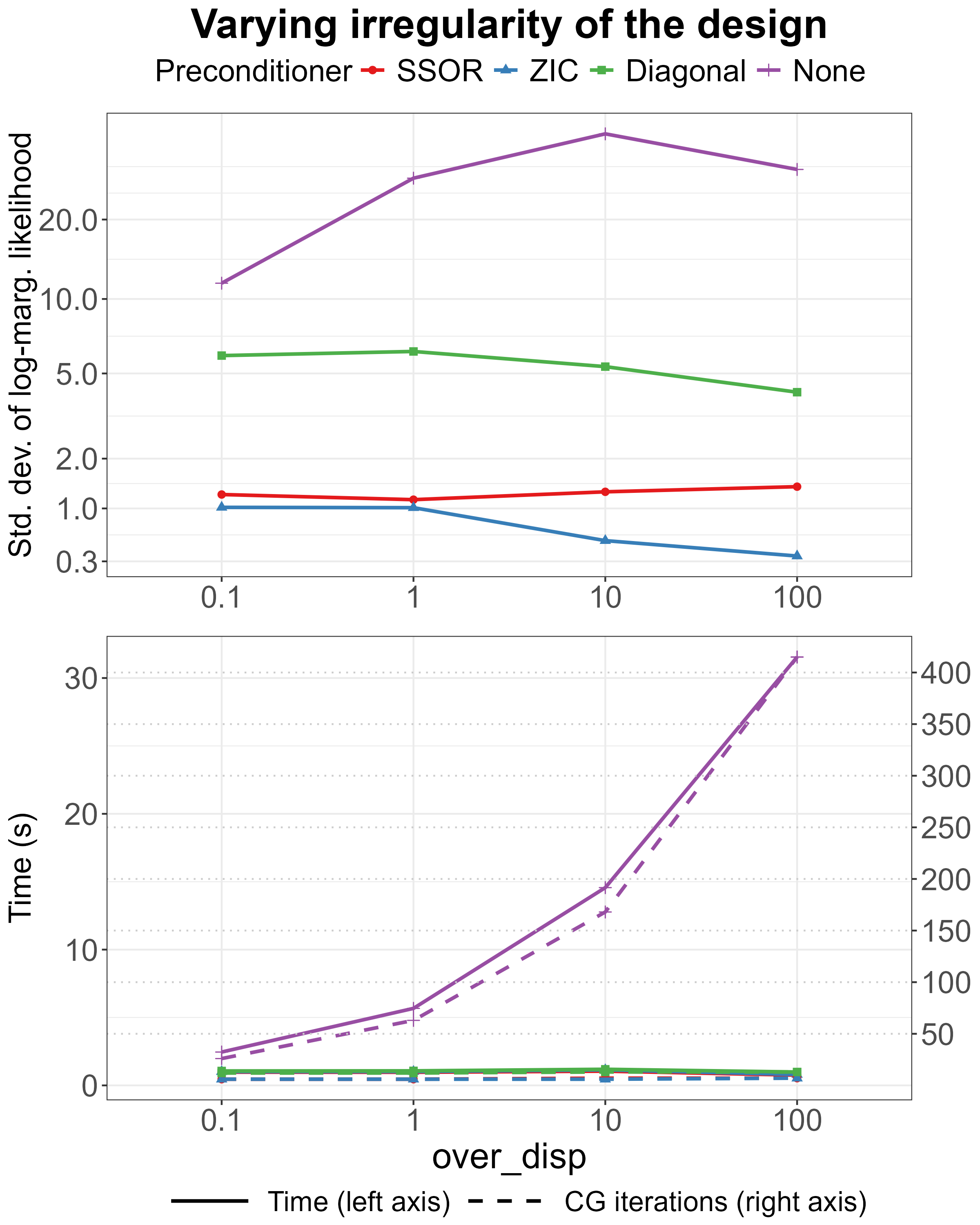}
    \includegraphics[width=0.49\linewidth]{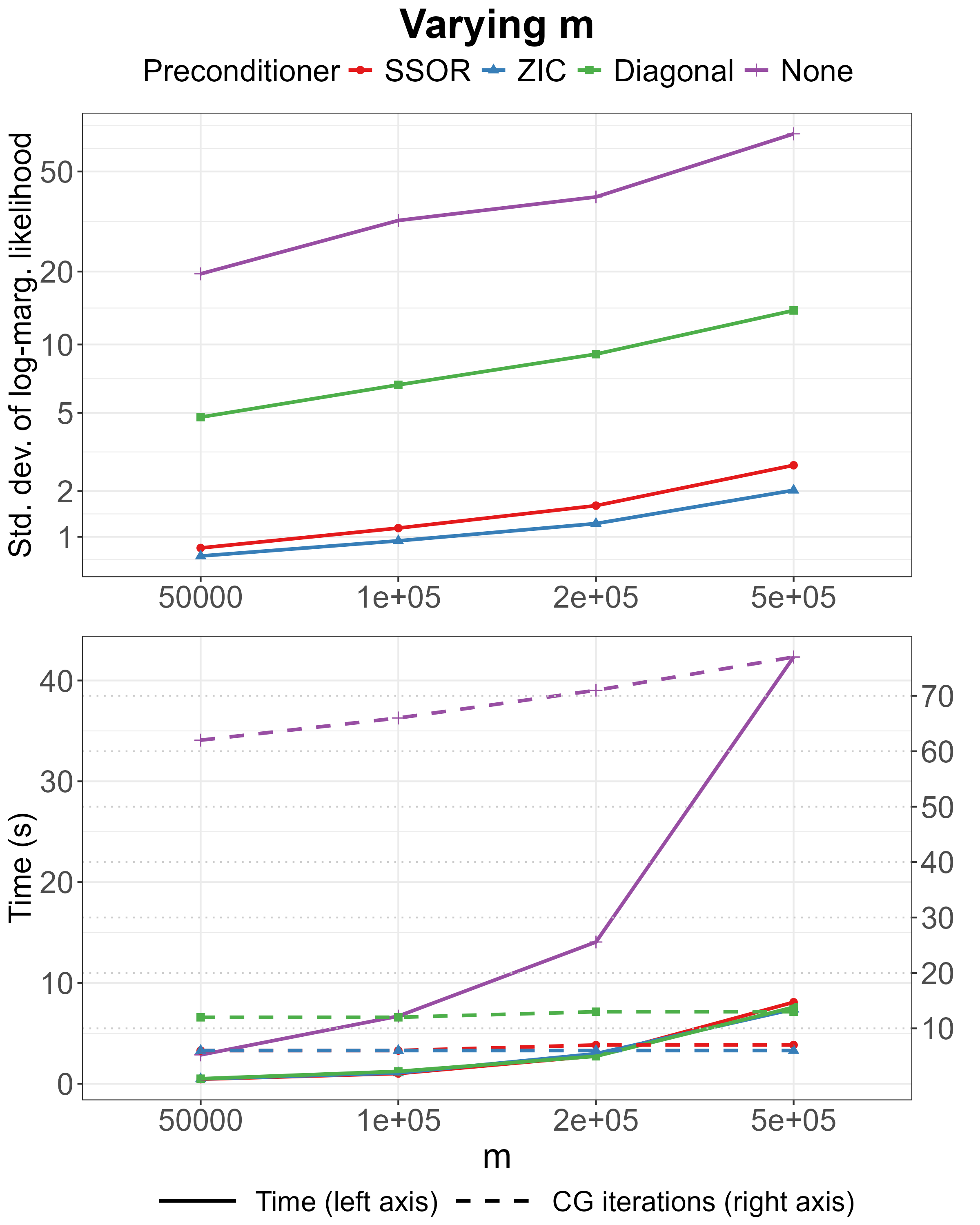}
    \caption{Accuracy (in standard deviations; top row), wall-clock time (in seconds; bottom row, left axis), and number of CG iterations (bottom row, right axis) for varying irregularity of the random effects design (left column) and random effects dimension $m$ (right column) for calculating log-marginal likelihoods for a Gaussian likelihood. `over\_disp' equals the inverse size parameter $1/r$ in the unbalanced design simulation approach explained in Appendix \ref{appendix:unbalance_design}. A larger `over\_disp' value corresponds to a less regular random effects design with more variability in the number of occurrences per random effect. An unbalanced design with $r=1$ is used when varying $m$.}
    \label{fig:Pcomparison_Gaussian_m_over_disp}
\end{figure}

\clearpage

\rev{
\section{Additional information on the \texttt{R-INLA} comparison}\label{appendix:inla}
\texttt{R-INLA} is run with its empirical-Bayes integration strategy (\texttt{int.strategy = 'eb'}) and with flat (improper uniform) priors both on the fixed effects (\texttt{control.fixed = list(mean = 0, prec = 0)}) and on the internally parameterized log-precisions of the random effects and, for a Gaussian likelihood, of the error term. With these priors, the EB mode maximizes the INLA-approximated marginal likelihood, which is the closest analogue to maximum likelihood available in \texttt{R-INLA}. Since \texttt{R-INLA} treats the fixed effects as part of the latent Gaussian field and integrates them out, this targets a restricted (REML-type) rather than an ML criterion; the variance parameters reported in Table \ref{table:real_world_variances} are obtained from the mode of the internal (log-precision) parameterization found by \texttt{R-INLA}'s optimizer, i.e., from \texttt{mode\$theta}, transformed to the variance scale. We do not use the \texttt{mode} column of \texttt{summary.hyperpar}, which is the mode of the marginal density after the nonlinear transformation to the precision scale and differs substantially. For the fixed-effect estimates reported in Table \ref{table:real_world_coefficients}, we use the \texttt{mode} column of \texttt{summary.fixed}. Under the Gaussian latent strategy used here, the marginal for each fixed effect is the corresponding Gaussian approximation at the modal hyperparameters, so this mode coincides with the conditional mode of the latent field at the EB mode and is directly comparable to the point estimates of the other implementations. The hyperparameters themselves are estimated from \texttt{R-INLA}'s Laplace approximation of the posterior of the log-precisions irrespective of the \texttt{strategy} argument, which only governs the approximation of the marginals of the latent field. Since we do not require posterior marginal approximations for the random effects, we set \texttt{strategy = 'gaussian'} and do not return the full marginal densities (\texttt{return.marginals = FALSE}); we verified that this leaves the hyperparameter mode and the marginal likelihood unchanged.
}

\section{Additional results for parameter estimation and prediction in simulated experiments}\label{appendix:add_res_est_pred}

\rev{This appendix reports the unbalanced-design runtime comparison referenced in Section~\ref{sec:times} and the detailed parameter-estimation and prediction-accuracy results discussed in Section~\ref{acc_param_estim}. For Bernoulli likelihoods, repeating the estimation $100$ times with \texttt{lme4} leads to very long runtimes; see also Section~\ref{sec:times}. Therefore, we do not report estimates and prediction accuracy for this case. Log scores are only reported for Krylov subspace methods and Cholesky-based computations from the \texttt{GPBoost} package since the calculation of predictive variances is not supported for \texttt{lme4} and leads to very long runtimes for \texttt{glmmTMB}.}

\begin{figure}[ht!]
    \centering
    \includegraphics[width=\linewidth]{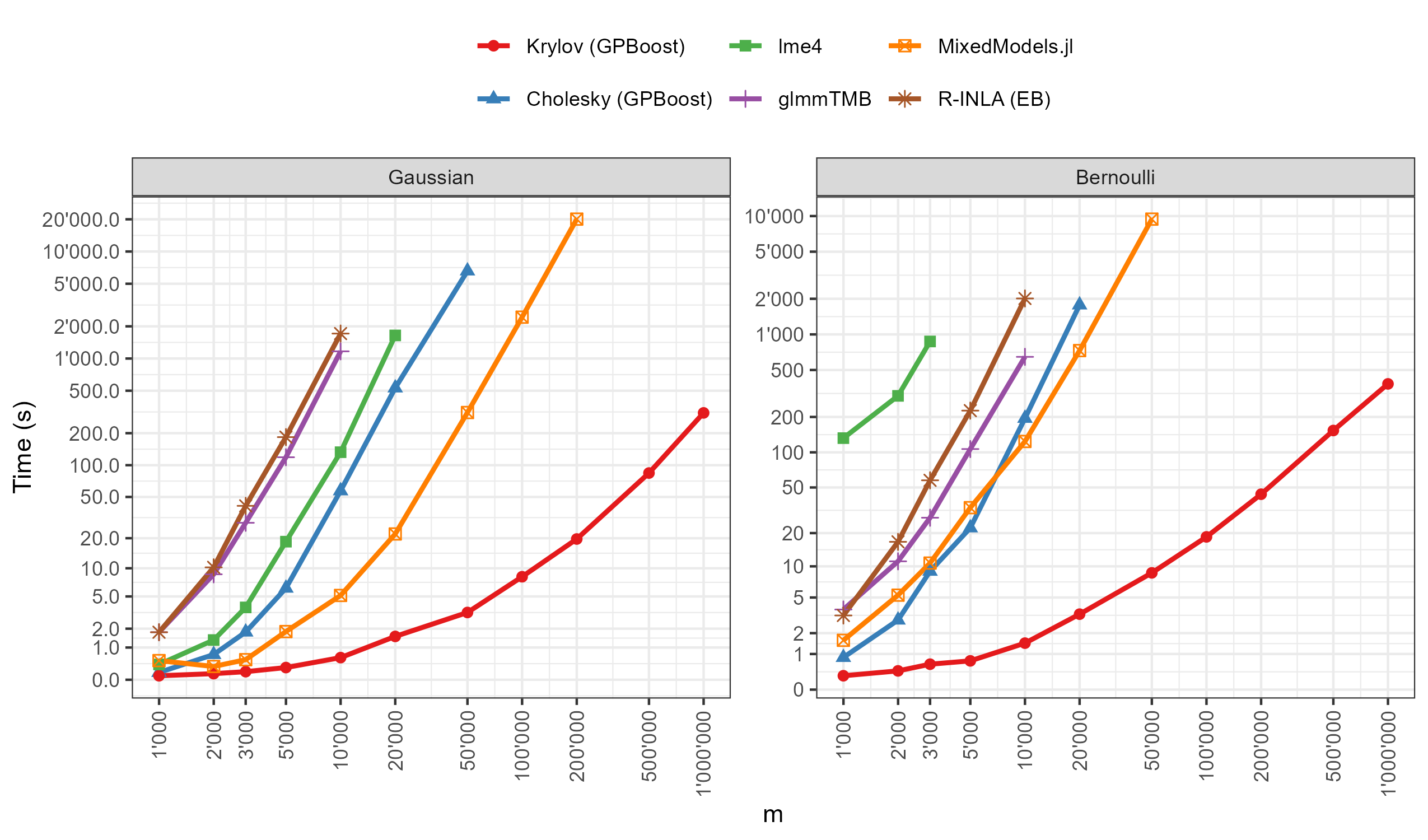}
    \caption{\rev{Parameter-estimation wall-clock times (s) for different $m$, $m_2=m_1/2$, and Gaussian or Bernoulli likelihoods.}}
    \label{fig:runtime_unbalanced}
\end{figure}

\begin{figure}[ht!]
    \centering
    \includegraphics[width=0.9\linewidth]{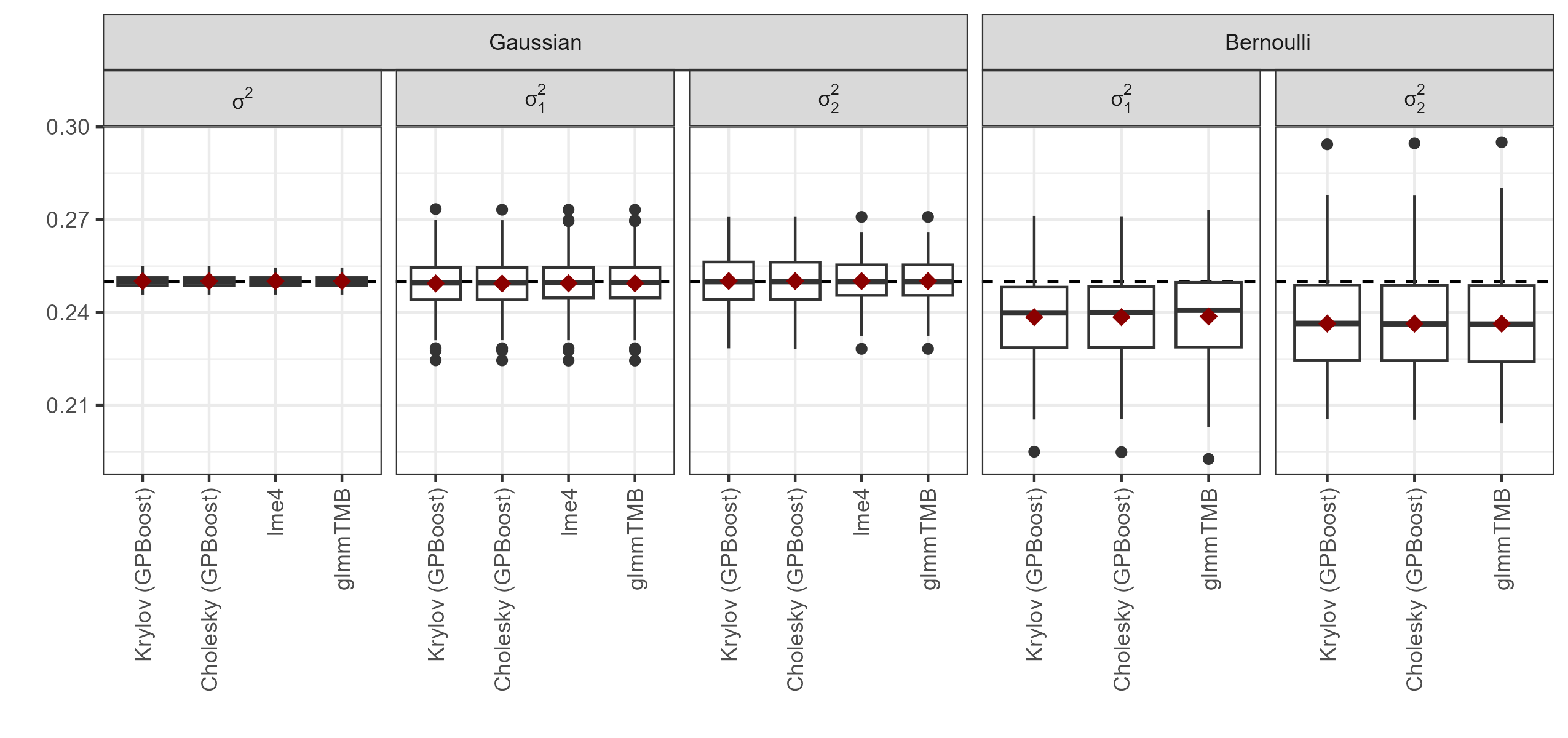}
    \caption{Estimated variance parameters. The red rhombi represent means. The dashed lines indicate the true values. For the Bernoulli likelihood, the estimates for \texttt{lme4} are not computed due to excessively long runtimes. The variance parameter estimates for Bernoulli likelihoods are slightly downward biased due to the Laplace approximation and the moderate sample size.}
    \label{fig:var_estimates}
\end{figure}

\begin{table}[ht]
\centering
\begin{tabular}{l|rrrr}
  \hline
 & Krylov (GPBoost) & Cholesky (GPBoost) & lme4 & glmmTMB \\
  \hline
  RMSE $\sigma^2$ & $1.75 \times 10^{-3}$ & $1.75 \times 10^{-3}$ & $1.73 \times 10^{-3}$ & $1.73 \times 10^{-3}$ \\ 
  Bias $\sigma^2$ & $1.24 \times 10^{-4}$ & $1.26 \times 10^{-4}$ & $1.23 \times 10^{-4}$ & $1.23 \times 10^{-4}$ \\ 
  RMSE $\sigma_1^2$ & $8.97 \times 10^{-3}$ & $8.96 \times 10^{-3}$ & $8.83 \times 10^{-3}$ & $8.83 \times 10^{-3}$ \\ 
  Bias $\sigma_1^2$ & $-5.96 \times 10^{-4}$ & $-6.15 \times 10^{-4}$ & $-5.57 \times 10^{-4}$ & $-5.57 \times 10^{-4}$ \\ 
  RMSE $\sigma_2^2$ & $8.34 \times 10^{-3}$ & $8.35 \times 10^{-3}$ & $8.13 \times 10^{-3}$ & $8.13 \times 10^{-3}$ \\ 
  Bias $\sigma_2^2$ & $1.79 \times 10^{-4}$ & $1.56 \times 10^{-4}$ & $1.53 \times 10^{-4}$ & $1.54 \times 10^{-4}$ \\ 
  RMSE $\beta_0$ & $1.46 \times 10^{-2}$ & $1.46 \times 10^{-2}$ & $1.46 \times 10^{-2}$ & $1.46 \times 10^{-2}$ \\ 
  Bias $\beta_0$ & $1.36 \times 10^{-3}$ & $1.36 \times 10^{-3}$ & $1.36 \times 10^{-3}$ & $1.36 \times 10^{-3}$ \\ 
  RMSE $\beta_1$ & $1.17 \times 10^{-2}$ & $1.17 \times 10^{-2}$ & $1.17 \times 10^{-2}$ & $1.17 \times 10^{-2}$ \\ 
  Bias $\beta_1$ & $-2.06 \times 10^{-4}$ & $-2.06 \times 10^{-4}$ & $-2.06 \times 10^{-4}$ & $-2.06 \times 10^{-4}$ \\ 
  RMSE $\beta_2$ & $1.08 \times 10^{-2}$ & $1.08 \times 10^{-2}$ & $1.08 \times 10^{-2}$ & $1.08 \times 10^{-2}$ \\ 
  Bias $\beta_2$ & $-4.41 \times 10^{-4}$ & $-4.41 \times 10^{-4}$ & $-4.41 \times 10^{-4}$ & $-4.41 \times 10^{-4}$ \\ 
  RMSE $\beta_3$ & $1.25 \times 10^{-2}$ & $1.25 \times 10^{-2}$ & $1.25 \times 10^{-2}$ & $1.25 \times 10^{-2}$ \\ 
  Bias $\beta_3$ & $-1.14 \times 10^{-3}$ & $-1.14 \times 10^{-3}$ & $-1.14 \times 10^{-3}$ & $-1.14 \times 10^{-3}$ \\ 
  RMSE $\beta_4$ & $1.31 \times 10^{-2}$ & $1.31 \times 10^{-2}$ & $1.31 \times 10^{-2}$ & $1.31 \times 10^{-2}$ \\ 
  Bias $\beta_4$ & $1.28 \times 10^{-3}$ & $1.28 \times 10^{-3}$ & $1.28 \times 10^{-3}$ & $1.28 \times 10^{-3}$ \\ 
  RMSE $\beta_5$ & $1.06 \times 10^{-2}$ & $1.06 \times 10^{-2}$ & $1.06 \times 10^{-2}$ & $1.06 \times 10^{-2}$ \\ 
  Bias $\beta_5$ & $3.21 \times 10^{-4}$ & $3.21 \times 10^{-4}$ & $3.21 \times 10^{-4}$ & $3.21 \times 10^{-4}$ \\ 
   \hline
\end{tabular}
\caption{Root mean squared error (RMSE) and bias of the variance and coefficient estimates for Gaussian likelihoods.}
\label{table:estimates_Gaussian}
\end{table}

\begin{table}[ht]
\centering
\begin{tabular}{l|rrr}
  \hline
 & Krylov (GPBoost)  & Cholesky (GPBoost)  & glmmTMB \\ 
  \hline
RMSE $\sigma_1^2$ & $2.01 \times 10^{-2}$ & $2.01 \times 10^{-2}$ & $2.05 \times 10^{-2}$ \\ 
  Bias $\sigma_1^2$ & $-1.15 \times 10^{-2}$ & $-1.15 \times 10^{-2}$ & $-1.13 \times 10^{-2}$ \\ 
  RMSE $\sigma_2^2$ & $2.15 \times 10^{-2}$ & $2.15 \times 10^{-2}$ & $2.20 \times 10^{-2}$ \\ 
  Bias $\sigma_2^2$ & $-1.36 \times 10^{-2}$ & $-1.36 \times 10^{-2}$ & $-1.36 \times 10^{-2}$ \\ 
  RMSE $\beta_0$ & $1.73 \times 10^{-2}$ & $1.73 \times 10^{-2}$ & $1.73 \times 10^{-2}$ \\ 
  Bias $\beta_0$ & $-3.68 \times 10^{-4}$ & $-3.67 \times 10^{-4}$ & $-4.76 \times 10^{-4}$ \\ 
  RMSE $\beta_1$ & $5.18 \times 10^{-2}$ & $5.18 \times 10^{-2}$ & $5.10 \times 10^{-2}$ \\ 
  Bias $\beta_1$ & $1.33 \times 10^{-3}$ & $1.34 \times 10^{-3}$ & $7.09 \times 10^{-4}$ \\ 
  RMSE $\beta_2$ & $5.33 \times 10^{-2}$ & $5.33 \times 10^{-2}$ & $5.19 \times 10^{-2}$ \\ 
  Bias $\beta_2$ & $-6.51 \times 10^{-3}$ & $-6.53 \times 10^{-3}$ & $-5.87 \times 10^{-3}$ \\ 
  RMSE $\beta_3$ & $5.10 \times 10^{-2}$ & $5.09 \times 10^{-2}$ & $5.05 \times 10^{-2}$ \\ 
  Bias $\beta_3$ & $6.00 \times 10^{-4}$ & $5.82 \times 10^{-4}$ & $1.86 \times 10^{-4}$ \\ 
  RMSE $\beta_4$ & $5.25 \times 10^{-2}$ & $5.25 \times 10^{-2}$ & $5.24 \times 10^{-2}$ \\ 
  Bias $\beta_4$ & $-1.24 \times 10^{-3}$ & $-1.23 \times 10^{-3}$ & $-1.30 \times 10^{-3}$ \\ 
  RMSE $\beta_5$ & $5.18 \times 10^{-2}$ & $5.18 \times 10^{-2}$ & $4.99 \times 10^{-2}$ \\ 
  Bias $\beta_5$ & $-6.36 \times 10^{-3}$ & $-6.36 \times 10^{-3}$ & $-5.61 \times 10^{-3}$ \\ 
   \hline
\end{tabular}
\caption{Root mean squared error (RMSE) and bias of the variance and coefficient estimates for Bernoulli likelihoods.} 
\label{table:estimates_Bernoulli}
\end{table}

\begin{figure}[ht!]
    \centering
    \includegraphics[width=0.9\linewidth]{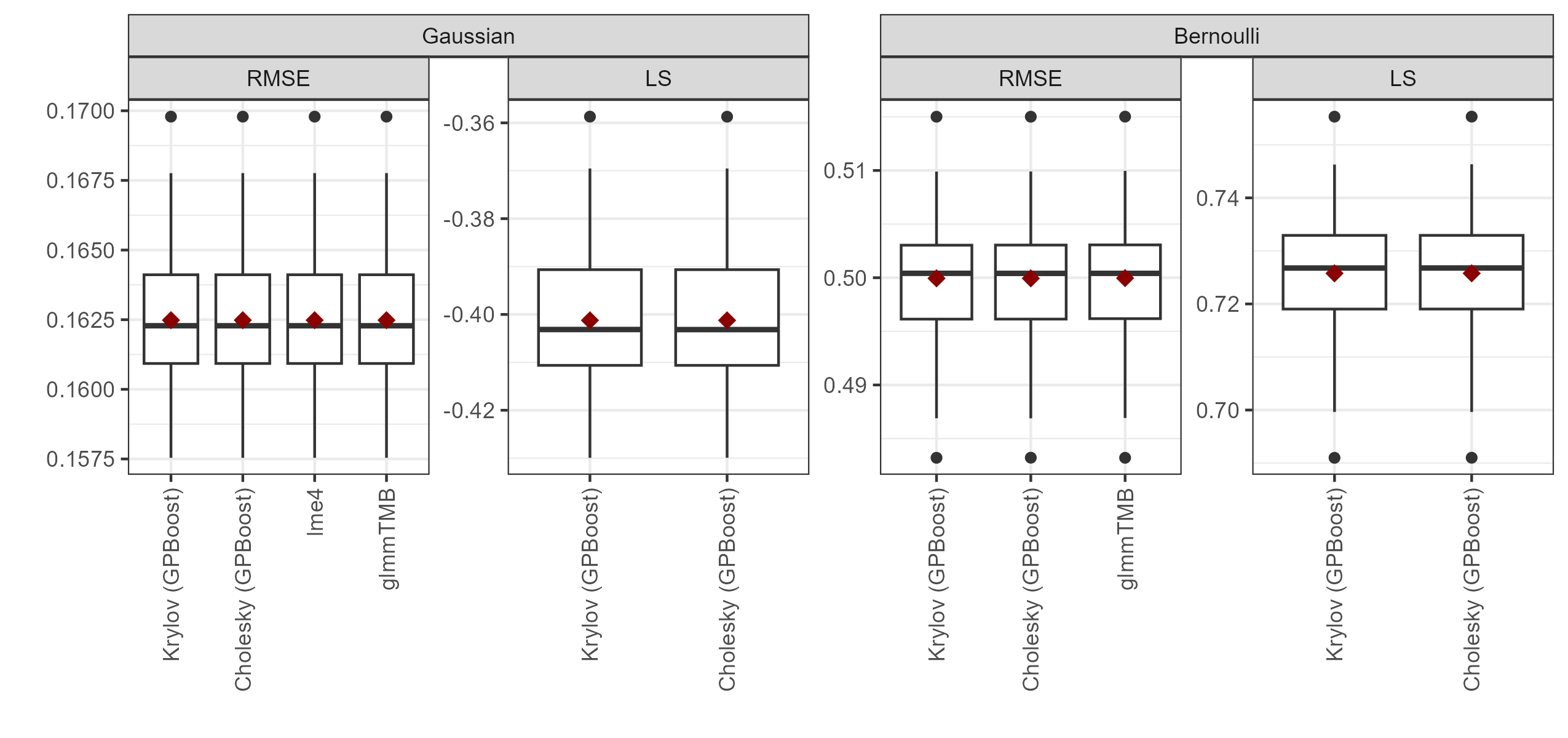}
    \caption{RMSE for predictive means and log score (LS) for probabilistic predictions for Gaussian and Bernoulli likelihoods.}
    \label{fig:RMSE_LS}
\end{figure}

\begin{table}[ht]
\centering
\begin{tabular}{l|rrrr}
  \hline
 & Krylov (GPBoost) & Cholesky (GPBoost) & lme4 & glmmTMB \\
  \hline
$\overline{\text{RMSE}}$ & $1.62 \times 10^{-1}$ & $1.62 \times 10^{-1}$ & $1.62 \times 10^{-1}$ & $1.62 \times 10^{-1}$ \\ 
  sd($\overline{\text{RMSE}}$) & $2.32 \times 10^{-4}$ & $2.32 \times 10^{-4}$ & $2.32 \times 10^{-4}$ & $2.32 \times 10^{-4}$ \\ 
  $\overline{\text{LS}}$ & $-4.01 \times 10^{-1}$ & $-4.01 \times 10^{-1}$ &  &  \\ 
  sd($\overline{\text{LS}}$) & $1.40 \times 10^{-3}$ & $1.40 \times 10^{-3}$ &  &  \\ 
   \hline
\end{tabular}
\caption{Average root mean squared error (RMSE) for predictive means and average log score (LS) for probabilistic predictions with corresponding standard deviations for Gaussian likelihoods.}
\label{table:prediction_Gaussian}
\end{table}

\begin{table}[ht]
\centering
\begin{tabular}{l|lll}
  \hline
 & Krylov (GPBoost)  & Cholesky (GPBoost)  & glmmTMB  \\ 
  \hline
$\overline{\text{RMSE}}$ & $5.00 \times 10^{-1}$ & $5.00 \times 10^{-1}$ & $5.00 \times 10^{-1}$ \\ 
  sd($\overline{\text{RMSE}}$) & $5.37 \times 10^{-4}$ & $5.37 \times 10^{-4}$ & $5.37 \times 10^{-4}$ \\ 
  $\overline{\text{LS}}$ & $7.26 \times 10^{-1}$ & $7.26 \times 10^{-1}$ &  \\ 
  sd($\overline{\text{LS}}$) & $1.12 \times 10^{-3}$ & $1.12 \times 10^{-3}$ &  \\ 
   \hline
\end{tabular}
\caption{Average root mean squared error (RMSE) for predictive means and average log score (LS) for probabilistic predictions with corresponding standard deviations for Bernoulli likelihoods.}
\label{table:prediction_Bernoulli}
\end{table}

\clearpage
\section{Additional information and results for the real-world data sets}\label{section:add_real_world}

\rev{This appendix provides additional data-set information, visualizations of $Z^TZ$, and detailed runtime and parameter-estimation results for the real-world applications discussed in Section~\ref{sec:real_world}.}

For more information on the `MovieLens\_32m' data set, see \citet{harper2015movielens}. The data sets `upselling' (`KDDCup09\_upselling') and `employee\_access' (`Amazon\_employee\_access') have previously been used by \citet{pargent2022regularized}. For more information on the `instEval' and `cars' data sets, we refer to \citet{simchoni2023integrating}, and the data set `building\_permits' (`chicago\_building\_permits') has been used by \citet{reyes2019statistical}. \rev{Table \ref{table:real_data_sets} gives a compact overview of the data sets, while Table \ref{table:real_data_sets_detailed} reports the categorical variables and their individual numbers of levels $m_k$.} In Figure \ref{fig:real_world_ZtZ}, we additionally visualize the non-zero entries of the matrix $Z^TZ$ for the different data sets. In general, an extra category is added for missing values in categorical variables and numeric predictor variables are simply imputed using the mean.

\begin{table}[ht!]
\centering
\ifrevision\color{red}\fi
\begin{tabular}{c|lrrrlcrr}
  \hline
  & data set & $n$ & $p$ & $K$ & $y$ & Cat. var. & $m_k$ & $\text{nnz}(Z^TZ)$\\
  \hline
  \multirow[c]{7}{*}[0in]{\rotatebox{90}{\textbf{Regression}}}
  & cars & 97{,}729 & 66 & 2 & log(price) & model\_id & 15{,}226 & 185{,}807\\
  && & & & & location\_id & 12{,}235 & \\
  \cline{2-9}
  & building\_permits & 527{,}168 & 6 & 3 & log(cost) & contact\_name & 98{,}362 & 2{,}544{,}893\\
  && & & & & latitude & 206{,}197 & \\
  && & & & & longitude & 206{,}170 & \\
  \cline{2-9}
  & instEval & 73{,}421 & 22 & 2 & rating & student & 2{,}972 & 150{,}942\\
  && & & & & teacher & 1{,}128 & \\
  \cline{2-9}
  & MovieLens\_32m & 32{,}000{,}204 & 2 & 2 & rating & userId & 200{,}948 & 64{,}285{,}788\\
  && & & & & movieId & 84{,}432 & \\
  \hline
  \multirow[c]{13}{*}[0in]{\rotatebox{90}{\textbf{Classification}}}
  & employee\_access & 32{,}769 & 0 & 9 & approval & resource & 7{,}518 & 469{,}262\\
  && & & & & manager & 4{,}243 & \\
  && & & & & role\_cat1 & 128 & \\
  && & & & & role\_cat2 & 177 & \\
  && & & & & role\_dep. & 449 & \\
  && & & & & role\_title & 343 & \\
  && & & & & role\_descr. & 2{,}358 & \\
  && & & & & role\_family & 67 & \\
  && & & & & role\_code & 343 & \\
  \cline{2-9}
  & upselling & 50{,}000 & 34 & 4 & up-selling & Var216 & 2{,}016 & 454{,}852\\
  && & & & & Var217 & 13{,}991 & \\
  && & & & & Var198 & 4{,}291 & \\
  && & & & & Var199 & 5{,}074 & \\
  \hline
\end{tabular}
\caption{Detailed summary of the real-world data sets. `Cat. var.': categorical variables modeled with random effects; $m_k$: number of levels of categorical variable $k$. The remaining columns are defined in Table \ref{table:real_data_sets}.}
\label{table:real_data_sets_detailed}
\end{table}

\begin{figure}[ht!]
    \centering
    \begin{tabular}{cc}
        \includegraphics[width=0.45\linewidth]{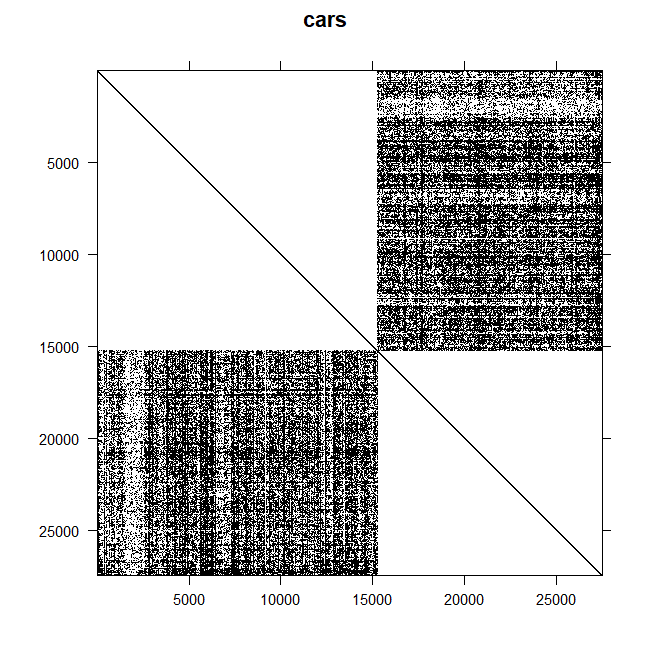} &
        \includegraphics[width=0.45\linewidth]{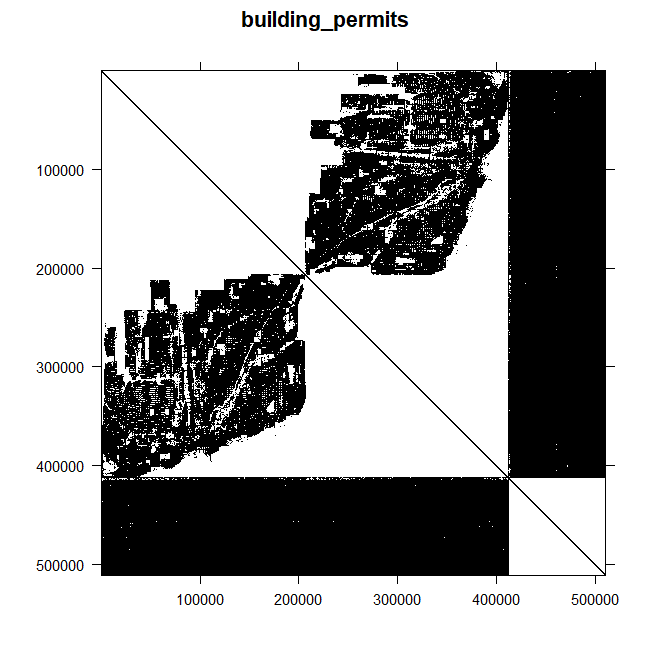} \\
        \includegraphics[width=0.45\linewidth]{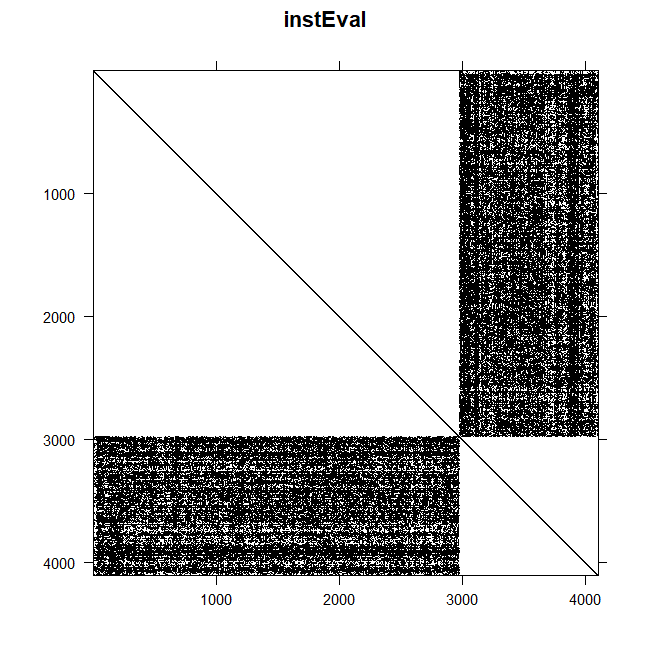} &
        \includegraphics[width=0.45\linewidth]{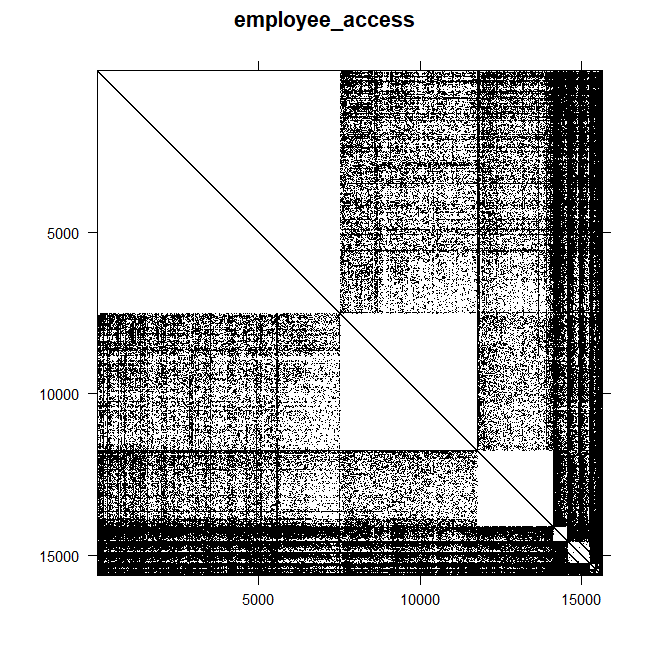} \\
        \includegraphics[width=0.45\linewidth]{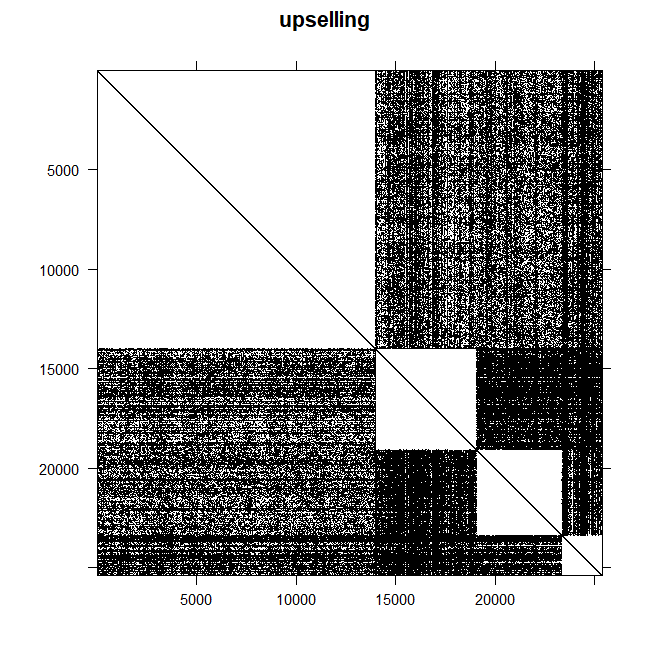} &
    \end{tabular}
    \caption{Non-zero entries of the matrix $Z^TZ$ for the real-world data sets (except the MovieLens\_32m data set for which $Z^TZ$ is too large to visualize).}
    \label{fig:real_world_ZtZ}
\end{figure}

\clearpage

\begin{table}[ht!]
\centering
\begin{tabular}{llrr}
  \hline
data set & method & Time (s) & nll\_optimum \\
  \hline
  cars & Cholesky (GPBoost) & 23.5 & 27518.0 \\
   & Krylov (GPBoost) & 5.2 & 27518.3 \\
   & glmmTMB & 1153.3 & 27518.0 \\
   & lme4 & 36.0 & 27518.0 \\
   & MixedModels.jl & 225.9 & 27518.0 \\
   & R-INLA (EB) & 213.0 & $27785.5^{\dagger}$ \\
  \hline
  building\_permits & Cholesky (GPBoost) & 2931.7 & 867567.7 \\
   & Krylov (GPBoost) & 60.2 & 867567.8 \\
   & glmmTMB & 2967.4 & 867567.8 \\
   & lme4 & 736.1 & 867567.3 \\
   & MixedModels.jl & crashed &  \\
   & R-INLA (EB) & 3263.4 & $867627.8^{\dagger}$ \\
  \hline
  instEval & Cholesky (GPBoost) & 2.3 & 118764.0 \\
   & Krylov (GPBoost) & 0.7 & 118764.0 \\
   & glmmTMB & 54.4 & 118764.0 \\
   & lme4 & 6.1 & 118764.0 \\
   & MixedModels.jl & 1.7 & 118764.0 \\
   & R-INLA (EB) & 13.9 & $118829.2^{\dagger}$ \\
  \hline
  MovieLens\_32m & Cholesky (GPBoost) & crashed &  \\
   & Krylov (GPBoost) & 230.4 & 40613057.0 \\
   & glmmTMB & crashed &  \\
   & lme4 & 139866.3 & 40613056.6 \\
   & MixedModels.jl & crashed &  \\
   & R-INLA (EB) & crashed &  \\
  \hline
  employee\_access & Cholesky (GPBoost) & 205.8 & 5494.4 \\
   & Krylov (GPBoost) & 30.1 & 5495.2 \\
   & glmmTMB & 228.4 & 5494.4 \\
   & *lme4 & 6139.6 & 5504.5 \\
   & *MixedModels.jl & 308881.0 & 5537.2 \\
   & R-INLA (EB) & 1534.5 & $5591.0^{\dagger}$ \\
  \hline
  upselling & Cholesky (GPBoost) & 569.2 & 11488.5 \\
   & Krylov (GPBoost) & 14.7 & 11488.5 \\
   & glmmTMB & NA &  \\
   & *lme4 & 241769.3 & 11632.8 \\
   & MixedModels.jl & NA &  \\
   & R-INLA (EB) & crashed &  \\
  \hline
\end{tabular}
\caption{Time for parameter estimation and negative log-marginal likelihood at the optimum for different real-world data sets and models. 'crashed' means that estimation aborted with an error or ran out of memory, and 'NA' that estimation terminated but returned missing values for the parameter estimates or the negative log-marginal likelihood at the optimum. A '*' in front of a method marks fits whose negative log-marginal likelihood at the optimum is more than 1 above the smallest one obtained for the data set, i.e. that did not converge properly. \rev{This convergence criterion applies only to the maximum marginal-likelihood-based implementations. The values marked with '$\dagger$' are the negative log-marginal likelihoods reported by \texttt{R-INLA}. They are obtained with the fixed effects integrated out and are thus restricted (REML-type) criteria evaluated with the INLA approximation, see Appendix~\ref{appendix:inla}.}}
\label{table:real_world_times}
\end{table}

\begin{table}[ht!]
\centering
\begin{tabular}{llcccc}
  \hline
data set & method & $\sigma^2$ & $\sigma_1^2$ & $\sigma_2^2$ & $\sigma_3^2$ \\
  \hline
  cars & Cholesky (GPBoost) & $5.24 \times 10^{-2}$ & $3.28 \times 10^{-1}$ & $2.04 \times 10^{-1}$ &  \\
   & Krylov (GPBoost) & $5.24 \times 10^{-2}$ & $3.28 \times 10^{-1}$ & $2.04 \times 10^{-1}$ &  \\
   & glmmTMB & $5.24 \times 10^{-2}$ & $3.28 \times 10^{-1}$ & $2.04 \times 10^{-1}$ &  \\
   & lme4 & $5.24 \times 10^{-2}$ & $3.28 \times 10^{-1}$ & $2.04 \times 10^{-1}$ &  \\
   & MixedModels.jl & $5.24 \times 10^{-2}$ & $3.28 \times 10^{-1}$ & $2.04 \times 10^{-1}$ &  \\
   & R-INLA (EB) & $5.24 \times 10^{-2}$ & $3.28 \times 10^{-1}$ & $2.04 \times 10^{-1}$ &  \\
  \hline
  building\_permits & Cholesky (GPBoost) & $1.31 \times 10^{0}$ & $9.10 \times 10^{-1}$ & $3.98 \times 10^{-2}$ & $3.97 \times 10^{-2}$ \\
   & Krylov (GPBoost) & $1.31 \times 10^{0}$ & $9.10 \times 10^{-1}$ & $3.88 \times 10^{-2}$ & $4.11 \times 10^{-2}$ \\
   & glmmTMB & $1.31 \times 10^{0}$ & $9.10 \times 10^{-1}$ & $3.51 \times 10^{-2}$ & $4.44 \times 10^{-2}$ \\
   & lme4 & $1.31 \times 10^{0}$ & $9.10 \times 10^{-1}$ & $7.95 \times 10^{-2}$ & $1.24 \times 10^{-7}$ \\
   & MixedModels.jl & crashed &  &  &  \\
   & R-INLA (EB) & $1.31 \times 10^{0}$ & $9.10 \times 10^{-1}$ & $4.00 \times 10^{-2}$ & $3.97 \times 10^{-2}$ \\
  \hline
  instEval & Cholesky (GPBoost) & $1.38 \times 10^{0}$ & $1.07 \times 10^{-1}$ & $2.57 \times 10^{-1}$ &  \\
   & Krylov (GPBoost) & $1.38 \times 10^{0}$ & $1.07 \times 10^{-1}$ & $2.58 \times 10^{-1}$ &  \\
   & glmmTMB & $1.38 \times 10^{0}$ & $1.07 \times 10^{-1}$ & $2.57 \times 10^{-1}$ &  \\
   & lme4 & $1.38 \times 10^{0}$ & $1.07 \times 10^{-1}$ & $2.57 \times 10^{-1}$ &  \\
   & MixedModels.jl & $1.38 \times 10^{0}$ & $1.07 \times 10^{-1}$ & $2.57 \times 10^{-1}$ &  \\
   & R-INLA (EB) & $1.38 \times 10^{0}$ & $1.07 \times 10^{-1}$ & $2.61 \times 10^{-1}$ &  \\
  \hline
  MovieLens\_32m & Cholesky (GPBoost) & crashed &  &  &  \\
   & Krylov (GPBoost) & $7.23 \times 10^{-1}$ & $1.93 \times 10^{-1}$ & $2.32 \times 10^{-1}$ &  \\
   & glmmTMB & crashed &  &  &  \\
   & lme4 & $7.23 \times 10^{-1}$ & $1.93 \times 10^{-1}$ & $2.32 \times 10^{-1}$ &  \\
   & MixedModels.jl & crashed &  &  &  \\
   & R-INLA (EB) & crashed &  &  &  \\
  \hline
  employee\_access & Cholesky (GPBoost) &  & $1.52 \times 10^{0}$ & $5.78 \times 10^{0}$ & $9.48 \times 10^{-6}$ \\
   & Krylov (GPBoost) &  & $1.51 \times 10^{0}$ & $5.76 \times 10^{0}$ & $1.54 \times 10^{-5}$ \\
   & glmmTMB &  & $1.52 \times 10^{0}$ & $5.78 \times 10^{0}$ & $3.14 \times 10^{-10}$ \\
   & *lme4 &  & $1.60 \times 10^{0}$ & $6.06 \times 10^{0}$ & $0.00 \times 10^{0}$ \\
   & *MixedModels.jl &  & $8.90 \times 10^{-1}$ & $3.23 \times 10^{0}$ & $4.97 \times 10^{-4}$ \\
   & R-INLA (EB) &  & $8.40 \times 10^{-1}$ & $3.08 \times 10^{0}$ & $1.31 \times 10^{-3}$ \\
  \hline
  upselling & Cholesky (GPBoost) &  & $3.56 \times 10^{-1}$ & $2.59 \times 10^{-1}$ & $9.86 \times 10^{-3}$ \\
   & Krylov (GPBoost) &  & $3.62 \times 10^{-1}$ & $2.60 \times 10^{-1}$ & $8.90 \times 10^{-3}$ \\
   & glmmTMB & NA &  &  &  \\
   & *lme4 &  & $4.36 \times 10^{-1}$ & $2.59 \times 10^{-1}$ & $1.17 \times 10^{-2}$ \\
   & MixedModels.jl & NA &  &  &  \\
   & R-INLA (EB) & crashed &  &  &  \\
  \hline
\end{tabular}
\caption{Estimates for $\sigma^2$, $\sigma_1^2$, $\sigma_2^2$, and $\sigma_3^2$ for different real-world data sets and models. For reasons of space, we do not report estimates for other variance parameters. 'crashed' means that estimation aborted with an error or ran out of memory, and 'NA' that estimation terminated but returned missing values for the parameter estimates or the negative log-marginal likelihood at the optimum. A '*' in front of a method marks fits whose negative log-marginal likelihood at the optimum is more than 1 above the smallest one obtained for the data set, i.e. that did not converge properly. \rev{This convergence criterion applies only to the maximum likelihood-based implementations since the objective function reported by \texttt{R-INLA} is not on the same scale, see Table \ref{table:real_world_times}.}}
\label{table:real_world_variances}
\end{table}

\begin{table}[ht!]
\centering
\begin{tabular}{llcccc}
  \hline
data set & method & $\beta_0$ & $\beta_1$ & $\beta_2$ & $\beta_3$ \\
  \hline
  cars & Cholesky (GPBoost) & $1.02 \times 10^{1}$ & $3.50 \times 10^{-1}$ & $-1.77 \times 10^{-1}$ & $-2.64 \times 10^{-3}$ \\
   & Krylov (GPBoost) & $1.02 \times 10^{1}$ & $3.50 \times 10^{-1}$ & $-1.77 \times 10^{-1}$ & $-2.64 \times 10^{-3}$ \\
   & glmmTMB & $1.02 \times 10^{1}$ & $3.50 \times 10^{-1}$ & $-1.77 \times 10^{-1}$ & $-2.64 \times 10^{-3}$ \\
   & lme4 & $1.02 \times 10^{1}$ & $3.50 \times 10^{-1}$ & $-1.77 \times 10^{-1}$ & $-2.64 \times 10^{-3}$ \\
   & MixedModels.jl & $1.02 \times 10^{1}$ & $3.50 \times 10^{-1}$ & $-1.77 \times 10^{-1}$ & $-2.64 \times 10^{-3}$ \\
   & R-INLA (EB) & $1.02 \times 10^{1}$ & $3.50 \times 10^{-1}$ & $-1.77 \times 10^{-1}$ & $-2.64 \times 10^{-3}$ \\
  \hline
  building\_permits & Cholesky (GPBoost) & $3.62 \times 10^{0}$ & $5.44 \times 10^{-4}$ & $-1.36 \times 10^{-5}$ & $2.01 \times 10^{-3}$ \\
   & Krylov (GPBoost) & $3.62 \times 10^{0}$ & $5.44 \times 10^{-4}$ & $-1.36 \times 10^{-5}$ & $2.04 \times 10^{-3}$ \\
   & glmmTMB & $3.62 \times 10^{0}$ & $5.44 \times 10^{-4}$ & $-1.36 \times 10^{-5}$ & $2.01 \times 10^{-3}$ \\
   & lme4 & $3.62 \times 10^{0}$ & $5.44 \times 10^{-4}$ & $-1.36 \times 10^{-5}$ & $2.01 \times 10^{-3}$ \\
   & MixedModels.jl & crashed &  &  &  \\
   & R-INLA (EB) & $3.62 \times 10^{0}$ & $5.44 \times 10^{-4}$ & $-1.36 \times 10^{-5}$ & $2.03 \times 10^{-3}$ \\
  \hline
  instEval & Cholesky (GPBoost) & $3.31 \times 10^{0}$ & $5.21 \times 10^{-2}$ & $7.23 \times 10^{-2}$ & $1.37 \times 10^{-1}$ \\
   & Krylov (GPBoost) & $3.31 \times 10^{0}$ & $5.21 \times 10^{-2}$ & $7.24 \times 10^{-2}$ & $1.37 \times 10^{-1}$ \\
   & glmmTMB & $3.31 \times 10^{0}$ & $5.21 \times 10^{-2}$ & $7.23 \times 10^{-2}$ & $1.37 \times 10^{-1}$ \\
   & lme4 & $3.31 \times 10^{0}$ & $5.21 \times 10^{-2}$ & $7.23 \times 10^{-2}$ & $1.37 \times 10^{-1}$ \\
   & MixedModels.jl & $3.31 \times 10^{0}$ & $5.21 \times 10^{-2}$ & $7.23 \times 10^{-2}$ & $1.37 \times 10^{-1}$ \\
   & R-INLA (EB) & $3.31 \times 10^{0}$ & $5.21 \times 10^{-2}$ & $7.23 \times 10^{-2}$ & $1.37 \times 10^{-1}$ \\
  \hline
  MovieLens\_32m & Cholesky (GPBoost) & crashed &  &  &  \\
   & Krylov (GPBoost) & $3.39 \times 10^{0}$ & $-7.29 \times 10^{-2}$ & $-4.93 \times 10^{-2}$ &  \\
   & glmmTMB & crashed &  &  &  \\
   & lme4 & $3.39 \times 10^{0}$ & $-7.30 \times 10^{-2}$ & $-4.93 \times 10^{-2}$ &  \\
   & MixedModels.jl & crashed &  &  &  \\
   & R-INLA (EB) & crashed &  &  &  \\
  \hline
  employee\_access & Cholesky (GPBoost) & $7.21 \times 10^{0}$ &  &  &  \\
   & Krylov (GPBoost) & $7.24 \times 10^{0}$ &  &  &  \\
   & glmmTMB & $7.19 \times 10^{0}$ &  &  &  \\
   & *lme4 & $6.92 \times 10^{0}$ &  &  &  \\
   & *MixedModels.jl & $5.05 \times 10^{0}$ &  &  &  \\
   & R-INLA (EB) & $3.95 \times 10^{0}$ &  &  &  \\
  \hline
  upselling & Cholesky (GPBoost) & $-1.74 \times 10^{0}$ & $2.32 \times 10^{-5}$ & $9.42 \times 10^{-3}$ & $3.86 \times 10^{-6}$ \\
   & Krylov (GPBoost) & $-1.75 \times 10^{0}$ & $2.31 \times 10^{-5}$ & $9.25 \times 10^{-3}$ & $3.92 \times 10^{-6}$ \\
   & glmmTMB & NA &  &  &  \\
   & *lme4 & $-2.15 \times 10^{0}$ & $2.85 \times 10^{-5}$ & $4.10 \times 10^{-2}$ & $1.91 \times 10^{-5}$ \\
   & MixedModels.jl & NA &  &  &  \\
   & R-INLA (EB) & crashed &  &  &  \\
  \hline
\end{tabular}
\caption{Estimates for $\beta_0$, $\beta_1$, $\beta_2$, and $\beta_3$ for different real-world data sets and models. For reasons of space, we do not report estimates for other coefficients. 'crashed' means that estimation aborted with an error or ran out of memory, and 'NA' that estimation terminated but returned missing values for the parameter estimates or the negative log-marginal likelihood at the optimum. A '*' in front of a method marks fits whose negative log-marginal likelihood at the optimum is more than 1 above the smallest one obtained for the data set, i.e. that did not converge properly. \rev{This convergence criterion applies only to the maximum likelihood-based implementations since the objective function reported by \texttt{R-INLA} is not on the same scale, see Table \ref{table:real_world_times}.}}
\label{table:real_world_coefficients}
\end{table}

\clearpage
\renewcommand{\refname}{Appendix References}
\putbib[bib_Iterative_GMMs] % prints only citations from this bibunit
\end{bibunit}

\end{document}